\documentclass[11pt]{article}
\usepackage{latexsym}
\usepackage{amssymb}
\usepackage{amsmath}
\usepackage{amsthm}
\newtheorem{thm}{Theorem}[section]
\newtheorem{lem}[thm]{Lemma}

\newtheorem{rem}[thm]{Remark}
\newtheorem{defin}[thm]{Definition}
\newtheorem{exm}[thm]{Example}
\newtheorem{pro}[thm]{Proposition}
\newtheorem{assumption}[thm]{Assumption}

\newcommand{\E}{{\mathcal E}}
\newcommand{\RR}{{\mathbb R}}
\newcommand{\CC}{{\mathbb C}}
\newcommand{\NN}{{\mathbb N}}

\newcommand{\FC}{{\mathfrak FC}^{\infty}_b}
\newcommand{\ep}{\varepsilon}

\newcommand{\tr}{{\rm tr}}
\newcommand{\la}{\lambda}

\newcommand{\ai}{\sqrt{-1}}

\newcommand{\FWU}{{\mathcal F}^W_{U}}

\newcommand{\rhoWu}{\rho^W_{u}}
\newcommand{\TNl}{T_{N,l}}
\newcommand{\dUAg}{d_U^{Ag}}
\newcommand{\uZ}{u_{{\cal Z}}}

\numberwithin{equation}{section}
\begin{document}
\newcounter{aaa}
\newcounter{bbb}
\newcounter{ccc}
\newcounter{ddd}
\newcounter{eee}
\newcounter{xxx}
\newcounter{xvi}
\newcounter{x}
\setcounter{aaa}{1}
\setcounter{bbb}{2}
\setcounter{ccc}{3}
\setcounter{ddd}{4}
\setcounter{eee}{32}
\setcounter{xxx}{10}
\setcounter{xvi}{16}
\setcounter{x}{38}
\title
{Tunneling for spatially cut-off $P(\phi)_2$-Hamiltonians}
\author{Shigeki Aida\footnote{
This research was partially supported by
Grant-in-Aid for Scientific Research (A) No.~21244009
and Grant-in-Aid for Scientific Research (B) No.~24340023.}\\
Mathematical Institute\\
Tohoku University,
Sendai, 980-8578, JAPAN\\
e-mail: aida@math.tohoku.ac.jp}
\date{}
\maketitle
\begin{abstract}
We study the asymptotic behavior of low-lying eigenvalues of
spatially cut-off $P(\phi)_2$-Hamiltonian under semi-classical limit.
The corresponding classical equation of
the $P(\phi)_2$-field is a
nonlinear Klein-Gordon equation
which is an infinite dimensional Newton's equation.
We determine the semi-classical
limit of the lowest eigenvalue of the spatially cut-off 
$P(\phi)_2$-Hamiltonian in terms of the Hessian of the 
potential function of the Klein-Gordon equation.
Moreover, we prove that the gap of the lowest two eigenvalues
goes to $0$ exponentially fast under semi-classical limit
when the potential function is double well type.
In fact, we prove that the exponential decay rate
is greater than or equal to
the Agmon distance between two zero points
of the symmetric double well potential function.
The Agmon distance is a Riemannian distance on
the Sobolev space $H^{1/2}(\RR)$
defined by a Riemannian metric
which is formally conformal to
$L^2$-metric.
Also we study basic properties of
the Agmon distance and instanton.
\end{abstract}

\section{Introduction}

Spatially cut-off $P(\phi)_2$-Hamiltonian is used to construct 
non-trivial quantum scalar fields in space time dimension two
and studied from various points of view, {\it e.g.},
\cite{arai2, gerard, nelson, eachus, eckmann, rosen,
segal, simon1, simon2, sh}.
Hamiltonians in quantum systems contain a small physical parameter,
Planck constant $\hbar$, and it is called semi-classical analysis
to study properties of quantum systems under
$\hbar\to 0$.
There are many studies
on spectral properties
of Schr\"odinger operators under semi-classical limit.
See, {\it e.g.}, \cite{d-s, he2, hes1, hes2, simon3, simon4}.
Classical mechanics corresponding to
$P(\phi)_2$-quantum field is given by
a nonlinear Klein-Gordon equation.
Therefore it is natural to conjecture that
the low-lying spectrum of the spatially cut-off $P(\phi)_2$-Hamiltonian
under semi-classical limit 
is related with the potential function 
$U$ of the classical dynamics given by the Klein-Gordon equation.
One of the aim of this paper is to study the asymptotic 
behavior of the lowest eigenvalue of 
spatially cut-off $P(\phi)_2$-Hamiltonian
under semi-classical limit.
We already studied the same problem
for $P(\phi)_2$-Hamiltonian in the case where the space is a finite
interval in \cite{aida2}.
In that case, one particle Hamiltonian has compact resolvent and 
it makes analysis in \cite{aida2} simple.
However, in the case of spatially cut-off $P(\phi)_2$-Hamiltonian, 
such a property does not hold and it cause 
some difficulties.
In addition to the asymptotics of the lowest 
eigenvalue,
we study the semi-classical tunneling of the spatially cut-off
$P(\phi)_2$-Hamiltonian with symmetric double well 
potential function.
That is, we show that the gap 
between lowest two eigenvalues is 
exponentially small
under semi-classical limit.
It is still an open problem to obtain the precise asymptotics of the gap of 
spectrum.

The organization of this paper is as follows.
In Section 2, we state our main results.
Let $\mu$ be the Gaussian measure on 
the space of tempered distributions
${\mathcal S}'(\RR)$ whose covariance operator
is $(m^2-\Delta)^{-1/2}$ on $L^2(\RR)$,
where $m$ is a positive number and
$\Delta$ is the Laplace operator.
Let $A=(m^2-\Delta)^{1/4}$ be the self-adjoint operator
on $H(=H^{1/2}(\RR))$.
One can define a Dirichlet form on $L^2({\cal S}'(\RR),\mu)$
using $A$ as a coefficient operator.
Let
$-L_A$ be the non-negative generator of the
Dirichlet form.
This operator is so-called a free Hamiltonian
and naturally unitarily equivalent to
the second quantization operator $d\Gamma((m^2-\Delta)^{1/2})$,
where $(m^2-\Delta)^{1/2}$ is a self-adjoint operator
on $H^{-1/2}(\RR)$.
Let $\la=\frac{1}{\hbar}$, that is $\la$ is a large positive parameter.
We consider spatially cut-off $P(\phi)_2$-Hamiltonian
$-L_A+V_{\la}$, where
$V_{\la}(w)=\la V(w/\sqrt{\la})$ is
an interaction potential function.
The potential function $V(w)$ is defined by 
$V(w)=\int_{\RR}:P(w(x)): g(x)dx$, where
$P(x)=\sum_{k=0}^{2M}a_kx^k$ is a polynomial
with $a_{2M}>0$.
Also $:P(w(x)):$ stands for the Wick polynomial and $g$ is a
non-negative smooth function with compact support.
The operator $-L_A+V_{\la}$ is formally unitarily
equivalent to an infinite dimensional
Schr\"odinger operator on
$L^2(\RR)$ with the fictitious infinite dimensional
Lebesgue measure and its formal
potential function $U$ is given by
$U(h)=\frac{1}{4}\|Ah\|_H^2+V(h)$, where 
$V(h)=\int_{\RR}P(h(x))g(x)dx$ and
$h\in H^1(\RR)$.
This potential function $U$ is the potential function for the 
classical dynamics which is defined by the corresponding
nonlinear Klein-Gordon equation.
In the first main theorem (Theorem~\ref{main theorem 0}),
we determine the semi-classical limit of the lowest eigenvalue 
$E_1(\la)$ of
$-L_A+V_{\la}$ under the assumptions that 
$U$ is non-negative,
has a finitely many zero points
and the Hessians of $U$ at zero points are nondegenerate.
When $U$ is a symmetric double well potential function,
one may expect that the gap between second lowest eigenvalue 
$E_2(\la)$ and the lowest eigenvalue $E_1(\la)$ of $-L_A+V_{\la}$
is exponentially small under the limit $\la\to\infty$.
In the case of Schr\"odinger operators,
the exponential decay rate is
equal to the Agmon distance between two zero points
of the potential function.
Second main result (Theorem~\ref{tunneling})
is concerned with this estimate (tunneling estimate).
Actually, we prove that
the exponential decay rate
is greater than or equal to 
the Agmon distance
$d^{Ag}_U(h_0,-h_0)$ between
zero points $h_0,-h_0$ of 
the symmetric double well potential function $U$.
What is the infinite dimensional analogue
of the Agmon distance in this case ?
Formally, it is the Riemannian distance on $H^{1/2}(\RR)$
determined by a Riemannian metric
$U(w)ds^2$ which is \lq\lq conformal'' to
the $L^2$-metric $ds^2$.
In this section, we define the Agmon distance on $H^1(\RR)$.
We prove that the distance can be extended to a
continuous distance function on $H^{1/2}(\RR)$ in Section~7.
Also, the Agmon distance is related with an instanton
which is a minimizing path of the Euclidean 
(imaginary time) action integral.
We summarize basic properties of Agmon distance and instanton
in our model in Section~7.
For instance, we prove existence of a minimal geodesic
between zero points of $U$ and an instanton solution.
In this paper, we do not use any properties of instanton.
However, we think the subjects are interesting by themselves.

In Section 3, first, we recall the definition of the 
spatially cut-off $P(\phi)_2$-Hamiltonian
based on Dirichlet forms on $L^2({\mathcal S}'(\RR),\mu)$.
Next, we prepare necessary tools
for the proof of main theorems. 
Actually we need a stronger theorem 
(Theorem~\ref{main theorem 1}) than
Theorem~\ref{main theorem 0} to prove
tunneling estimate.
In the proof of the lower
bound of the limit of $E_1(\la)$,
we use large deviation results of
Wiener chaos and a lower bound estimate of the generator of
hyperbounded semi-groups.
To apply these results,
we need to approximate the operator
$A$ by operators of the form
$\sqrt{m}I+T$,
where
$I$ is the identity operator and $T$ is a trace class operator
on $H$.
Comparing the case where the space is a finite interval in \cite{aida2},
this step is not so simple, since
the operator $A$ has a continuous spectrum.
These approximations
are constructed by using the Fourier transform.
After these preliminaries, we prove
Theorem~\ref{main theorem 1} in Section 4.
We give the proofs of some lemmas 
in Section 7.
In Section 5, 
we introduce an approximate Agmon distance
$d^W_U(h_0,-h_0)$ between $h_0$
and $-h_0$ and
prove
that the exponential
decay rate of the gap of the spectrum
is greater than or equal to $d^W_U(h_0,-h_0)$.
This kind of decay estimate follows from
the estimate for the ground state function
(ground state measure).
It is important that the Agmon distance function belongs to
an $H^1$-Sobolev space in the classical proof.
However, the Agmon distance $d^{Ag}_U$ 
is a distance function
on $H^{1/2}(\RR)$ and it seems that
the distance function cannot be extended to a function
on $W$ on which $H^1$-Sobolev space is defined.
To overcome this difficulty, we 
introduce a family of 
non-negative bounded Lipschitz continuous functions $u$ on $W$
which approximate $U$
and using $u$ we define a distance function
$\rho^W_u(O,\cdot)$ from an open subset $O$.
This $\rho^W_u(O,\cdot)$ does
belong to $H^1$-Sobolev space.
Using $\rho^W_u(O,\cdot)$ and
Theorem~\ref{main theorem 1}, we can give an exponential decay estimate
for the ground state measure
in a similar way to finite dimensional cases.
By optimizing $u$ and so $\rho^W_u$, we define 
$d^W_U(h_0,-h_0)$.
In the last step,
we prove that $d^{Ag}_U(h_0,-h_0)=d^W_U(h_0,-h_0)$
which implies the second main theorem.
In Section 6, we give an example.

\section{Statement of main results}
\label{statement}

Let $L^2(\RR)=L^2(\RR\to \RR, dx)$
and $L^2(\RR)_{\CC}$ be the complexification of
$L^2(\RR)$.
Let $\Delta=\frac{d^2}{dx^2}$ be the Laplace operator
on $L^2(\RR)_{\CC}$ with the domain ${\rm D}(\Delta)$.
The subspace $L^2(\RR)$ is invariant under the operator
$\Delta$ and 
$\Delta|_{D(\Delta)\cap L^2(\RR)}$ is also a self-adjoint operator
in $L^2(\RR)$.
We denote it also by $\Delta$.
Let $m>0$ and we set $\tilde{A}=(m^2-\Delta)^{1/4}$.
Let
$H^s(\RR)=H^s(\RR\to \RR)$~$(={\rm D}(\tilde{A}^{2s}))$~$(s\ge 0)$ be the
Hilbert space with the norm $\|~\|_{H^s}$ 
defined by 
\begin{equation}
\|\varphi\|_{H^s}=\|\tilde{A}^{2s}\varphi\|_{L^2}
\end{equation}
where, we identify $H^s(\RR)$ as a subset of $L^2(\RR)$.
We may denote $H^s(\RR)$ by $H^s$ simply.
Let $H^s(\RR)_{\CC}$
be the complexification of $H^s(\RR)$.
There exists a unique Gaussian measure $\mu$
on ${\mathcal S}'(\RR)$ such that
\begin{equation}
\int_{{\mathcal S}'(\RR)}\exp\left(
\sqrt{-1}\langle
\varphi,w\rangle\right)
d\mu(w)=
\exp\left(-\frac{1}{2}(\varphi,\tilde{A}^{-2}\varphi)_{L^2(\RR)}\right),
\end{equation}
where $\langle \varphi,w\rangle$ is a natural coupling
of $\varphi\in {\cal S}(\RR)$ and $w\in {\cal S}^{\prime}(\RR)$.
The Hilbert space $H^{1/2}(\RR)$ is nothing but the Cameron-Martin subspace of 
$\mu$.
Below, we write $H=H^{1/2}(\RR)$.
Let us define a self-adjoint operator $A$ on $H$ by setting
${\rm D}(A)=H^1$ and $Af=(m^2-\Delta)^{1/4}f$.
Let $\Phi=\tilde{A}^{-1} : L^2\to H$.
Then $\Phi$ is a unitary operator and
$A$ and $\tilde{A}$ are unitarily equivalent to
each other by this unitary map.
That is $A=\Phi\circ \tilde{A}\circ \Phi^{-1}$
holds. 
Let us consider a second quantization
operator $d\Gamma((m^2-\Delta)^{1/2})$,
where $(m^2-\Delta)^{1/2}$ is a self-adjoint on
$H^{-1/2}(\RR)$ with the domain ${\rm D}((m^2-\Delta)^{1/2})=H^{1/2}(\RR)$.
There exists a unitarily equivalent operator $-L_A$ on 
$L^2({\mathcal S}'(\RR),\mu)$ to $d\Gamma((m^2-\Delta)^{1/2})$.
These operators are free Hamiltonians and we use the version
$-L_A$ in this paper.
We give the precise definition of $-L_A$
based on Dirichlet forms
in the next section.
Spatially cut-off $P(\phi)_2$-Hamiltonian 
is a perturbation of $-L_A$ 
by an interaction potential function.
Now we define the interaction potential in $L^2({\mathcal S}'(\RR),\mu)$.
We refer the reader for basic results of spatially cut-off $P(\phi)_2$-
Hamiltonian to \cite{simon2, sh}.

\begin{defin}\label{Pphi2 hamiltonian}
For $w\in {\cal S}^{\prime}(\RR)$, define
$w_n(x)=\langle p_n(x-\cdot),w\rangle$,
where 
$$
p_n(x)=\left(\frac{n}{4\pi}\right)^{1/2}\exp\left(-\frac{nx^2}{4}\right).
$$
Let $\la>0$.

\noindent
$(1)$~
Let us define
\begin{equation}
:\left(\frac{w_{n}(x)}{\sqrt{\la}}\right)^k:
=\left(\frac{w_{n}(x)}{\sqrt{\la}}\right)^k+
\sum_{j=1}^{[k/2]}c_{k,j}
\left(\frac{w_{n}(x)}{\sqrt{\la}}\right)^
{k-2j}\left(\frac{c_{n}}{\sqrt{\la}}\right)^{2j},
\end{equation}
where
$c_{k,j}=\left(-\frac{1}{2}\right)^j
\frac{k!}{j!(k-2j)!}$
and $c_n^2=\int_{{\mathcal S}'(\RR)}w_n(x)^2d\mu(w)=\frac{1}{2\pi}
\int_0^{\infty}\frac{e^{-m^2t}}{\sqrt{t(t+\frac{2}{n})}}dt$.

\noindent
$(2)$~
Let $P(x)=\sum_{k=0}^{2M}a_kx^{k}$ be 
a polynomial function with
$M\ge 2$ and $a_{2M}>0$.
Let $g$ be a non-negative $C^{\infty}$ function on $\RR$ with compact support.
Define
\begin{eqnarray}
\int_{\RR}:P\left(\frac{w(x)}{\sqrt{\la}}\right): g(x)dx
&=&
\sum_{k=0}^{2M}
a_k\int_{\RR}
:\left(
\frac{w(x)}{\sqrt{\la}}\right)^k:g(x)dx\nonumber\\
&=&\lim_{n\to\infty}
\sum_{k=0}^{2M}a_k\int_{\RR}
:\left(\frac{w_n(x)}{\sqrt{\la}}\right)^k:g(x)dx
\end{eqnarray}
as a limit in $L^2({\mathcal S}^\prime(\RR),d\mu)$.
We define
\begin{eqnarray}
:V\left(\frac{w}{\sqrt{\la}}\right):&=&
\int_{\RR}:P\left(\frac{w(x)}{\sqrt{\la}}\right): g(x)dx
\end{eqnarray}
and
\begin{equation}
V_{\la}(w)=\la :V\left(\frac{w}{\sqrt{\la}}\right):.
\end{equation}

\noindent
$(3)$~
Let $\FC$ be the set of all functions
of the form 
$$
f\left(\langle
\varphi_1,w\rangle,
\ldots,
\langle
\varphi_n,w\rangle\right),
$$
where $f$ is a smooth bounded function on $\RR^n$
whose all derivatives are also
bounded.
It is known that $(-L_A+V_{\la},\FC)$
is essentially self-adjoint.
We use the same notation
$-L_A+V_{\la}$ for
the self-adjoint extension
which is called
a spatially cut-off $P(\phi)_2$-Hamiltonian.
Also it is known that the operator $-L_A+V_{\la}$ is bounded from below.
Let $E_1(\la)=\inf\sigma(-L_A+V_{\la})$,
where $\sigma(-L_A+V_{\la})$ denotes the spectral set of
$-L_A+V_{\la}$.
\end{defin}
Formally,
$-L_A+V_{\la}$ is unitarily equivalent to
the infinite dimensional Schr\"odinger operator
on $L^2(L^2(\RR),dw)$:
\begin{equation}
-\Delta_{L^2(\RR)}
+\la :U(w/\sqrt{\la}):-\frac{1}{2}\tr (m^2-\Delta)^{1/2},
\label{formal representation}
\end{equation}
where $dw$ is an infinite dimensional
Lebesgue measure,
\begin{eqnarray*}
:U(w):&=&
\frac{1}{4}\int_{\RR}w'(x)^2dx
+\int_{\RR}\left(
\frac{m^2}{4}w(x)^2+:P(w(x)):g(x)
\right)dx
\end{eqnarray*}
and $\Delta_{L^2(\RR)}$ denotes the \lq\lq Laplacian''on
$L^2(\RR,dx)$.
That is the $P(\phi)_2$-Hamiltonian is related with 
the quantization of the nonlinear Klein-Gordon equation:
\begin{equation}
\frac{\partial^2 u}{\partial t^2}(t,x)=-2(\nabla U)(u(t,x)), \label{nlkg}
\end{equation}
where $\nabla$ denotes the $L^2$-gradient.
Hence, it is natural to put assumptions on $U$
to study asymptotic behavior of
spectrum of $-L_A+V_{\la}$ under semi-classical limit
$\la\to\infty$.
Here let us recall standard assumptions on potential function
$U$ on $\RR^d$
in the case of Schr\"odinger operators
$-H_{\la,U}=-\Delta+\la U(x/\sqrt{\la})$ in $L^2(\RR^d,dx)$.
Under the assumptions
\begin{enumerate}
\item[(H1)]$U$ is sufficiently smooth, $\min U=0$ 
and the zero point is a finite set,
\item[(H2)]The Hessians of $U$ at zero points are strictly positive,
\item[(H3)]$\liminf_{|x|\to\infty} U(x)>0$.
\end{enumerate}
It is well-known (\cite{d-s, hes1, hes2, simon4, heni})
that $\lim_{\la\to\infty}\inf\sigma(-H_{\la,U})$ is determined by
the spectral bottom of
the harmonic oscillators which are obtained by replacing $U$
by quadratic approximate functions near zero points of $U$.
By the analogy, we consider the
following assumptions on our potential functions.

\begin{assumption}\label{assumption on U}
Let $P$ be the polynomial in 
Definition~$\ref{Pphi2 hamiltonian}$
and $U$ be the function on $H^1$ which is given by
\begin{equation}
U(h)=\frac{1}{4}\int_{\RR}h'(x)^2dx+\int_{\RR}\left(
\frac{m^2}{4}h(x)^2+P(h(x))g(x)
\right)dx\qquad \mbox{for all $h\in H^1$}.
\label{U}
\end{equation}

\noindent
{\rm (A1)}~The function $U$ is non-negative and 
the zero point set
\begin{equation}
{\cal Z}:=
\{h\in H^1~|~U(h)=0\}=\{h_1,\ldots,h_{n_0}\}
\end{equation}
is a finite set.

\noindent
{\rm (A2)}~For all $1\le i\le n_0$,
the Hessian 
$\nabla^2U(h_i)$ is non-degenerate.
That is, there exists $\delta_i>0$ for each $i$ such that
\begin{eqnarray}
\nabla^2U(h_i)(h,h)
&:=&
\frac{1}{2}\int_{\RR}h'(x)^2dx+\int_{\RR}
\left(\frac{m^2}{2}h(x)^2+P''(h_i(x))g(x)h(x)^2\right)dx\nonumber\\
&\ge&\delta_i\|h\|_{L^2(\RR)}^2
\qquad \mbox{for all $h\in H^1(\RR)$}.
\end{eqnarray}
\end{assumption}

Clearly, the nondegeneracy of the Hessian is equivalent to
the strictly positivity of the
Schr\"odinger operator $m^2-\Delta+4v_i$,
where
\begin{equation}
v_i(x)=\frac{1}{2}P''(h_i(x))g(x).
\end{equation}
Using the Taylor expansion of $U$ at $h_i$,
we can obtain the approximate operator
$-L_A+Q_{v_i}$.
Hence it is natural to expect the following
theorem which is our first main result.

\begin{thm}\label{main theorem 0}
~Assume that {\rm (A1)} and {\rm (A2)} hold.
Let $E_1(\la)=\inf\sigma(-L_A+V_{\la})$.
Then
\begin{equation}
\lim_{\la\to\infty}
E_1(\la)=\min_{1\le i\le n_0}E_i,\label{main 01}
\end{equation}
where 
\begin{equation}
E_i=\inf\sigma(-L_A+Q_{v_i})
\end{equation}
and $Q_{v_i}$ is given by 
\begin{equation}
Q_{v_i}(w)=\int_{\RR}:w(x)^2:v_i(x)dx.
\end{equation}
\end{thm}

In the theorem above, the function $v_i$
is a $C^{\infty}$ function but may take negative values.
However, the Wick polynomial $Q_{v_i}$ can be defined 
in the same way as in Definition~\ref{Pphi2 hamiltonian}.
Actually, $E_i>-\infty$ for all $i$ and
we can give the explicit form of the number $E_i$
using the Hilbert-Schmidt norm of a certain operator.
See Lemma~\ref{quadratic}.
By analogy with Schr\"odinger operators in $L^2(\RR^d)$,
one may expect that there exist eigenvalues near the values
$E_1,\ldots,E_{n_0}$ for large $\la$.
By the Simon and Hoegh-Krohn's result, if
$E_j-\min_iE_i<m$, then 
there exist eigenvalues near $E_j$ for large $\la$.
However, if $E_j-\min_iE_i\ge m$, then embedded
eigenvalues in the essential spectrum of $-L_A+V_{\la}$
may appear.
Of course, there are some constraints on the numbers $E_1,\ldots,E_{n_0}$
because they are related with some variational problems.
At the moment, the author has no answer to this problem.
Simon~\cite{simon1} gave examples of 
embedded eigenvalues in the essential spectrum 
of spatially cut-off $P(\phi)_2$-Hamiltonian in a different 
situation.

Next we state our second main result.
Let 
$
E_2(\la)=\inf\left\{
\sigma(-L_A+V_{\la})\setminus\{E_1(\la)\}\right\}.
$
In the second main theorem, we prove that
$E_2(\la)-E_1(\la)$ is exponentially small when
$U$ is a symmetric double well type potential function
under semi-classical limit.
This kind of estimate is related with tunneling in 
quantum mechanical system.
We refer the reader to \cite{d-s, hes1, hes2, simon4, heni} 
for tunneling estimates 
in the case of Schr\"odinger operators.
See \cite{dobrokhotov} also for large dimension cases.
To state our estimate, we introduce infinite dimensional
analogue of Agmon distance in quantum mechanics.

\begin{defin}\label{agmon distance}
Let $0<T<\infty$ and $h,k\in H^1(\RR)$.
Let $AC_{T,h,k}(H^1(\RR))$ be the all
absolutely continuous functions $c : [0,T]\to H^1(\RR)$
satisfying $c(0)=h, c(T)=k$.
We omit the subscript $T$ when $T=1$
and omit denoting $h,k$ if there are no constraint.
Let $U$ be the potential function in
$(\ref{U})$.
Assume $U$ is non-negative.
We define the Agmon distance between $h, k$
by
\begin{eqnarray}
\dUAg(h,k)&=&
\inf\left\{\ell_U(c)~|~
c\in AC_{T,h,k}(H^1(\RR))
\right\},
\end{eqnarray}
where
\begin{equation}
\ell_U(c)=
\int_{0}^T\sqrt{U(c(t))}\|c'(t)\|_{L^2}dt.
\end{equation}
\end{defin}

In this paper, we consider separable Hilbert space valued functions
defined on intervals of $\RR$.
In that case, the notion of absolute continuity of the functions
is equivalent to that the functions are equal to 
indefinite integrals of 
Bochner integrable functions and the same property 
(a.e.-differentiability, etc)
as finite dimensions hold.
See \cite{diestel and uhl}.
Note that the definition of Agmon distance
above does not depend on $T$.
We give another definition of the Agmon distance
in Section~7
so that the distance function can be extended to
a continuous distance function on $H^{1/2}(\RR)$.
Next we introduce symmetric double well
type potential functions.

\begin{assumption}\label{symmetric polynomial}
Let $P=P(x)$ be the polynomial function
in the definition of $U$.
We consider the following assumption.

\noindent
{\rm (A3)}~For all $x$, $P(x)=P(-x)$.
and ${\cal Z}=\{h_0, -h_0\}$, where $h_0\ne 0$.
\end{assumption}

The following is our second main theorem.

\begin{thm}\label{tunneling}~
Assume that $U$ satisfies {\rm (A1),(A2),(A3)}.
Then it holds that
\begin{equation}
\limsup_{\lambda\to\infty}\frac{\log\left(E_2(\lambda)-E_1(\lambda)\right)}
{\lambda}\le -d^{Ag}_U(h_0,-h_0).\label{tunneling upper bound}
\end{equation}
\end{thm}

In \cite{aida2}, we determine the semi-classical limit
of the lowest eigenvalue of $P(\phi)_2$-Hamiltonian
in the case where the space is a finite interval.
By a similar kind of proof, we can
prove that a similar estimate to Theorem~\ref{tunneling}
holds true in such a case too.
Finally,
we make remarks on researches on semi-classical limit
of $-L_A+V_{\la}$.
Arai~\cite{arai2} studied a semi-classical limit of partition functions
for $P(\phi)_2$-Hamiltonians in the case where 
the space is a finite interval.
The semi-classical properties of spectrum of Schr\"odinger operators 
in large dimension are studied
in \cite{he1,he2, sj1, sj2, matte, dobrokhotov}.

\section{Preliminaries}

The probability measure $\mu$ whose covariance operator
$(m^2-\Delta)^{-1/2}$ on $L^2(\RR)$ exists on ${\mathcal S}'(\RR)$.
However, we can choose a proper subset $W$ of ${\mathcal S}'(\RR)$
on which $\mu$ exists.
Let $S$ be a non-negative
self-adjoint trace class operator on $H$
such that $Sh\ne 0$ for any $h\ne 0$.
Let $\|h\|_S=(Sh,h)_H$.
Let $H_S$ be the completion of $H$ with respect to
the Hilbert norm $\|~\|_S$.
Then $\mu(H_S)=1$.
Of course, there are no
significance in a particular choice of $H_S$.
However, the following choice $H_{S_0}$
is useful in some estimate.
See Lemma~\ref{gagliard-nirenberg estimate}.
Let $-\Delta_H=1+x^2-\Delta$ be the
Schr\"odinger operator on $L^2(\RR)$.
Clearly $-\Delta_H^{-1}$ is a Hilbert-Schmidt operator
on $L^2$.
Let us consider a trace class self-adjoint operator on $H$:
$$
S_0=\tilde{A}^{-2}\left(-\Delta_H\right)^{-2}.
$$
Then $H_{S_0}$ can be identified with a subset
of ${\cal S}'(\RR)$ and
$$
\|w\|_{S_0}^2=
\int_{\RR}|(1+x^2-\Delta)^{-1}w(x)|^2dx.
$$
Throughout this paper, we set
$W=H_{S_0}$.
Now, we recall the definition of the free Hamiltonian.

\begin{defin}\label{Dirichlet form}
Let
$\E_A$ be the Dirichlet form defined by
\begin{equation}
\E_A(f,f)=\int_W\|ADf(w)\|_H^2d\mu(w), \quad f\in {\rm D}({\cal E}_A),
\label{Dirichlet form}
\end{equation}
where
\begin{eqnarray}
{\rm D}({\cal E}_A)=
\Bigl\{f\in {\rm D}({\cal E}_I)~|~
\mbox{$Df(w)\in {\rm D}(A)$~$\mu-a.s.~w$ and
$\int_W\|ADf(w)\|_H^2d\mu(w)<\infty$
}
\Bigr\}
\end{eqnarray}
and $D$ is an $H$-derivative and
${\cal E}_{I}$ stands for the Dirichlet form 
which is obtained by replacing $ADf(w)$ by $Df(w)$ in
$(\ref{Dirichlet form})$.
We denote the non-negative generator of $\E_A$ by $-L_A$
and write $D_Af(w)=ADf(w)$ for $f\in {\rm D}(\E_A)$.
\end{defin}

In the above definition, $I$ stands for the identity operator
on $H$ and the generator $-L_I$ of the Dirichlet form
$\E_I$ is the number operator (Ornstein-Uhlenbeck operator).
We refer the reader for $H$-derivative and analysis on 
(abstract) Wiener spaces
to \cite{bogachev, g1, kuo}.
Here is a remark on the derivative $D_A$.

\begin{rem}
Let $f$ be a smooth function on
$W$ in the sense of Fr\'echet.
Let $(\nabla f)(w)$ be the unique element
in $L^2(\RR)$ such that for any $\varphi\in L^2(\RR)$,
\begin{equation}
\lim_{\ep\to 0}\frac{f(w+\ep \varphi)-f(w)}{\ep}
=\left((\nabla f)(w),\varphi\right)_{L^2}.
\end{equation}
Then $\|D_Af(w)\|_H^2=\|\nabla f(w)\|_{L^2}^2$
holds.
Also by the analogy of finite dimensional cases,
the Riemannian metric on $H$ corresponding to the
Dirichlet form $\E_A$ is $L^2$-Riemannian metric.
\end{rem}

Also we note that
the potential function $U$ in Assumption~\ref{assumption on U}
can be rewritten in the following form:
\begin{equation}
U(h)=\frac{1}{4}\|Ah\|_H^2+V(h) \quad \mbox{for all $h\in {\rm D}(A)$},
\label{U on H}
\end{equation}
where
\begin{equation}
V(h)=\int_{\RR}P(h(x))g(x)dx.\label{V on H}
\end{equation}
The function $V=V(h)$ is well-defined on $H$
by the following lemma.
We refer the reader for
basic results of Sobolev spaces
$H^s$ to \cite{adams}.

\begin{lem}\label{sobolev}
Let $p\ge 2$ and $s>\frac{p-2}{2p}$.
Then there exists a constant $C_{p,s}$ such that
\begin{equation}
\|\varphi\|_{L^p}\le C_{p,s}\|\varphi\|_{H^s}.
\end{equation}
\end{lem}

The constant $C_{p,s}$ actually depends on
$m$ because our Sobolev spaces are defined by
$m^2-\Delta$.
Concerning the zero point function $h_i$ of $U$,
we have the following result.

\begin{lem}
The minimizer $h_i$ of $U$ belongs to $H^2(\RR)$
and satisfies the equation
\begin{equation}
(m^2-\Delta) h_i(x)+2P'(h_i(x))g(x)=0.
\end{equation}
\end{lem}

Let $v$ be a continuous function with compact
support.
We use the notation
\begin{equation}
Q_v(w)=\int_{\RR}:w(x)^2:v(x)dx.
\end{equation}
On the other hand, for any Hilbert-Schmidt operator $K$ on $H$, 
we can define a quadratic Wiener functional $:\langle Kw,w\rangle:$ 
as the limit
\begin{equation}
\lim_{n\to\infty}\left\{(P_nKP_nw,w)_H-\tr P_nKP_n\right\},
\end{equation}
where $\{P_n\}$ is a family of projection operators
onto finite dimensional subspaces
on $H$ such that ${\rm Im}\, P_n\subset {\rm Im}\, P_{n+1}$ for any $n$
and $\lim_{n\to\infty}P_n=I$ strongly.
When $K$ is a trace class operator, we denote 
the limit $\lim_{n\to\infty}(P_nKP_nw,w)_H$ by
$(Kw,w)_H$.
Now we recall another characterization of the Wick polynomial
$Q_v(w)=\int_{\RR}:w(x)^2:v(x)dx$ using the
corresponding Hilbert-Schmidt operator.
Recall that $\Phi=\tilde{A}^{-1} : L^2\to H$.

\begin{lem}\label{Kv}
Let $v$ be a $C^1$ function with compact support on $\RR$.
Let $M_v$ be the multiplication operator by $v$ in $L^2(\RR)$.
Let us define a bounded linear operator on $H$ by
\begin{equation}
K_v=\Phi\circ \left(\tilde{A}^{-1}M_v\tilde{A}^{-1}\right)\circ\Phi^{-1}.
\end{equation}

\noindent
$(1)$~It holds that $K_vh=\tilde{A}^{-2}M_v h$ for all $h\in H$.
The operator $\tilde{A}^{-1}M_v\tilde{A}^{-1}$ belongs to
Hilbert-Schmidt class.
Consequently, 
$AK_vA$ is a bounded linear operator and
$K_v$ is a Hilbert-Schmidt operator on $H$.

\noindent
$(2)$~It holds that
\begin{equation}
\int_{\RR}:w(x)^2:v(x)dx=:\langle K_vw,w\rangle_H:.
\end{equation}
\end{lem}

Since the function $v$ may be negative in the lemma below, 
we cannot apply the results in Definition~\ref{Pphi2 hamiltonian}
(3) directly to prove the lower boundedness of $-L_A+Q_v$.
However, it is not difficult to
show such a result because the ground 
state function is explicitly known.
We summarize the results.

\begin{lem}\label{quadratic}~
Let $v$ be a $C^1$ function with compact support.

\noindent
$(1)$~It holds that
\begin{equation}
\Phi^{-1}\left(A^4+4AK_vA\right)\Phi
=m^2-\Delta+4v.
\end{equation}
In particular, the strict positivity of $m^2-\Delta+4v$
on $L^2(\RR)$
is equivalent to the strict positivity of
$A^4+4AK_vA$.

\noindent
$(2)$~Assume that $m^2-\Delta+4v$ on $L^2(\RR)$
is strictly positive
and we define
$\tilde{A}_v=(m^2-\Delta+4v)^{1/4}$.
\begin{itemize}
\item[{\rm (i)}]~
$\tilde{A}_v^2-\tilde{A}^2$ is a Hilbert-Schmidt 
operator on $L^2(\RR)$.
\item[{\rm (ii)}]~
Let $A_v=(A^4+4AK_vA)^{1/4}$
and $T_v=A^{-1}(A_v^2-A^2)A^{-1}$.
Then $T_v$ is a Hilbert-Schmidt operator on
$H$ with $\inf\sigma(T_v)>-1$.
\item[{\rm (iii)}]~The densely defined 
linear operator $(-L_A+Q_v, \FC)$ is 
bounded from below.
We denote by the same notation $-L_A+Q_v$ the
Friedrichs extension. The spectral bottom
$\inf\sigma(-L_A+Q_v)$ is a simple eigenvalue of
$-L_A+Q_v$ and the eigenvalue and the associated 
normalized positive eigenfunction
$\Omega_v$ are given by
\begin{eqnarray}
\inf\sigma(-L_A+Q_v)&=&-\frac{1}{4}\|\left(A_v^2-
A^2\right)
A^{-1}\|_{L_{(2)}(H)}^2,
\label{ground state energy}\\
\Omega_v(w)
&=&\det{}_{(2)}(I+T_v)^{1/4}
\exp\left[
-\frac{1}{4}:\langle T_vw,w \rangle_H:
\right],\label{ground state}
\end{eqnarray}
where
$\|~\|_{L_{(2)}(H)}$ denotes the Hilbert-Schmidt
norm.
\item[{\rm (iv)}]~
The weighted measure $\Omega_v(w)^2d\mu$ is the Gaussian 
probability measure whose
covariance operator is $(m^2+4v-\Delta)^{-1/2}$
on $L^2(\RR,dx)$.
\end{itemize}
\end{lem}

Clearly, $\|\left(A_v^2-A^2\right)A^{-1}\|_{L_{(2)}(H)}^2$
is equal to $\|\left(\tilde{A}_v^2-\tilde{A}^2\right)
\tilde{A}^{-1}\|_{L_{(2)}(L^2(\RR))}^2$.
The following is an extension of the above lemma.
We need this lemma to study tunneling.

\begin{lem}\label{quadratic 2}
Let $v$ be the same function as in Lemma~$\ref{quadratic}$~$(2)$.
Also we use the same notation as in Lemma~$\ref{quadratic}$.
Let $J$ be a Hilbert-Schmidt operator on $H$.
Assume that $AJA$ is also a Hilbert-Schmidt operator.
Moreover we assume that 
$A^4+4AK_vA+4AJA$ is strictly positive
operator.
Let $A_{v,J}=\left(A^4+4AK_vA+4AJA\right)^{1/4}$.

\noindent
$(1)$~$A_{v,J}^2-A^2$ is a Hilbert-Schmidt operator
on $H$.

\noindent
$(2)$~Let $T_{v,J}=A^{-1}(A_{v,J}^2-A^2)A^{-1}$.
$T_{v,J}$ is a Hilbert-Schmidt operator with
$\inf\sigma(T_{v,J})>-1$.
Let $Q_{v,J}=Q_v+:\langle Jw,w\rangle_H:$.
Then $(-L_A+Q_{v,J},\FC)$
is bounded from below.
We use the same notation to indicate the Friedrichs extension.
$E_{v,J}=\inf\sigma(-L_A+Q_{v,J})$ is a simple eigenvalue.
The lowest eigenvalue and the corresponding normalized positive
eigenfunction $\Omega_{v,J}$ is given by
\begin{eqnarray}
E_{v,J}&=&
-\frac{1}{4}\|\left(A_{v,J}^2-
A^2\right)
A^{-1}\|_{L_{(2)}(H)}^2.\label{ground state energy 2}\\
\Omega_{v,J}(w)
&=&\det{}_{(2)}(I+T_{v,J})^{1/4}
\exp\left[
-\frac{1}{4}:\langle T_{v,J}w,w \rangle_H:
\right].\label{ground state 2}
\end{eqnarray}
\end{lem}

We will give the proof of the 
above three lemmas in Section~7.
The operator $J$ in Lemma~\ref{quadratic 2}
will appear as a second derivative of
the squared norm on $W$.

\begin{lem}\label{smooth function on W}
~Let $F(w)=\frac{1}{2}\|w\|_W^2$.

\noindent
$(1)$~We have $DF(w)=\tilde{A}^{-2}\Delta_H^{-2}w$
and $D^2F(w)=\tilde{A}^{-2}\Delta_H^{-2}$.
That is $D^2F(w)$ is equal to $S_0$.
In particular $D^2F(w)$ is a trace class operator on $H$.
Also it holds that
\begin{equation}
:\langle S_0 w,w\rangle:=
\|w\|_W^2-\tr\, S_0.\label{quadratic function S_0}
\end{equation}

\noindent
$(2)$~
It holds that
$F\in {\rm D}({\cal E}_A)$
and
$\|D_AF(w)\|_H^2=
\|(1+x^2-\Delta)^{-2}(w)\|_{L^2}^2\le C\|w\|_W^2$.

\noindent
$(3)$~The operator $AS_0A$ is a Hilbert-Schmidt operator.
\end{lem}

\begin{proof}
We denote by $\Delta_{H}$ the operator
$1+x^2-\Delta$ acting on tempered distribution.
Then $F(w)=\frac{1}{2}\int_{\RR}(\Delta_H^{-1}w)^2(x)dx$.
Hence for any $\varphi\in {\cal S}(\RR)$,
\begin{eqnarray}
D_{\varphi}F(w)
&=&\left(\Delta_H^{-1}w,\Delta_H^{-1}\varphi\right)_{L^2}\nonumber\\
&=&\left(\tilde{A}^{-2}\Delta_H^{-2}w,\varphi\right)_{H^{1/2}}.
\end{eqnarray}
Hence
$DF(w)=\tilde{A}^{-2}\Delta_H^{-2}w$,
$DF(w)\in {\rm D}(A)={\rm D}(\tilde{A}^2)$
and
$\|D_AF(w)\|_H^2=\|\Delta_H^{-2}w\|_{L^2}^2$.
A similar calculation shows also that the second derivative of $F$
is equal to $\tilde{A}^{-2}\Delta_H^{-2}$ and it belongs to 
trace class.
Finally we prove the identity (\ref{quadratic function S_0}).
Let $\{P_n\}$ be projection operators onto
the finite dimensional subspace spanned by the eigenfunctions
of $S_0$ such that $P_n$ converges to the identity operator strongly
on $H$.
By the definition of the norm of $\|~\|_W$,
we have
$\|P_nw\|_W^2=\left(S_0P_nw,P_nw\right)_H$.
Since $\tr (P_nS_0P_n)\to \tr\, S_0$ and $\|P_nw\|_W^2\to \|w\|_W^2$
for any $w\in W$ by the definition,
the proof of (1) is completed.
Since $AS_0A$ is unitarily equivalent to
$\tilde{A}^{-1}(-\Delta_H)^{-2}\tilde{A}^{-1}$,
the proof of (3) is evident.
\end{proof}

We introduce a 
set of functions dominated by $U$
to state a theorem which is an extension of
Theorem~\ref{main theorem 0}.

\begin{defin}\label{the set of FWU}~
Let $\FWU$ be the set of non-negative bounded 
globally Lipschitz continuous
functions $u$ on $W$
which satisfy the following 
conditions.

\noindent
$(1)$~It holds that
$
0\le u(h)\le U(h)~\mbox{for all}~h\in H^1
$ 
and
\begin{equation}
\{h\in H~|~U(h)-u(h)=0\}=\{h_1,\ldots,h_{n_{0}}\}={\cal Z},
\end{equation}
where ${\cal Z}$ is the zero point set of $U$.

\noindent
$(2)$~
There exists a non-negative number $\ep_i$ for each $h_i$ $(1\le i\le n_0)$
such that
\begin{equation}
u(w)=\ep_i\|w-h_i\|_W^2
\qquad \mbox{in a neighborhood of $h_i$ in the topology of $W$}.
\end{equation}

\noindent
$(3)$~Let $J_i=-\frac{1}{2}D^2u(h_i)$.
Then the self-adjoint operators
$$
A^4+4AK_{v_i}A+4AJ_iA
\qquad
(1\le i\le n_0)
$$
are strictly positive.
\end{defin}

\begin{rem}
\noindent
$(1)$~In the definition above,
we assume $u$ is equal to the squared norm of
$W$. Actually Theorem~$\ref{main theorem 1}$ holds
for more general $C^2$ function $u$ near 
${\cal Z}$ which satisfies the conditions
{\rm (1)} and {\rm (3)} in Definition~$\ref{the set of FWU}$.
But just for simplicity we consider the 
case of squared norm.
Also in this case, we have
\begin{equation}
J_i=-\ep_i\tilde{A}^{-2}\Delta_H^{-2}.\label{Ji}
\end{equation}
From now on we use the notation 
$J_i$ to express this operator.

\noindent
$(2)$~The operator $A^4+4AK_{v_i}A+4AJ_iA$ is unitarily equivalent to
the operator $m^2-\Delta+4v_i-\ep_i\Delta_H^{-2}$ on $L^2(\RR)$.
So the assumption $(3)$ implies the nondegeneracy of the
$L^2$-Hessian of $U-u$ at $h_i$.
\end{rem}

\begin{exm}\label{example of FWU}
Set
\begin{equation}
u_{{\cal Z}}(w)=
\min_{1\le i\le n_0}\|w-h_i\|_W^2.
\end{equation}
Let $R>0$ and $\kappa$ be a sufficiently small positive number.
Then $\kappa \min\left(
\uZ,R\right)\in \FWU$.
This function will appear in Section $5$.
The proof of the property $(1)$ in Definition~$\ref{the set of FWU}$
is similar to that of Lemma~$5.1$ in {\rm \cite{aida4}}.
\end{exm}

Now we state a theorem which is stronger than 
Theorem~\ref{main theorem 0} which corresponds to
the case where $u=0$.

\begin{thm}\label{main theorem 1}
~Assume that {\rm (A1)} and {\rm (A2)} hold.
Let $u\in \FWU$ and set
$u_{\la}(w)=\la u(w/\sqrt{\la})$.
Let $E_1(\la,u)=\inf\sigma(-L_A+V_{\la}-u_{\la})$.
Then
\begin{equation}
\lim_{\la\to\infty}
E_1(\la,u)=\min_{1\le i\le n_0}E_i,\label{main 1}
\end{equation}
where 
\begin{equation}
E_i=\inf\sigma(-L_A+Q_{v_i,J_i})+\tr J_i.
\end{equation}
\end{thm}

By (\ref{quadratic function S_0}),
we have
\begin{equation}
-L_A+Q_{v_i,J_i}+\tr\, J_i=
-L_A+Q_{v_i}-\ep_i\|w\|_W^2.
\end{equation}
In order to prove 
${\rm LHS}\ge {\rm RHS}$
in (\ref{main 1}),
we need a lower bound estimate for
Schr\"odinger operators of the forms
$-L_I+V$ which are perturbations
of the number operator $-L_I$ by
potential functions $V$.
This lower boundedness 
was discovered and developed by
Nelson~\cite{nelson}, Glimm~\cite{glimm},
Segal~\cite{segal}, Federbush~\cite{federbush}
and Gross~\cite{g2}.

\begin{lem}\label{NGS estimate}
$(1)$~
Let $\tilde{V}$ be a bounded measurable function.
Let $T$ be a trace class self-adjoint operator on $H$
with $\inf\sigma(I+T)>0$.
Then
\begin{eqnarray}
\lefteqn{m\int_W\|(I+T)Df(w)\|_H^2d\mu+
\int_W\tilde{V}(w)f(w)^2d\mu}\nonumber\\
& &\hspace{-0.3cm}\ge
-\frac{m}{2}\log
\left\{
\int_W\exp\left(-\frac{2}{m}\tilde{V}\left(w\right)
-\left(Tw,w\right)_H
-\frac{1}{2}\|Tw\|_H^2
\right)
d\mu(w)
\right\}\|f\|_{L^2(\mu)}^2\nonumber\\
& &+\left(\frac{m}{2}\log\det\left(I+
T\right)-\frac{m}{2}
\tr\left(T^2\right)-m\,\tr T\right)\|f\|_{L^2(\mu)}^2.
\label{NGS estimate 1}
\end{eqnarray}

\noindent
$(2)$~
Let $d\mu_{v,J}=\Omega_{v,J}^2d\mu$.
Let
\begin{eqnarray}
\E_{A,v,J}(g,g)
&=&\int_W\|ADg(w)\|_H^2d\mu_{v,J},\quad g\in {\rm D}(\E_{A,v,J}),
\end{eqnarray}
be the closure of the closable form
$(\E_{A,v,J},\FC)$.
Let $c_{v,J}=\inf\sigma(I+T_{v,J})$.
Then
the following logarithmic Sobolev inequality holds.
For any $g\in {\rm D}(\E_{A,v,J})$,
\begin{eqnarray}
\int_Wg(w)^2\log\left(g(w)^2/\|g\|_{L^2(\mu_{v,J})}^2\right)
d\mu_{v,J}(w)
&\le&
\frac{2}{mc_{v,J}}\int_W\|ADg(w)\|_H^2d\mu_{v,J}(w).
\label{log Sobolev}
\end{eqnarray}

\noindent
$(3)$~
Let $\tilde{f}=f\Omega_{v,J}^{-1}$.
Then for any $f\in \FC$,
\begin{eqnarray}
\lefteqn{\left((-L_A+Q_{v,J}
+\tilde{V}-E_{v,J})f,f\right)_{L^2(W,d\mu)}}
\nonumber\\
& &=
\int_{W}\|D_A\tilde{f}(w)\|_{H}^2d\mu_{v,J}(w)
+\int_{W}\tilde{V}(w)\tilde{f}(w)^2d\mu_{v,J}(w)
\label{ground state transform}
\end{eqnarray}
and
\begin{eqnarray}
\lefteqn{\left((-L_A+Q_{v,J}
+\tilde{V}-E_{v,J})f,f\right)_{L^2(W,d\mu)}}
\nonumber\\
& &\ge
-\frac{mc_{v,J}}{2}
\log\left(
\int_W\exp\left(
-\frac{2}{mc_{v,J}}\tilde{V}(w)
\right)d\mu_{v,J}(w)
\right)\|f\|_{L^2(\mu)}^2.\nonumber\\
& &\label{NGS estimate 2}
\end{eqnarray}
\end{lem}

\begin{proof}
The inequality in (1) follows from 
Gaussian logarithmic Sobolev inequality
~\cite{g2, federbush}.
For example, see Theorem~4.3 in \cite{aida3}.
We prove (2).
By the Bakry-Emery criterion, we obtain
\begin{equation}
\int_Wg(w)^2\log\left(
g(w)^2/\|g\|_{L^2(\mu_{v,J})}^2
\right)d\mu_{v,J}(w)\le
\frac{2}{c_{v,J}}
\int_W\|Dg(w)\|_H^2d\mu_{v,J}(w).
\end{equation}
By combining this inequality with
$m\|Dg(w)\|_H^2\le \|ADg(w)\|_H^2$,
we get the desired inequality.
We prove (\ref{ground state transform}).
Note that $:\langle T_{v,J}w,w\rangle :\in {\rm D}(\E_{A,v,J})$
and the sequence
$\{\tilde{f}\chi(:\langle T_{v,J}w,w\rangle :/n)\}$
converges to $\tilde{f}$ in ${\rm D}(\E_{A,v,J})$, where
$\chi$ is a $C^{\infty}$ function with
$\chi(x)=1$ for $|x|\le 1$ and $\chi(x)=0$
for $|x|\ge 2$.
Thus $\tilde{f}\in {\rm D}(\E_{A,v,J})$.
The identity (\ref{ground state transform}) follows from
$(-L_A+Q_{v,J})\Omega_{v,J}=E_{v,J}\Omega_{v,J}$
and a simple direct calculation.
(\ref{NGS estimate 2}) follows from 
(\ref{log Sobolev}) and (\ref{ground state transform}).
See \cite{federbush, g2}.
\end{proof}
In the proof of (\ref{main 1}), we 
use Lemma~\ref{NGS estimate},
large deviation estimates and Laplace's asymptotic formula
for Wiener chaos in Lemma~\ref{large deviation} and
Lemma~\ref{Laplace}.
The following two lemmas are essential for
the large deviation estimates.
The hypercontractivity of the Ornstein-Uhlenbeck
semigroup is the key for the proofs.
But the proofs are almost similar to
those of Lemma~2.14 and Lemma~2.15 in \cite{aida2}
and we refer the reader for the proofs to
them.
Note that the large deviation estimates originally
are due to \cite{borell}.
See also \cite{ledoux1, ledoux2, fang}.

\begin{lem}\label{good approximation}
Let $w_n$ be the approximation function of $w$ defined in
the Section~$2$.
For any $\delta>0$,
\begin{eqnarray}
\lim_{n\to\infty}\limsup_{\la\to\infty}\frac{1}{\la}
\log\mu\left(\left\{w~\Bigg |~\left|
:V\left(\frac{w}{\sqrt{\la}}\right):-V\left(\frac{w_n}{\sqrt{\la}}\right)
\right|>\delta\right\}\right)&=&-\infty.
\end{eqnarray}
\end{lem}

\begin{lem}\label{large deviation}
Let $T$ be a trace class self-adjoint operator on $H$
and $v$ be a bounded continuous function on $W$.
We write $v_{\la}(w)=\la v(w/\sqrt{\la})$.
Let $\chi$ be a non-negative bounded continuous function
and
set
$$
F_{\la}(w)=\left(V_{\la}(w)-v_{\la}(w)\right)
\chi\left(\frac{\|w\|_W^2}{\la}\right)
+\left(Tw,w\right)_H,\qquad w\in W
$$ 
and
$F(h)=\left(V(h)-v(h)\right)\chi\left(\|h\|_W^2\right)
+(Th,h)_H$ for $h\in H$.

\noindent
$(1)$~The image measure of $\mu$ by the measurable map $\frac{F_{\la}}{\la}$ 
satisfies the large deviation principle with 
the good rate function:
$$
I_F(x)=
\begin{cases}
\inf\left\{\frac{1}{2}\|h\|_H^2~\Big |~
\mbox{there exists $h\in H$ such that $F(h)=x$}\right\},& \\
+\infty \qquad \mbox{there are no $h\in H$ such that $F(h)=x$}. &
\end{cases}
$$

\noindent
$(2)$~Assume that $I+2T$ is a strictly positive operator on $H$.
Then there exists $\alpha_0>1$ such that for any $0<\alpha<\alpha_0$,
\begin{eqnarray}
\lefteqn{\lim_{\la\to\infty}\frac{1}{\la}
\log\left(
\int_W\exp\Bigl(
-\alpha F_{\la}(w)
\Bigr)d\mu(w)
\right)}\nonumber\\
& &
=-\min\left\{\frac{1}{2}\|h\|_H^2+
\alpha \left(V(h)-v(h)\right)\chi(\|h\|_W^2)
+\alpha(Th,h)_H~\Big |~h\in H\right\}.\label{Varadhan}
\end{eqnarray}
\end{lem}

Further, we need Laplace asymptotic formula for Wiener chaos.
We use Lemma~\ref{gagliard-nirenberg estimate} which will be proved
later.

\begin{lem}\label{Laplace}
Let $\chi$ be a smooth non-negative function such that
$\chi(x)=1$ for $|x|\le 1$,
$\chi(x)=0$ for $|x|\ge 2$ and
$0\le \chi\le 1$.
Set 
$\chi_{\la,\ep}(w)=\chi\left(\frac{\|w\|_W^2}{\la\ep}\right)$.
Let $f_k(x)$~$(3\le k\le 2M-1)$ be continuous functions on $\RR$
and $f_{2M}(x)=b_{2M}$ be a positive
constant.
Let
\begin{eqnarray}
G_{\la}(w)&=&\la\sum_{k=3}^{2M}\int_{\RR}
:\left(\frac{w(x)}{\sqrt{\la}}\right)^k:f_k(x)g(x)dx.
\end{eqnarray}
Then for sufficiently small $\ep$,
\begin{equation}
\lim_{\la\to\infty}\int_We^{-G_{\la}(w)\chi_{\la,\ep}(w)}
d\mu(w)=1.\label{Laplace2}
\end{equation}
\end{lem}

\begin{proof}
Let 
\begin{eqnarray}
G(h)&=&\sum_{k=3}^{2M}\int_{\RR}h(x)^kf_k(x)g(x)dx.
\end{eqnarray}
Using the following identity with sufficiently small positive $\delta, \kappa$,
\begin{eqnarray}
\sum_{k=3}^{2M}h(x)^kf_k(x)
&=&
h(x)^{2M}((b_{2M}-\delta)-(\delta+\kappa f_k(x)))\nonumber\\
& &+
\sum_{k=3}^{2M-1}\left(\frac{1}{2M-3}h(x)^{2M}-\kappa^{-1}h(x)^k\right)
\delta\nonumber\\
& &+
\sum_{k=3}^{2M-1}
\left(\frac{1}{2M-3}h(x)^{2M}+
\kappa^{-1}h(x)^k\right)(\delta+\kappa f_k(x)),
\label{decomposition of polynomial}
\end{eqnarray}
we have
\begin{equation}
\inf_{h\in H}G(h)
>-\infty.
\end{equation}
So for any $a>0$
$$
\lim_{\|h\|_H\to\infty}
\left(\frac{1}{2}\|h\|_H^2+a G(h)
\chi_{1,\ep}(h)\right)=\infty.
$$
By Lemma~\ref{gagliard-nirenberg estimate},
for any $\delta>0$ and $R>0$, there exists
$C(\delta,R)$ such that
\begin{equation}
\int_{\RR}|h(x)|^k|f_k(x)|g(x)dx
\le C\delta^{k-2}\|h\|_H^2
\qquad \mbox{for}~
\|h\|_W\le C(\delta,R),~\|h\|_H\le R.
\end{equation}
Thus, for sufficiently small $\ep$,
$$
\inf\left\{\frac{1}{2}\|h\|_H^2+2G(h)
\chi_{1,\ep}(h)
~\Big |~h\in H\right\}>0.
$$
Also by the decomposition of the polynomial
(\ref{decomposition of polynomial})
we obtain that for any $\alpha>0$,
\begin{equation}
\limsup_{\la\to\infty}\frac{1}{\la}\log 
\left(\int_W\exp\left(
-\alpha G_{\la}(w)\right)\right)<\infty.
\end{equation}
Hence, by a similar argument to
the proof of Lemma~\ref{good approximation},
$G_{\la}(w)/\la$ satisfies the large deviation principle
with the rate function 
$I_{G}$ which is defined similarly to
$I_F$.
By a similar argument to the proof of
Lemma~2.16 in \cite{aida2}, we can complete the proof.
\end{proof}

\begin{rem}
In the estimate $(2.61)$ in {\rm \cite{aida2}},
we used large deviation results
without mentioning the above decomposition of
the polynomial $(\ref{decomposition of polynomial})$.
However, the argument above is necessary
because the function $f_k$ may be different
for each $k$ differently from the setting
in Lemma~$\ref{large deviation}$.
\end{rem}

To apply Lemma~\ref{NGS estimate} to
the proof of Theorem~\ref{main theorem 1},
we need to approximate $A$ by
bounded linear operators of the form,
$\sqrt{m}(I+\mbox{trace class operator})$.
To this end, we introduce a family of projection
operators which depend on a positive parameter $l$.
We fix a complete orthonormal system
$\{e_n\}_{n=1}^{\infty}$ on $L^2([0,1],dx)$.
The set $\{e_n\}$ are also a c.o.n.s of
$L^2([0,1],dx)_{{\mathbb C}}$.
Let $l>0$ and set $e_{n,l}(x)=\frac{1}{\sqrt{l}}e_n(x/l)$.
Then $\{e_{n,l}\}$ is a c.o.n.s of
$L^2([0,l],dx)$ and $L^2([0,l],dx)_{{\mathbb C}}$.
Let $I_{k,l}=[kl,(k+1)l)$, where $k\in {\mathbb Z}$.
Let us define
$$
e_{n,l,k}(x)=
\begin{cases}
e_{n,l}(x-kl) & (k\ge 0),\\
e_{n,l}(-x+(k+1)l) & (k<0).
\end{cases}
$$
We extend $e_{n,l,k}$ to a function on $\RR$ setting
$e_{n,l,k}(x)=0$ for $x\notin I_{k,l}$.
The family of functions on $\RR$,
$\{e_{n,l,k}~|~n\in {\mathbb N}, k\in {\mathbb Z}\}$
is a c.o.n.s. of $L^2(\RR)$ and
$L^2(\RR)_{{\mathbb C}}$.
Let us define the Fourier transform and the inverse transform 
on $L^2(\RR)_{{\mathbb C}}$:
\begin{eqnarray}
{\mathfrak F}\varphi(\xi)&=&\hat{\varphi}(\xi)=\frac{1}{\sqrt{2\pi}}
\int_{\RR}e^{-\ai x\xi}\varphi(x)dx,\\
{\mathfrak F}^{-1}\psi(x)&=&\check{\psi}(x)
=\frac{1}{\sqrt{2\pi}}\int_{\RR}e^{\ai x\xi}\psi(\xi)d\xi.
\end{eqnarray}
Note that $e_{n,l,k}$ satisfies the following relation:
\begin{equation}
\int_{\RR}e^{-\sqrt{-1}x\xi}e_{n,l,k}(x)dx=
\int_{\RR}e^{\sqrt{-1}x\xi}e_{n,l,-k-1}(x)dx,
\quad k\in {\mathbb Z}.
\end{equation}
Let $K,N$ be natural numbers.
Let
$\tilde{P}_{N,K,l}$ be the projection operator onto the linear span of
\begin{equation}
\left\{\left({\mathfrak F}^{-1}e_{n,l,k}\right)(x)~\Big |~
1\le n\le N, -K\le k\le K-1
\right\}
\end{equation}
in
$L^2(\RR)_{{\mathbb C}}$.
More explicitly,
\begin{eqnarray}
\tilde{P}_{N,K,l}\varphi(x)&=&\sum_{1\le n\le N, -K\le k\le K-1}
\left(\varphi,{\mathfrak F}^{-1}e_{n,l,k}\right)_{L^2(\RR)_{\mathbb C}}
\left({\mathfrak F}^{-1}e_{n,l,k}\right)(x).
\end{eqnarray}
If $\varphi$ is a real-valued function, for
$k\in {\mathbb Z}$,
\begin{eqnarray}
\overline{
\left(\varphi,{\mathfrak F}^{-1}e_{n,l,k}\right)_{L^2(\RR)_{\mathbb C}}
\left({\mathfrak F}^{-1}e_{n,l,k}\right)(x)}
&=&
\left(\varphi,{\mathfrak F}^{-1}e_{n,l,-k-1}\right)_{L^2(\RR)_{\mathbb C}}
\left({\mathfrak F}^{-1}e_{n,l,-k-1}\right)(x)
\end{eqnarray}
where $\bar{z}$ denotes the complex conjugate of $z$.
Hence $\tilde{P}_{N,K,l}\varphi$ is also a real valued function.
This implies that $\tilde{P}_{N,K,l}$ is also a projection operator
on $L^2(\RR)$.

Next, we define a family of projection operators on $H=H^{1/2}(\RR)$.
Let us consider a unitary map 
$\Psi :L^2(\RR)_{\CC}\to H^{1/2}(\RR)_{\CC}$ which is defined by
$\Psi={\mathfrak F}^{-1}M_{\omega^{-1/2}}{\mathfrak F}$,
where $M_{\omega^{-1/2}}g(\xi)=\omega(\xi)^{-1/2}g(\xi)$
and $\omega(\xi)=\left(m^2+\xi^2\right)^{1/2}$.
Clearly this unitary transformation preserves the
real-valued subspaces and $\Psi|_{L^2(\RR)}=\Phi$.
We define a projection operator on $H^{1/2}_{{\mathbb C}}$ by
\begin{equation}
P_{N,K,l}h(x)=
\Psi\circ \tilde{P}_{N,K,l}\circ \Psi^{-1}h(x).
\end{equation}
Since $\Psi$ preserves the real-valued subspace,
$P_{N,K,l}$ is a projection operator on $H^{1/2}$.
This operator can be defined in the following way too.
Take $h\in L^2$.
Then $h\in H$ is equivalent to
${\mathfrak F}h\in L^2(\omega(\xi)d\xi)$ and
\begin{equation}
(h,k)_H=\int_{\RR}({\mathfrak F}h)(\xi)\overline{({\mathfrak F}k)(\xi)}
\omega(\xi)d\xi.
\end{equation}
Therefore
$\{{\mathfrak F}^{-1}\left(\omega^{-1/2}e_{n,l,k}\right)~|~
n\in {\mathbb N}, k\in {\mathbb Z}
\}$ constitutes a c.o.n.s. of $H_{{\mathbb C}}$.
The projection $P_{N,K,l}$ is nothing but a projection operator
onto a linear span of
$\{{\mathfrak F}^{-1}\left(\omega^{-1/2}e_{n,l,k}\right)~|~
1\le n\le N, -K\le k\le K-1
\}$.
Note that
${\rm Im}\,P_{N,K,l}\subset {\rm D}(A^n)$
for all $n\ge 1$.
Also for any $l$, 
$
P_{N,K,l}
$
converges to the identity operator on $H$ strongly as
$N,K\to\infty$.
The following Gagliard-Nirenberg type estimate is used in the proof of
Lemma~\ref{Laplace} and the estimate for
the weighted $L^p$-estimate on
$P_{N,K,l}h-h$.

\begin{lem}\label{gagliard-nirenberg estimate}
Let $p\ge 2$.
Let $g$ be a non-negative bounded measurable function
such that 
$$
C(p,g)=\max\left\{
\int_{\RR}|g(x)|dx, \int_{\RR}|x|^pg(x)dx,
\int_{\RR}|x|^{2p}g(x)dx
\right\}<\infty.
$$
Let 
$s$ be a positive number such that $\frac{p-2}{2p}<s<\frac{1}{2}$.
Then there exists a positive constant $C$ which depends on
$C(p,g)$, $\|g\|_{\infty}$ and $s$ such that for any $h\in H^{1/2}$,
\begin{eqnarray}
\left\{\int_{\RR}|h(x)|^pg(x)dx\right\}^{1/p}
&\le&
C\|h\|_{H^{1/2}}^{a(s)}\|h\|_{W}^{1-a(s)},\label{g-n estimate}
\end{eqnarray}
where $a(s)=3/(4-2s)$.
\end{lem}

\begin{proof}
Let $\varphi$ be a 
$C^{\infty}$ function such that
$\varphi(x)=1$ for $|x|\le 1$ and $\varphi(x)=0$
for $|x|\ge 2$ and $0\le \varphi(x)\le 1$
for all $x$.
Let $R\ge 1$.
We consider the following decomposition.
\begin{equation}
h={\mathfrak F}^{-1}\left(\hat{h}\varphi(\cdot/R)\right)
+{\mathfrak F}^{-1}\left(\hat{h}(1-\varphi(\cdot/R))\right)
=h_1+h_2.
\end{equation}
Let $\Delta_{H,\xi}=(1+|\xi|^2-\Delta_{\xi})$.
Using the integration by parts formula, we obtain
\begin{eqnarray}
h_1(x)&=&\frac{1}{\sqrt{2\pi}}
\int_{\RR}e^{\ai x\xi}\varphi(\xi/R)
\Delta_{H,\xi}\Delta_{H,\xi}^{-1}\hat{h}(\xi)d\xi\nonumber\\
&=&
\frac{1}{\sqrt{2\pi}}\int_{\RR}e^{\ai x\xi}\varphi(\xi/R)
(1+|\xi|^2+|x|^2)\Delta_{H,\xi}^{-1}\hat{h}(\xi)d\xi\nonumber\\
& &-\frac{2}{\sqrt{2\pi}}\int_{\RR}\frac{\ai x}{R}e^{\ai x\xi}
\varphi'(\xi/R)
\Delta_{H,\xi}^{-1}\hat{h}(\xi)d\xi\nonumber\\
& &-\frac{1}{\sqrt{2\pi}}\int_{\RR}
\frac{1}{R^2}\varphi''(\xi/R)e^{\ai x\xi}
\Delta_{H,\xi}^{-1}\hat{h}(\xi)d\xi.
\end{eqnarray}
Using the Schwarz inequality and the commutativity of
${\mathfrak F}^{-1}$ and $\Delta_{H,\xi}$, we have
\begin{eqnarray}
|h_1(x)|&\le&
\frac{C_{\varphi}}{\sqrt{2\pi}}\|h\|_W\Bigl(
\sqrt{R}|x|^2+2|x|R^{-1/2}+\sqrt{R}(1+4R)+R^{-3/2}
\Bigr)
\end{eqnarray}
and 
\begin{equation}
\|h_1\|_{L^p(gdx)}
\le
C R^{3/2}\|h\|_W.\label{Lp estimate for h1}
\end{equation}
Next we estimate $h_2$.
By Lemma~\ref{sobolev},
\begin{eqnarray}
\|h_2\|_{L^p}&\le &
C_{p,s}\left(\int_{\RR}
|\hat{h}(\xi)(1-\varphi(\xi/R))|^2
(m^2+\xi^2)^{s}d\xi\right)^{1/2}\nonumber\\
&\le&
C_{p,s}
\left(\int_{\RR}
\frac{(m^2+\xi^2)^{1/2}}{(m^2+R^2)^{\frac{1}{2}-s}}
|\hat{h}(\xi)|^2d\xi
\right)^{1/2}\nonumber\\
&=&C_{p,s}\left(m^2+R^2\right)^{-\frac{1}{2}(\frac{1}{2}-s)}
\|h\|_{H^{1/2}}.\label{Lp estimate for h2}
\end{eqnarray}
The estimates
(\ref{Lp estimate for h1}) and (\ref{Lp estimate for h2})
imply for any $R\ge 1$,
\begin{eqnarray}
\left\{\int_{\RR}|h(x)|^pg(x)dx\right\}^{1/p}
&\le&
\frac{C_1\|g\|_{\infty}}{R^{(\frac{1}{2}-s)}}\|h\|_{H^{1/2}}
+C_2 R^{3/2}\|h\|_W,
\end{eqnarray}
where $C_1$ is a constant which depends on 
$m,p,s$ and $C_2$ is a constant which depends on
$\|g\|_{L^1}$, $\int_{\RR}|x|^pg(x)dx$ and
$\int_{\RR}|x|^{2p}g(x)dx$.
Let $C_3=\sup_{h\ne 0}\frac{\|h\|_W}{\|h\|_H}$.
Clearly $C_3<\infty$.
Putting $R=\left(C_3\frac{\|h\|_H}{\|h\|_W}\right)
^{1/(2-s)}$, we get the estimate~(\ref{g-n estimate}).
\end{proof}

Using the preliminaries above, 
we approximate 
$A$ by bounded linear operators which are of the form
$\sqrt{m}(I+\mbox{trace class operator})$.
Let $R$ be a positive number.
Let $\psi_R(x)$ be a positive function on $[0,\infty)$
such that $\psi_R(x)=1$ for $0\le x\le \omega(R)^{1/2}$ and
$\psi_R(x)=\omega(R)^{1/2}/x$ for $x\ge \omega(R)^{1/2}$.
Let $A^{(R)}=A\psi_R(A)$.
Then $A^{(R)}$ is a bounded linear operator and
$\|A^{(R)}\|_{op}=\omega(R)^{1/2}$.
In the first step, we approximate $A$
by $A^{(R)}$ as in the following lemma.
From now on, we use the following notation.
For $r>0$ and $z\in W, k\in H$, let
$B_r(z)=\{w\in W~|~\|w-z\|_W<r\}$
and $B_{r,H}(k)=\{h\in H~|~\|h-k\|_H<r\}$.
Also, we define
$
B_{\ep}({\cal Z})=\cup_{i=1}^{n_0}B_{\ep}(h_i).
$

\begin{lem}\label{first approximation}
Assume $U$ satisfies {\rm (A1)} and {\rm (A2)}.
Let $u\in \FWU$.

\noindent
$(1)$~For any $\ep>0$, there exists $\beta(\ep)>0$ such that
\begin{equation}
\inf\left\{U(h)-u(h)~|~h\in B_{\ep}({\cal Z})^c
\cap {\rm D}(A)\right\}\ge \beta(\ep).
\end{equation}

\noindent
$(2)$~For any $\ep>0$, there exist $R>0$ and
$\delta(\ep,R)>0$ such that
\begin{eqnarray}
\inf\left\{\frac{1}{4}\left\|A^{(R)}h\right\|_H^2
+V(h)-u(h)~\Big |~
h\in B_{\ep}({\cal Z})^{c}\cap H
\right\}&\ge& \delta(\ep,R).\label{lower bound AR0}
\end{eqnarray}
\end{lem}

\begin{proof}
(1)~
We have $\|Ah\|_H\ge \sqrt{m}\|h\|_H$ for any $h$.
Since for any $h\in H$
\begin{equation}
V(h)-u(h)\ge \int_{\RR}
\left(\inf_x P(x)\right) g(x)dx-\sup_hu(h)
=:\kappa>-\infty,\label{kappa}
\end{equation}
$\liminf_{\|h\|_H\to\infty}\left(\frac{1}{4}\|Ah\|_H^2+
V(h)-u(h)\right)=+\infty$.
Hence it suffices to show that for fixed $R_0>0$
\begin{eqnarray}
\inf\left\{\frac{1}{4}\left\|Ah\right\|_H^2+V(h)-u(h)~\Big |~
h\in B_{\ep}({\cal Z})^c\cap B_{R_0,H}(0)\cap
{\rm D}(A)
\right\}&\ge& \beta(\ep)>0.
\end{eqnarray}
Assume that there exist 
$\{\varphi_n\}\subset B_{\ep}({\cal Z})^c\cap B_{R_0,H}(0)\cap
{\rm D}(A)$ such that $\lim_{n\to\infty}(U(\varphi_n)
-u(\varphi_n))=0$.
By Lemma~\ref{sobolev}, $\sup_n |V(\varphi_n)|<\infty$.
Hence $\sup_{n}\|\varphi_n\|_{H^1}^2<\infty$.
Therefore we may assume that $\varphi_n$ converges 
weakly to some $\tilde{\varphi}\in H^1$
in $H^1$.
Since the inclusion $H\hookrightarrow W$ is a
Hilbert-Schmidt operator,
$\lim_{n\to\infty}\|\varphi_n-\tilde{\varphi}\|_W=0$
and $\lim_{n\to\infty}u(\varphi_n)=u(\varphi)$.
By Lemma~\ref{gagliard-nirenberg estimate},
$\lim_{n\to\infty}V(\varphi_n)=V(\tilde{\varphi})$.
Combining these, we get
$$
U(\tilde{\varphi})-u(\tilde{\varphi})\le\liminf_{n\to\infty}
\left(\frac{1}{4}\|\varphi_n\|_{H^1}^2+V(\varphi_n)
-u(\varphi_n)\right)=
0.
$$
This implies $\tilde{\varphi}\in {\cal Z}$.
However, since $\lim_{n\to\infty}\|\varphi_n-\tilde{\varphi}\|_W=0$,
this contradicts the assumptions on $\{\varphi_n\}$.

(2)~
It suffices to show that for fixed $R_0>0$ and
any $\ep>0$, there exist $R$ and $\delta(\ep,R)>0$ such that
for any $h\in B_{\ep}({\cal Z})^c\cap B_{R_0,H}(0)$,
\begin{eqnarray}
\frac{1}{4}\left\|A^{(R)}h\right\|_H^2+V(h)-u(h)&\ge&
\delta(\ep,R).\label{lower bound AR}
\end{eqnarray}
Pick any $h\in B_{\ep}({\cal Z})^c\cap B_{R_0,H}(0)$.
There are two cases where 
\begin{itemize}
\item[(i)]there exists $i$ such that $\psi_R(A)h\in B_{\ep/2}(h_i)$,
\item[(ii)]it holds that $\psi_R(A)h\in B_{\ep/2}({\cal Z})^c
\cap B_{R_0,H}(0)$.
\end{itemize}
We consider the case (i).
Since $\|h-h_i\|_W\ge \ep$,
$$
\|\chi_{[\omega(R)^{1/2},\infty)}(A)h\|_H
\ge \|h-\psi_R(A)h\|_H\ge C\|h-\psi_R(A)h\|_W\ge C\ep/2.
$$ 
Noting $A^{(R)}h=A\chi_{[0,\omega(R)^{1/2})}(A)h+
\omega(R)^{1/2}\chi_{[\omega(R)^{1/2},\infty)}(A)h$,
we have
$$
\frac{1}{4}\|A^{(R)}h\|_H^2+V(h)-u(h)
\ge \frac{\ep^2C^2\omega(R)}{16}+\kappa,
$$
where $\kappa$ is defined in (\ref{kappa}).
Hence, for large $R$, (\ref{lower bound AR}) holds.
We consider the case (ii).
If $\frac{1}{4}\|A^{(R)}h\|_H^2\ge |\kappa|+\ep$, then
$\frac{1}{4}\|A^{(R)}h\|_H^2+V(h)-u(h)\ge \ep$.
So we may assume that $\frac{1}{4}\|A^{(R)}h\|_H^2\le |\kappa|+\ep$.
In this case, 
\begin{equation}
\|h-\psi_R(A)h\|_H^2\le
\|\chi_{[\omega(R)^{1/2},\infty)}(A)h\|_H^2
\le \frac{4(|\kappa|+\ep)}{\omega(R)}.\label{psiperp}
\end{equation}
Hence
\begin{eqnarray}
|V(h)-V(\psi_R(A)h)|&\le&
C(1+\|h\|_H)^{2M-1}\|h-\psi_R(A)h\|_H\nonumber\\
&\le& C(1+R_0)^{2M-1}\left(\frac{4(|\kappa|+\ep)}{\omega(R)}\right)^{1/2},
\label{difference of V and Vpsi}\\
|u(h)-u(\psi_R(A)h)|&\le&
C\|h-\psi_R(A)h\|_W
\le 2C\left(\frac{|\kappa|+\ep}{\omega(R)}\right)^{1/2}.
\end{eqnarray}
In (\ref{difference of V and Vpsi}), we have used 
Lemma~\ref{sobolev}.
Thus, we have
\begin{eqnarray}
\frac{1}{4}\left\|A^{(R)}h\right\|_H^2+V(h)-u(h)
&=&
\frac{1}{4}\left\|A^{(R)}h\right\|_H^2+V(\psi_R(A)h)-u(\psi_R(A)h)
\nonumber\\
& &
~~+V(h)-V(\psi_R(A)h)
-(u(h)-u(\psi_R(A)h))\nonumber\\
&\ge&
\beta(\ep/2)-4C(1+R_0)^{2M-1}
\left(\frac{|\kappa|+\ep}{\omega(R)}\right)^{1/2}.
\end{eqnarray}
Therefore, (\ref{lower bound AR}) holds.
\end{proof}

We introduce approximate operators 
of $A$.
From now on, we assume that $R/l\in \NN$.
Let
\begin{eqnarray*}
A_l&=&
\sum_{k=0}^{\infty}\omega(kl)^{1/2}
1_{[\omega(kl)^{1/2},\omega((k+1)l)^{1/2})}(A).
\end{eqnarray*}
By the assumption $R/l\in \NN$,
we have $(A^{(R)})_l=\left(A_l\right)^{(R)}$.
Hence we can use the notation $A^{(R)}_l$ for
this operator without ambiguity.
Also we define
\begin{equation}
\quad A^{(R)}_{N,K,l}=A^{(R)}_{l}P_{N,K,l},\quad 
A^{(R)}_{N,l}=A^{(R)}_{N,R/l,l}.
\end{equation}
More explicitly,
\begin{eqnarray}
A^{(R)}_{N,l}h&=&
\sum_{n=1}^N
\sum_{-(R/l)\le k\le -1}\omega((k+1)l)^{1/2}
\left(h,{\mathfrak F}^{-1}(\omega^{-1/2}e_{n,l,k})\right)_H
{\mathfrak F}^{-1}(\omega^{-1/2}e_{n,l,k})\nonumber\\
& &+
\sum_{n=1}^N
\sum_{0\le k\le (R/l)-1}\omega(kl)^{1/2}
\left(h,{\mathfrak F}^{-1}(\omega^{-1/2}e_{n,l,k})\right)_H
{\mathfrak F}^{-1}(\omega^{-1/2}e_{n,l,k}).
\end{eqnarray}
Finally, we set
\begin{equation}
A_{N,l}=A^{(Nl)}_{N,l},\quad
P_{N,l}=P_{N,N,l}.
\end{equation}
We have 
${\rm Im}\, P_{N,l}\subset {\rm Im}\, P_{N+1,l}$ 
and $\lim_{N\to\infty}P_{N,l}=I$ strongly.
Note that 
$\sqrt{m}P_{N,l}^{\perp}+A_{N,l}$ is an approximation
operator of $A$ for small $l$ and large $N$.
We have the following lemmas for these operators.
The first lemma is easy and we omit the proof.

\begin{lem}
$(1)$~The bounded linear operators 
$P_{N,l}, P_{N,l}^{\perp}, A^{(Nl)}_l, A_{N,l}$ commute.

\noindent
$(2)$~The image of the operator $A_{N,l}$ is a
finite dimensional subspace of $H$.
\end{lem}

\begin{lem}\label{second approximation}

$(1)$~
For any $h\in {\rm D}(A)$
\begin{equation}
\|Ah\|_H^2\ge \|A^{(Nl)}h\|_H^2\ge
\|A^{(Nl)}_lh\|_H^2\ge
\|(\sqrt{m}P_{N,l}^{\perp}+A_{N,l})h\|_H^2
\end{equation}
and $I-\left(P_{N,l}^{\perp}
+\frac{1}{\sqrt{m}}A_{N,l}\right)$
is a finite dimensional operator,
especially, a trace class operator.

\noindent
$(2)$~Assume that $U$ satisfies {\rm (A1)} and {\rm (A2)}.
Let $u\in \FWU$.
For any $\ep>0$, there exist $\delta(\ep)'>0$,
$N\in {\mathbb N}, l>0$ such that
\begin{eqnarray*}
\inf\left\{\frac{1}{4}\left\|(\sqrt{m}P_{N,l}^{\perp}
+A_{N,l})h\right\|_H^2+V(h)-u(h)~\Big |~
h\in B_{\ep}({\cal Z})^{c}\cap H
\right\}&\ge& \delta(\ep)'.
\end{eqnarray*}
\end{lem}

\begin{proof}[Proof of Lemma~$\ref{second approximation}$]

It is easy to check (1).
We prove (2).
By a similar argument to the proof of 
Lemma~\ref{first approximation} (1),
it is enough to show that for fixed large $R_0>0$ and
any $\ep>0$, there exist $\delta(\ep)'$
and $N\in \NN$, $l>0$
such that for any 
$h\in B_{\ep}({\cal Z})^{c}\cap B_{R_0,H}(0)$,
\begin{equation}
\frac{1}{4}\left\|(\sqrt{m}P_{N,l}^{\perp}
+A_{N,l})h\right\|_H^2+V(h)-u(h)
\ge \delta(\ep)'.\label{lower bound ARNl}
\end{equation}
Note that for $h\in H$,
\begin{eqnarray}
\|A^{(Nl)}_lh\|_H^2&\ge&
\|A^{(Nl)}h\|_H^2-l\|h\|_H^2.
\end{eqnarray}
Take small $l$ and large $N$ such that $lR_0^2\le \delta(\ep,Nl)$,
where $\delta(\ep,R)$ is the number in (\ref{lower bound AR0}).
Then 
by (\ref{lower bound AR0}), 
we get
\begin{eqnarray}
\inf\left\{\frac{1}{4}\left\|A^{(Nl)}_lh\right\|_H^2
+V(h)-u(h)~\Big |~
h\in B_{\ep}({\cal Z})^{c}\cap B_{R_0,H}(0)
\right\}&\ge& \delta(\ep,Nl)/2=:\delta(\ep)'.\label{lower bound ARl}
\end{eqnarray}
By using the commutativity of $A^{(Nl)}_{N,l}$ and $P_{N,l}$,
we have
\begin{eqnarray}
\lefteqn{\frac{1}{4}
\|\sqrt{m}P_{N,l}^{\perp}h+A_{N,l}h\|_H^2+
V(h)-u(h)}\nonumber\\
& &=
\frac{m}{4}\|P_{N,l}^{\perp}h\|_H^2+
\frac{1}{4}\|A_{N,l}h\|_H^2+
V(P_{N,l}h)-u(P_{N,l}h)\nonumber\\
& &\quad+V(h)-V(P_{N,l}h)-
\left(u(h)-u(P_{N,l}h)\right).\label{identity ARl}
\end{eqnarray}
By the same argument as in the proof of Lemma~\ref{first approximation},
we may assume that 
\begin{equation}
\frac{1}{4}\|A_{N,l}h\|_H^2\le |\kappa|+\ep.\label{psiRP}
\end{equation}
By using the H\"older inequality and Lemma~\ref{sobolev}
we have 
\begin{eqnarray}
|V(h)-V(P_{N,l}h)|
&\le& C
\sum_{k=1}^{2M}
\sum_{r=0}^{k-1}
\|P_{N,l}h\|_{L^k(gdx)}^r
\|P_{N,l}^{\perp}h\|_{L^k(gdx)}^{k-r}\nonumber\\
&\le& C_g
\sum_{k=1}^{2M}
\sum_{r=0}^{k-1}
\|h\|_{H}^r
\|P_{N,l}^{\perp}h\|_{L^k(gdx)}^{k-r}.\label{V difference}
\end{eqnarray}
Let $\delta$ be a positive number.
Note that $\lim_{N\to\infty}P_{N,l}=I$ strongly and
the inclusion $H\hookrightarrow W$ is a Hilbert-Schmidt operator.
By taking $N$ sufficiently large
and using 
Lemma~\ref{gagliard-nirenberg estimate}, 
we have for $h\in H$ with $\|h\|_H\le R_0$
$$
|V(h)-V(P_{N,l}h)|\le C(1+R_0)^{2M-1}\delta,
~~|u(h)-u(P_{N,l}h)|
\le C\delta R_0.
$$
Let $h\in B_{\ep}({\cal Z})^c\cap B_{R_0,H}(0)$.
There are two cases where
(i)~for some $i$, $P_{N,l}h\in B_{\ep/2}(h_i)$,
(ii)~$P_{N,l}h\in B_{\ep/2}({\cal Z})^c$.
Let us consider the case (i).
We estimate the quantity on the right-hand side
of (\ref{identity ARl}).
\begin{eqnarray}
\lefteqn{\frac{1}{4}
\|A_{N,l}h\|_H^2+
V(P_{N,l}h)-u(P_{N,l}h)}\nonumber\\
& &=
\frac{1}{4}\|A_{N,l}h\|_H^2+
V(\psi_{Nl}(A)P_{N,l}h)-u(\psi_{Nl}(A)P_{N,l}h)\nonumber\\
& &~+
V(P_{N,l}h)-V(\psi_{Nl}(A)P_{N,l}h)
-\left(u(P_{N,l}h)-u(\psi_{Nl}(A)P_{N,l}h)\right).
\end{eqnarray}
We have
\begin{eqnarray}
\frac{1}{4}\|A_{N,l}h\|_H^2&=&
\frac{1}{4}\|A_{l}^{(Nl)}P_{N,l}h\|^2\nonumber\\
&\ge&
\frac{1}{4}\|A^{(Nl)}P_{N,l}h\|_H^2-
\frac{l}{4}\|P_{N,l}h\|_H^2\nonumber\\
&\ge&
\frac{1}{4}\|A\psi_{Nl}(A)P_{N,l}h\|_H^2-
\frac{l}{4}R_0^2.
\end{eqnarray}
By (\ref{psiRP}) and a similar proof to (\ref{psiperp}),
we obtain
\begin{equation}
\|P_{N,l}h-\psi_{Nl}(A)P_{N,l}h\|_H^2
\le \frac{4(|\kappa|+\ep)}{\omega(Nl)}.
\end{equation}
The estimate $\|P_{N,l}h-h_i\|_W\le \frac{\ep}{2}$ implies
$\|P_{N,l}^{\perp}h\|_H\ge C\|P_{N,l}^{\perp}h\|_W
\ge C\ep/2$.
Consequently,
\begin{eqnarray}
\lefteqn{\frac{1}{4}\|\sqrt{m}P_{N,l}^{\perp}h
+A_{N,l}h\|_H^2+
V(h)-u(h)}\nonumber\\
& &\ge \frac{m C^2\ep^2}{16}-\frac{l}{4}R_0^2-
C(1+R_0)^{2M-1}\delta
-C\delta R_0\nonumber\\
& &\quad-C(1+R_0)^{2M-1}
\left(\frac{4(|\kappa|+\ep)}{\omega(Nl)}\right)^{1/2}-
2C\left(\frac{|\kappa|+\ep}{\omega(Nl)}\right)^{1/2}.
\end{eqnarray}
 which proves
(\ref{lower bound ARNl}).
It remains to consider the case (ii).
In this case,
$$
\frac{1}{4}\|\sqrt{m}P_{N,l}^{\perp}h+A_{N,l}
h\|_H^2+
V(h)-u(h)\ge \delta\left(\frac{\ep}{2}\right)'-(C+1)
\delta (1+R_0)^{2M-1}.
$$
This completes the proof.
\end{proof}

\section{Proof of Theorem~\ref{main theorem 0} and
Theorem~\ref{main theorem 1}}

\begin{proof}
[Proof of Theorem~$\ref{main theorem 1}$]
(1)~Lower bound estimate:
To prove the inequality ${\rm LHS}\ge {\rm RHS}$
in (\ref{main 1}), 
we divide the estimate 
into two parts:
(I)~Neighborhood of the zero points of $U$,
(II)~Outside neighborhood of the zero points of $U$.

Let $\chi$ be a cut-off function as in Lemma~\ref{Laplace}.
Let $\ep>0$ and
$\chi_i(w)=\chi\left(\frac{\|\left(w-\sqrt{\la}h_i\right)
\|_W^2}
{\ep^2\la}\right)$
and $\chi_{\infty}(w)=\sqrt{1-\sum_{i=1}^{n_0}\chi_i(w)^2}$.
Let $f_{\ast}(w)=f(w)\chi_{\ast}(w)$, where
$\ast=i,\infty$~$(1\le i\le n_0)$.
Then 
\begin{eqnarray}
\left((-L_A+V_{\la}-u_{\la})f,f\right)&=&
\sum_{\{\ast=1,\ldots,n_0,\infty\}}
\left((-L_A+V_{\la}-u_{\la})f_{\ast},f_{\ast}\right)
\nonumber\\
& &-\sum_{\{
\ast=1,\ldots,n_0,\infty\}}\int_W\|D_A\chi_{\ast}\|_H^2f(w)^2d\mu(w).
\label{IMS}
\end{eqnarray}
By Lemma~\ref{smooth function on W},
there exists a positive constant $C$ such that
$
\|D_A\chi_{\ast}(w)\|_H^2\le \frac{C}{\ep^2\la}
$~$\mu$-$a.s.$~$w$
for all $\ast$.
First, we consider the case where $\ast=1,\ldots,n_0$.

\noindent
(I)~Neighborhood of the zero points of $U$:~
Let $1\le i\le n_0$. Using the Cameron-Martin formula,
\begin{eqnarray}
\lefteqn{\left((-L_A+V_{\la}-u_{\la})f_i,f_i\right)}\nonumber\\
& &=
\int_W\|(D_Af_i)(w+\sqrt{\la}h_i)\|_H^2\exp\left(
-\sqrt{\la}(h_i,w)_H-\frac{\la}{2}\|h_i\|_H^2
\right)d\mu\nonumber\\
& &+\int_W
\left(V_{\la}\left(w+\sqrt{\la}h_i\right)
-u_{\la}\left(w+\sqrt{\la}h_i\right)
\right)f_i(w+\sqrt{\la}h_i)^2\nonumber\\
& &\qquad \exp\left(-\sqrt{\la}(h_i,w)_H-\frac{\la}{2}\|h_i\|_H^2
\right)d\mu.
\end{eqnarray}
Let $\bar{f}_i(w)=f_i(w+\sqrt{\la}h_i)
\exp\left(-\frac{\sqrt{\la}}{2}(h_i,w)_H-\frac{\la}{4}\|h_i\|_H^2
\right).
$
Note that $\|\bar{f}_i\|_{L^2(\mu)}=\|f_i\|_{L^2(\mu)}$.
Using the integration by parts formula, we have
\begin{eqnarray}
\lefteqn{\int_W\|(ADf_i)(w+\sqrt{\la}h_i)\|_H^2\exp\left(
-\sqrt{\la}(h_i,w)_H-\frac{\la}{2}\|h_i\|_H^2
\right)d\mu}\nonumber\\
& &=
\int_W\left\|A\left(D\bar{f}_i(w)+\frac{\sqrt{\la}}{2}h_i\bar{f}_i(w)\right)
\right\|_H^2d\mu\nonumber\\
& &=\int_W\|(AD\bar{f}_i)(w)\|_H^2d\mu
+\sqrt{\la}\int_W
\left(A^2h_i,w\right)_H\frac{\bar{f}_i(w)^2}{2}d\mu
+\frac{\la}{4}\int_W\|Ah_i\|_H^2\bar{f}_i(w)^2d\mu.\nonumber
\end{eqnarray}
Also note that
\begin{eqnarray}
V_{\la}\left(w+\sqrt{\la}h_i\right)
&=&\la\int_{\RR}P(h_i(x))g(x)dx+
\sqrt{\la}\int_{\RR} P'(h_i(x))w(x)g(x)dx
+\int_{\RR}:w(x)^2:v_i(x)dx\nonumber\\
& &+\sum_{k=3}^{2M}\la^{1-\frac{k}{2}}\int_{\RR}
:w(x)^k:
\frac{P^{(k)}(h_i(x))}{k!}g(x)dx,\\
-u_{\la}(w+\sqrt{\la}h_i)&=&-\ep_i\|w\|_W^2
(=:\langle J_iw,w\rangle:+\tr\, J_i)
\qquad \mbox{for $w$ with $\chi_i(w)\ne 0$}.
\end{eqnarray}
By the Euler-Lagrange equation, we have
$\frac{1}{2}\left(A^2h_i,w\right)_H+\int_{\RR}
P'(h_i(x))w(x)g(x)dx=0$
~$\mu$-a.s.~$w$.
By this and $U(h_i)=\frac{1}{4}\|Ah_i\|_H^2+\int_{\RR}P(h_i(x))g(x)dx=0$,
we have
\begin{eqnarray}
\left((-L_A+V_{\la}-u_{\la})f_i,f_i\right)
&=&
\int_W\|AD\bar{f}_i(w)\|^2d\mu
+\int_W\left(
Q_{v_i,J_i}(w)+\tr J_i\right)\bar{f}_i(w)^2d\mu
\nonumber\\
& &+\int_W
R_{\la,i}(w)\bar{f}_i(w)^2d\mu,\label{unitary transformation}
\end{eqnarray}
where
\begin{equation}
R_{\la,i}(w)
=\sum_{k=3}^{2M}\la^{1-\frac{k}{2}}\int_{\RR}
:w(x)^k:g_{k,i}(x)dx
\label{remainder term}
\end{equation}
and
$g_{k,i}(x)=\frac{P^{(k)}(h_i(x))}{k!}g(x)$.
By Lemma~\ref{NGS estimate}~(3),
setting $\tilde{V}=R_{\la,i}$,
we obtain
\begin{eqnarray}
\lefteqn{
\left((-L_A+V_{\la}-u_{\la}-E_i)f_i,f_i\right)_{L^2(\mu)}}
\nonumber\\
& &\ge
-\frac{mc_{v_{i},J_i}}{2}
\log\left(
\int_W\exp\left(
-\frac{2}{mc_{v_i,J_i}}\tilde{R}_{\la,i}(w)\chi_{\ep,\la}(w)
\right)\Omega_{v_i,J_i}(w)^2d\mu(w)
\right)\|\bar{f}_i\|_{L^2(\mu)}^2,\nonumber\\
& &\label{intermediate-estimate1}
\end{eqnarray}
were $\chi_{\ep,\la}(w)=\chi\left(\frac{\|w\|_W^2}
{3\ep^2\la}\right)$.
By Lemma~\ref{Laplace} and using 
the same argument as in page 3363--3364 in 
\cite{aida2},
we have
$$
\liminf_{\la\to\infty}
\left((-L_A+V_{\la}-u_{\la}-E_i)f_i,f_i\right)_{L^2(\mu)}\ge 0.
$$

\noindent
(II)~Outside neighborhood of the zero points of $U$:~
We estimate $\left((-L_A+V_{\la}-u_{\la})f_{\infty},f_{\infty}\right)$.
To this end, let
$\bar{\chi}_i(w)=\chi\left(\frac{3\|w-\sqrt{\la}h_i\|_W^2}
{\ep^2\la}\right)$
and
$\bar{\chi}_{\infty}(w)=\sqrt{1-\sum_{i=1}^{n_0}
\bar{\chi}_i(w)^2}$.
$\bar{\chi}_{\infty}$ satisfies that $\bar{\chi}_{\infty}(w)=1$
for $w$ with $\chi_{\infty}(w)\ne 0$ and
$$
\left\{w\in W~|~\bar{\chi}_{\infty}(w)\ne 0\right\}
\subset
\left(\cup_{i=1}^{n_0}B_{\ep\sqrt{\frac{\la}{3}}}\left(\sqrt{\la}h_i\right)
\right)^c.
$$
Let $\ep'<\frac{\ep}{\sqrt{3}}$.
For this $\ep'$, we choose a number $N,l$
as in Lemma~\ref{second approximation}~(2)
and define a trace class operator on $H$ by
$$
\TNl=\frac{A_{N,l}}{\sqrt{m}}-P_{N,l}.
$$
We have
\begin{eqnarray}
\lefteqn{\left((-L_{A}+V_{\la}-u_{\la})
f_{\infty},f_{\infty}\right)}\nonumber\\
& &\ge
m\int_W\|(I+\TNl)Df_{\infty}(w)\|_H^2d\mu(w)\nonumber\\
& &\qquad +
\int_W\left(V_{\la}(w)-u_{\la}(w)-\frac{1}{2}\la\delta(\ep')\right)
\bar{\chi}_{\infty}(w)f_{\infty}(w)^2d\mu(w)\nonumber\\
& &\qquad +\int_W\frac{1}{2}\la\delta(\ep')\bar{\chi}_{\infty}(w)
f_{\infty}(w)^2d\mu(w).
\end{eqnarray}
Note that
$$
\int_W\frac{1}{2}\la\delta(\ep')\bar{\chi}_{\infty}(w)f_{\infty}(w)^2
d\mu(w)
=\frac{1}{2}\la\delta(\ep')\|f_{\infty}\|_{L^2(\mu)}^2.
$$
Let $\tilde{V}_{\la}(w)=\left(V_{\la}(w)-u_{\la}(w)
-\frac{1}{2}\la\delta(\ep')\right)
\bar{\chi}_{\infty}(w)$.
Applying Lemma~\ref{NGS estimate}~(1),
\begin{eqnarray}
\lefteqn{J_2(\la)=m\int_W\|(I+\TNl)Df_{\infty}(w)\|_H^2d\mu(w)+
\int_W\tilde{V}_{\la}(w)
f_{\infty}(w)^2d\mu(w)}\nonumber\\
& &\ge
-\frac{m}{2}\log
\left\{
\int_W\exp\left(-\frac{2}{m}\tilde{V}_{\la}(w)
-\left(\TNl w,w\right)_H
-\frac{1}{2}\|\TNl w\|_H^2
\right)
d\mu(w)
\right\}\|f_{\infty}\|_{L^2(\mu)}^2\nonumber\\
& &+\left(\frac{m}{2}\log\det(I+\TNl)
-\frac{m}{2}
\tr\left(\TNl^2\right)-m~ \tr(\TNl)\right)
\|f_{\infty}\|_{L^2(\mu)}^2.
\label{outside estimate}
\end{eqnarray}
Because
\begin{eqnarray}
\frac{m}{4}\|(I+\TNl)h\|_H^2+
\left(V(h)-u(h)-\frac{1}{2}\delta(\ep')\right)\tilde{\chi}_{\infty}(h)\ge
0\qquad \mbox{for all $h\in H$},
\end{eqnarray}
where 
\begin{equation}
\tilde{\chi}_{\infty}(h)=
\left(1-\sum_{i=1}^{n_0}\chi\left(
\frac{3\|h-h_i\|_W^2}{\ep^2}\right)^2\right)^{1/2},
\end{equation}
by the large deviation estimate, 
we obtain for any
$\ep''>0$ it holds that
$$
J_2(\la)\ge 
(-\ep''\la+C_m)\|f_{\infty}\|_{L^2(\mu)}^2
\qquad \mbox{for large $\la$}.
$$
Putting the above estimates together, we complete the proof
of lower bound estimate.

\noindent
(2)~Upper bound estimate:~
In (\ref{unitary transformation}),
putting $\bar{f}_i(w)=\Omega_{v_i,J_i}(w)$
and using 
$$
\lim_{\la\to\infty}\int_WR_{\la,i}(w)\Omega_{v_i,J_i}(w)^2d\mu(w)=0,
$$
we obtain the upper bound estimate.
\end{proof}

\section{Proof of Theorem~\ref{tunneling}}
\label{Proof of tunneling}

In this section, we prove Theorem~\ref{tunneling}.
Before doing so, let us recall 
the result
in the case of Schr\"odinger operator
$-H_{\la,U}=-\Delta+\la U(x/\sqrt{\la})$ in $L^2(\RR^d,dx)$
and we sketch an idea of the proof of
Theorem~\ref{tunneling}.
Let us put standard assumptions on the
potential function $U$ on $\RR^d$ as 
in (H1), (H2), (H3) in Section 2 and
\begin{enumerate}
\item[(H4)] $U(x)=U(-x)$ for all $x$ and
the zero points of $U$ consists two points.
\end{enumerate}
Then the gap of the spectrum of the lowest eigenvalue 
and the second lowest eigenvalue is exponentially small
under $\la\to\infty$ and the exponential decay rate
is given by the Agmon distance between two zero points
of $U$.
One of the key of the proof of this result
is that the operator $-\Delta$ is bounded from below
which is obtained by subtracting
the potential term $\la U(x/\sqrt{\la})$ from the
Schr\"odinger operator $-H_{\la,U}$.
In the case of $-L_A+V_{\la}$, although it is formally written as
in (\ref{formal representation}), we cannot do the same thing.
However, the bottom of spectrum
of $-L_A+V_{\la}-u_{\la}$ is uniformly bounded from below
for large $\la$ if $u\in \FWU$.
So we can apply the standard argument to 
the operator $-L_A+V_{\la}$
by replacing $-\Delta$ and $U_{\la}$
by $-L_A+V_{\la}-u_{\la}$ and $u_{\la}$ respectively.
Therefore
we will introduce distance functions using
$u\in \FWU$ by which
we can give estimates 
for the decay rate.
After that, we optimize the estimates and we
arrive at the desired estimate in 
Theorem~\ref{tunneling}.
So, first, we introduce the following.

\begin{defin}\label{modified agmon 1}~
Let $\varphi,\psi\in L^2(\RR)$.
Let
$AC_{T,\varphi,\psi}\left(L^2(\RR)\right)$ be the set of
all absolutely continuous
functions
$c : [0,T]\to L^2(\RR)$ 
with $c(0)=\varphi$ and $c(T)=\psi$.
We may omit the subscript $T$ when $T=1$
and omit denoting $\varphi,\psi$ if there are no constraint.
Let $u$ be a non-negative bounded continuous function on $W$.
For $w_1,w_2\in W$ with $w_2-w_1\in L^2(\RR)$, define
\begin{multline}
\rhoWu(w_1,w_2)
=
\inf\Bigl\{
\int_{0}^T\sqrt{u(w_1+c(t))}\|c'(t)\|_{L^2}dt~
\Big |~c\in AC_{T,0,w_2-w_1}(L^2(\RR))\Bigr\}.
\phantom{llllllllllllllllllll}\label{rho Agmon distance}
\end{multline}
If
$w_1-w_2\notin L^2(\RR)$,
we set $\rho^{W}_u(w_1,w_2)=\infty$.
\end{defin}

The definition of $\rho_u^W$ does not depend on $T$.
Clearly, $\rho^W_u(w,w+\varphi)\le \sqrt{\|u\|_{\infty}}\|\varphi\|_{L^2}$
for any $w\in W, \varphi\in L^2$.
We define an approximate Agmon distance.

\begin{defin}[Approximate Agmon distance]\label{approximate Agmon distance}
Let $u\in \FWU$ and $w_1,w_2\in W$.
Define
\begin{eqnarray}
\underline{\rho}^W_{u}(w_1,w_2)
&=&
\lim_{\ep\to 0}\inf\Biggl\{\rho^{W}_u(w,\eta)~\Big |
~w\in B_{\ep}(w_1), \eta\in B_{\ep}(w_2)\Biggr\}.
\end{eqnarray}
Using $\underline{\rho}^W_u$, we define
\begin{eqnarray}
d^W_U(w_1,w_2)
&=&\sup_{u\in \FWU}
\underline{\rho}^W_u(w_1,w_2).
\label{dWU}
\end{eqnarray}
\end{defin}

\begin{rem}\label{remark on distance}
$(1)$~Assume $U$ satisfies {\rm (A1), (A2), (A3)}.
We show that
$d^W_U(h_0,-h_0)>0$.
For sufficiently small positive $\kappa$ and $R$,
$u=\kappa \min(\uZ, R)\in \FWU$.
Note that
$$
\inf\{\|w_1-w_2\|_{W}~|~w_1\in B_{(\ep/\kappa)^{1/2}}(h_0), 
w_2\in B_{(\ep/\kappa)^{1/2}}(-h_0)\}\ge 2(\|h_0\|_W-(\ep/\kappa)^{1/2})>0,
$$
and $L^2$ norm is stronger than the norm of $\|~\|_W$, we have
$\underline{\rho}^W_{u}(h_0,-h_0)>0$ and
$d^W_U(h_0,-h_0)>0$.
Also it is obvious that $d_U^W(h_0, -h_0)\le \dUAg(h_0,-h_0)$.

\noindent
$(2)$~Let us consider the case where $U(h)=U(-h)$.
Take $u\in \FWU$.
Then $v(w)=\max(u(w),u(-w))$ also belongs to
$\FWU$.
So, the value of $\dUAg(w_1,w_2)$ in $(\ref{dWU})$
does not change by restricting the domain $\FWU$
to the proper subset consisting of
such symmetric functions.
\end{rem}

Note that for any $\eta\in W$,
$$
\rho^W_u(\eta,w)=+\infty
\qquad \mu\mbox{-}a.s.~w
$$
So, still, we cannot argue as finite dimensional cases
and we need some preliminaries to prove main theorem.
For a non-empty open set $O$ of $W$, let 
$$
\rho_{u}^W(w,O)=
\inf\{\rhoWu(w,\phi)~|~\phi\in O\}.
$$

\begin{lem}\label{rho distance}
Let $u$ be a bounded continuous 
function on $W$.
Let $O$ be a non-empty open set.
We have $\rhoWu(w,O)<\infty$ for all $w\in W$
and
the function
$w\mapsto \rhoWu(w,O)$ is a Borel measurable function.
Let $\la>0$.
Let us write $\rho_{\la,O}(w)=\rhoWu(w/\sqrt{\la},O)$ 
for simplicity.
Then the function $\rho_{\la,O}$
belongs to ${\rm D}(\E_A)$
and
\begin{equation}
\|D_A\rho_{\la,O}(w)\|_H^2
\le \frac{u(w/\sqrt{\la})}{\la},\qquad\quad
\mbox{$\mu$-a.s.~$w$.} 
\end{equation}
\end{lem}

The proof of this lemma is a suitable modification
of that of Lemma~3.2 in \cite{aida5}.

\begin{proof}
For any $w\in W$, there exists $h\in H$ such that
$w+h\in O$ which implies $\rhoWu(w,O)<\infty$.
The measurability follows from the same argument as in the proof of
Lemma~3.2 in \cite{aida5}.
We prove the latter half of the statement.
We prove the estimate in the case where
$\la=1$.
The proof of other cases is similar to it.
Let $C=\|u\|_{\infty}$.
Let $h\in H$.
By the definition of $\rhoWu$,
\begin{eqnarray}
\rhoWu(w+h,O)&\le&
\rhoWu(w,O)+\int_0^1\sqrt{u(w+th)}\|h\|_{L^2}dt\nonumber\\
&\le&\rhoWu(w,O)+\sqrt{C}
\|h\|_{L^2}.\label{H derivative}
\end{eqnarray}
This shows $\rhoWu(w,O)$ is almost surely 
$H$-Lipschitz continuous function on $W$.
By 5.4.10. Example in \cite{bogachev},
$\rhoWu$
belongs to ${\rm D}(\E_I)$ and
$\|D\rhoWu(w,O)\|_H\le \sqrt{C}$
~for $\mu$-almost all $w$.
Actually, (\ref{H derivative}) shows
for any $h\in H$,
\begin{equation}
\left(D\rhoWu(w,O),h\right)_H\le 
\sqrt{u(w)}\|A^{-1}h\|_H.
\end{equation}
This shows 
$\|D_A\rhoWu(w,O)\|_H\le \sqrt{u(w)}$~$\mu$-almost all $w$.
\end{proof}

\begin{rem}
Let $d_H(w,\eta)=\|w-\eta\|_H$.
This function $d_H$ is so-called an $H$-distance on $W$
and $\rho^H(O,w)=\inf\{d_H(w,\eta)~|~\eta\in O\}$
belongs to ${\rm D}(\E_I)$ for any non-empty open set $O$ in $W$.
The definition of ${\rm D}(\E_I)$ was given in
Definition~$\ref{Dirichlet form}$.
The topology defined by the Agmon distance $\dUAg$ on $H^{1/2}(\RR)(=H)$
is nothing but the topology of the Sobolev space $H^{1/2}(\RR)$.
See Theorem~$\ref{Theorem 1 for Agmon}$.
Approximate Agmon distance $d_U^W$ may be viewed as
an extension of $d_U^{Ag}$ on $W$ similarly to
$d_H$ in view of Lemma~$\ref{Agmon and dWU}$. However I think
$d_U^W(w,\eta)=+\infty$ if $w\notin H^{1/2}$ or
$\eta\notin H^{1/2}$
differently from $d_H$.
\end{rem}

It is known that $E_1(\la)$ is a simple eigenvalue
and there exists an associated strictly positive
normalized eigenfunction $\Omega_{0,\la}(w)$.
Intuitively, the ground state measure
$\Omega_{0,\la}(w)^2d\mu(w)$ concentrates on a certain neighborhood
of 
$\sqrt{\la}h_0,-\sqrt{\la}h_0$ when $\la$ is large.
We need such an estimate to obtain our second main theorem.

\begin{lem}\label{Agmon estimate}~
Let $0<q<1$. Let $\xi$ be a globally Lipschitz continuous function 
such that the support of
the first derivative of $\xi$ is compact.
Let $u\in \FWU$ and set
\begin{equation}
\eta(w)=\xi(\rhoWu(w,B_{\ep}({\cal Z}))),\quad
\rho_{u}(w)=\rhoWu\left(w,B_{\ep}({\cal Z})\right).
\end{equation}
Then 
\begin{eqnarray}
\lefteqn{\int_W\left\{
\la(1-q^2)u(w/\sqrt{\la})-(C_u+E_1(\lambda))
\right\}e^{2 \la q \rho_{u}(w/\sqrt{\la})}
\eta(w/\sqrt{\la})^2
\Omega_{0,\la}(w)^2
d\mu(w)}\nonumber\\
& &\le
\frac{1}{\la}\int_We^{2 \la q\rho_{u}(w/\sqrt{\la})}
\xi'\left(\rho_{u}(w/\sqrt{\la})\right)^2u(w/\sqrt{\la})
\Omega_{0,\la}(w)^2d\mu(w)\nonumber\\
& &~~+2q\int_We^{2 \la q\rho_{u}(w/\sqrt{\la})}
\xi'\left(\rho_{u}(w/\sqrt{\la})\right)
\eta(w/\sqrt{\la})
u(w/\sqrt{\la})
\Omega_{0,\la}(w)^2d\mu(w).
\label{simon3}
\end{eqnarray}
\end{lem}

\begin{proof}
Let $F$ and $G$ be bounded $C^{\infty}$ functions on $W$.
We use the notation
$F_{\la}(w)=\la F(w/\sqrt{\la})$.
Using the lower bound
$$
\inf\sigma(-L_A+V_{\la}-u_{\la})\ge -C_{u}
\quad \mbox{for sufficiently large} \la,
$$
we have
\begin{eqnarray}
\lefteqn{\left(
e^{F_{\la}}G, \left(-L_A+V_{\la}-E_1(\la)\right)
(e^{-F_{\la}}G)
\right)
_{L^2(\mu)}}
\nonumber\\
& &=
\left((-L_A+V_{\la}-u_{\la})-E_1(\la))G,G\right)_{L^2(\mu)}
+\left((u_{\la}-\|D_AF_{\la}\|_H^2)G,G\right)_{L^2(\mu)}\nonumber\\
& &\ge 
\left(\left\{
\la\left(
u\left(\frac{w}{\sqrt{\la}}\right)-\left\|(D_AF)
\left(\frac{w}{\sqrt{\la}}\right)\right\|_H^2\right)-
(C_u+E_1(\la))\right\}G,
G\right)_{L^2(\mu)}.\label{F and G}
\end{eqnarray}
Let $\eta$ be another smooth function
and set
\begin{equation}
G=\Omega_{0,\la}e^{F_{\la}}\eta_{\la}.
\end{equation}
Using
$(-L_A+V_{\la})\Omega_{0,\la}=E_1(\la)\Omega_{0,\la}$,
the left-hand side of (\ref{F and G}) reads
\begin{eqnarray}
\lefteqn{\left(e^{2F_{\la}}\Omega_{0,\la}\eta_{\la}, 
(-L_A+V_{\la}-E_1(\la))(\Omega_{0,\la}\eta_{\la})\right)}\nonumber\\
& &=-\left(e^{2F_{\la}}\Omega_{0,\la}^2\eta_{\la},
L_A\eta_{\la}
\right)-
\left(D_A(\Omega_{0,\la}^2), e^{2F_{\la}}\eta_{\la} D_A\eta_{\la}
\right)
\nonumber\\
& &=
\int_We^{2F_{\la}}\|D_A\eta_{\la}\|_H^2\Omega_{0,\la}^2d\mu
+2\int_We^{2F_{\la}}
(D_AF_{\la},D_A\eta_{\la})_{H}\eta_{\la}\Omega_{0,\la}^2d\mu.
\end{eqnarray}
Consequently, we obtain
\begin{multline}
\int_W
\left\{
\la\left(
u\left(\frac{w}{\sqrt{\la}}\right)-\left\|(D_AF)
\left(\frac{w}{\sqrt{\la}}\right)\right\|_H^2\right)-
(C_u+E_1(\la))\right\}
e^{2F_{\la}(w)}\eta_{\la}(w)^2\Omega_{0,\la}(w)^2d\mu(w)\\
\le
\int_We^{2F_{\la}}\|D_A\eta_{\la}\|_H^2\Omega_{0,\la}^2d\mu
+2\int_We^{2F_{\la}}
(D_AF_{\la},D_A\eta_{\la})_{H}\eta_{\la}\Omega_{0,\la}^2d\mu.
\end{multline}
We apply this estimate in the case where
\begin{eqnarray}
F(w)&=& q\rhoWu(w,B_{\ep}({\cal Z})),\\
\eta(w)&=&\xi(\rhoWu(w,B_{\ep}({\cal Z}))).
\end{eqnarray}
These functions does not satisfy the assumptions we assume so far.
But standard approximation argument works and
we complete the proof.
\end{proof}

Let $\kappa$ be a positive number and 
$
\xi(t)
$
be the piecewise linear function such that
$\xi(t)=0$ for $t\le 0$
and $\xi(t)=1$ for $t\ge \kappa$.
Then we obtain an exponential decay estimate
of the ground state measure.

\begin{lem}\label{Agmon estimate 2}~
Let $r>\kappa$ and $0<q<1$.
For large $\la$, we have
\begin{eqnarray}
\int_{\rhoWu(w/\sqrt{\la},B_{\ep}({\cal Z}))\ge r}
\Omega_{0,\la}(w)^2d\mu(w)
&\le&
C_1\frac{e^{-2q\la (r-\kappa)}}
{\kappa^2(\la(1-q^2)\ep^2-C_2)}
\|u\|_{\infty},
\end{eqnarray}
where $C_i$ are positive constants
independent of $\la$.
\end{lem}

We write $\mu_{\la,U}=\Omega_{0,\la}(w)^2d\mu(w)$.
Let $S_{\la} : w\mapsto w/\sqrt{\la}$ be the scaling map.
Then the above lemma shows that the image measure 
$(S_{\la})_{\ast}\mu_{\la,U}$ concentrates on a neighborhood of
zero points $\{h_0,-h_0\}$ of the potential function $U$.
Now we are going to prove second main theorem.
As the first step,
we prove the following.

\begin{lem}
Assume the same assumptions as in Theorem~$\ref{tunneling}$.
Then we have
\begin{equation}
\limsup_{\lambda\to\infty}\frac{\log\left(E_2(\lambda)-E_1(\lambda)\right)}
{\lambda}\le -d^{W}_U(h_0,-h_0).\label{tunneling upper bound dWU}
\end{equation}
\end{lem}

\begin{proof}
Note that
\begin{eqnarray}
\lefteqn{E_2(\la)-E_1(\la)}\nonumber\\
& &=
\inf\Biggl\{
\frac{\int_W\|D_Af(w)\|_H^2d\mu_{\la,U}}
{{\rm Var}_{\mu_{\la,U}}(f)}
~\Bigg |~f\not\equiv {\rm const},~
f\in {\rm D}({\cal E}_A)\cap L^{\infty}(W,\mu)~\mbox{with}~
\nonumber\\
& &\qquad\qquad
\int_W\|D_Af(w)\|_H^2d\mu_{\la,U}<\infty
\Biggr\},
\end{eqnarray}
where ${\rm Var}_{\mu_{\la,U}}$
stands for the variance with respect to
the ground state measure $\mu_{\la,U}$.
To prove this result, we need to identify
the domain of the Dirichlet form which is obtained by
the ground state transformation by $\Omega_{0,\la}$.
We refer the reader to \cite{aida6} for this problem in a
setting of hyperbounded semi-group.
By taking a trial function $f$ which
satisfies the assumption of the right-hand side of the
above, we prove (\ref{tunneling upper bound dWU}).
Let $u\in \FWU$ which satisfies $\underline{\rho}^W_u(h_0,-h_0)>0$.
Without loss of generality, we may
assume that $u(w)=u(-w)$ for all $w$
because of Remark~\ref{remark on distance}~(2).
Take $\delta>0$ such that
$0<\delta<\frac{\underline{\rho}^W_u(h_0,-h_0)}{4}$.
Let $\psi_{\delta}$ be the piecewise linear function such that
$
\psi_{\delta}(t)=1
$ 
for 
$t\le \frac{\underline{\rho}_{u}^W(h_0,-h_0)}{2}-2\delta$
and
$
\psi_{\delta}(t)=0
$
for 
$t\ge \frac{\underline{\rho}_{u}^W(h_0,-h_0)}{2}-\delta$.
Let 
\begin{equation}
f_{\delta}(w)=
\psi_{\delta}
\left(\rhoWu\left(\frac{w}{\sqrt{\la}}, B_{\ep}(h_0)\right)\right)-
\psi_{\delta}\left(\rhoWu\left(\frac{w}{\sqrt{\la}}, B_{\ep}(-h_0)\right)
\right).
\end{equation}
Let $\ep$ be a sufficiently small positive number such that
\begin{eqnarray}
\rhoWu(w,B_{\ep}(h_0))+
\rhoWu(w, B_{\ep}(-h_0))&>&
\underline{\rho}^W_u(h_0,-h_0)-2\delta\qquad
\mbox{for all $w\in W$}.\label{rho h0}
\end{eqnarray}
We can choose such a number because of the definition of
$\underline{\rho}^W_u$ and
the triangle inequality for
$\rho^W_u$.
Let 
$\kappa$ be a positive number such that
\begin{equation}
\frac{\underline{\rho}_{u}^W(h_0,-h_0)}{2}-2\delta>
\kappa.
\end{equation}
Since
\begin{eqnarray}
{\rm Var}_{\mu_{\la,U}}(f_{\delta})
&\ge& 2\mu_{\la,U}
(f_{\delta}=1)
\mu_{\la,U}
(f_{\delta}=-1)\nonumber\\
&=& 2\mu_{\la,U}
\left(
\rhoWu\left(\frac{w}{\sqrt{\la}},B_{\ep}(h_0)\right)
\le \frac{\underline{\rho}_{u}^W(h_0,-h_0)}{2}-2\delta
\right)^2\label{symmetry of u}\\
&=&\frac{1}{2}\mu_{\la,U}
\left(
\rhoWu\left(\frac{w}{\sqrt{\la}},B_{\ep}({\cal Z})\right)\le 
\frac{\underline{\rho}_{u}^W(h_0,-h_0)}{2}-2\delta
\right)^2,\label{rho distance}
\end{eqnarray}
we obtain
\begin{equation}
\liminf_{\la\to\infty}
{\rm Var}_{\mu_{\la,U}}(f_{\delta})>0.
\end{equation}
We have used (\ref{rho h0}) and the symmetry, {\it i.e.}, 
$u(w)=u(-w)$,
$\Omega_{0,\la}(w)=\Omega_{0,\la}(-w)$ for all $w\in W$
in (\ref{symmetry of u}).
Also we have used (\ref{rho h0}) in (\ref{rho distance}).
On the other hand,
\begin{eqnarray}
\int_W\|D_Af_{\delta}(w)\|_H^2d\mu_{\la,U}(w)
&\le&
\frac{C}{\la\delta^2}\int_{D_{\delta}}
u\left(\frac{w}{\sqrt{\la}}\right)
d\mu_{\la,U}(w)\nonumber\\
& &\le \frac{C_1\|u\|_{\infty}}{\la\delta^2\kappa^2
(\la(1-q^2)\ep^2-C_2)}
e^{-q\la\Bigl(\underline{\rho}^W_u(h_0,-h_0)-4\delta-2\kappa\Bigr)},
\end{eqnarray}
where
\begin{eqnarray}
D_{\delta}
&=&
\left\{w\in W~\Big |~
\rhoWu\left(\frac{w}{\sqrt{\la}},B_{\ep}({\cal Z})\right)\in
\left[\frac{\underline{\rho}^W_u(h_0, -h_0)}{2}-2\delta,~
\frac{\underline{\rho}^W_u(h_0,-h_0)}{2}-\delta\right]\right\}.
\end{eqnarray}
Thus by optimizing $\underline{\rho}^W_u(h_0,-h_0)$,
this completes the proof.
\end{proof}

Now we complete the proof of 
Theorem~\ref{tunneling}.
It suffices to prove the following lemma.

\begin{lem}\label{Agmon and dWU}
Let us consider the situation in Theorem~
$\ref{tunneling}$.
Then we have
\begin{equation}
\dUAg(h_0,-h_0)=d^{W}_U(h_0,-h_0).
\end{equation}
\end{lem}

From now on, until the end of this section, we assume
that $U$ satisfies (A1), (A2) and (A3).
We need preparations for the proof of this lemma.
Let $I=[-L/2,L/2]$.
Let $\Delta_D=\frac{d^2}{dx^2}$ be the Laplace-Beltrami operator on
$L^2(I,dx)$ with Dirichlet boundary condition,
where $dx$ denotes the Lebesgue measure.
Set 
$e_{2k}(x)=\sqrt{\frac{2}{L}}\sin\left(\frac{2k\pi}{L}x\right)$
~$(k=1,2,\ldots)$,
$e_{2k+1}(x)=\sqrt{\frac{2}{L}}\cos\left(\frac{(2k+1)\pi }{L}x\right)$
~$(k=0,1,\cdots)$.
Then 
$\{e_n\}_{n\ge 0}$ 
is a complete orthonormal system of $L^2(I,dx)$.
We define Sobolev spaces:
\begin{eqnarray*}
H_0^{s}(I,dx)&=&\left\{h\in {\mathcal D}'(I)~\Big |~
h=\sum_{n\ge 0}a_ne_n, ~\mbox{and}~
\|h\|_{H^s_0(I)}^2:=\sum_{n\ge 0}\omega(n)^{2s}|a_n|^2<\infty
\right\},
\end{eqnarray*}
where $\omega(n)=\left(m^2+\left(\frac{n\pi}{L}\right)^2\right)^{1/2}$
and $s\in \RR$.
Clearly $\|\varphi\|_{H_0^s}=
\|(m^2-\Delta_D)^{s/2}\varphi\|_{L^2}$.
We consider projection operators
$P_{N}h=\sum_{0\le n\le N}a_ne_n$ on $H_0^s$.
Below, we denote the set of $C^{\infty}$ functions with compact 
support on $\RR$ by $C^{\infty}_0(\RR)$.

\begin{lem}\label{multiplication}
Let $I=[-L/2,L/2]$.

\noindent
$(1)$~Let $\chi\in C^{\infty}_0(\RR)$ and assume
the support of $\chi$ is included in the open interval $(-L/2,L/2)$.
Let $M_{\chi}$ be the multiplication operator 
defined by $M_{\chi}h=\chi\cdot h\in C^{\infty}(I)$,
where $h\in C^{\infty}_0(\RR)$.
Then $M_{\chi}$ can be extended to a bounded linear operator
from $W$ to $H^{-2}_0(I)$.

\noindent
$(2)$~Let $h\in H^1_0(I,dx)$.
Define $\tilde{h}(x)=h(x)$~$(x\in I)$
and $\tilde{h}(x)=0$~$(x\in I^c)$.
Then the zero extension $\tilde{h}$ belongs to $H^1(\RR)$
and $\|\tilde{h}\|_{H^1(\RR)}=\|h\|_{H^1_0(I)}$.
\end{lem}

\begin{proof}
(1)~Let $h\in C^{\infty}_0(\RR)$
and $\varphi\in L^2(I)$.
We write $(m^2-\Delta_D)^{-1}\varphi=\psi\in H^2_0(I)$.
Then (the zero extension of) $\psi\cdot\chi$ 
belongs to $H^2(\RR)$ and
\begin{equation}
(1-\Delta) \left(\psi\cdot\chi\right)=
\varphi\cdot \chi+(1-m^2)\psi\cdot \chi
-\Delta\chi\cdot \psi-
2\psi'\chi'.\label{Delta psi}
\end{equation}
We have
\begin{eqnarray}
\int_I(m^2-\Delta_D)^{-1}(h\chi)(x)\varphi(x)dx&=&
\int_{\RR}h(x)\chi(x)(m^2-\Delta_D)^{-1}\varphi(x)dx\nonumber\\
&=&\int_{\RR}(1+x^2-\Delta)^{-1}h(x)
(1+x^2-\Delta)(\chi\psi)(x)dx\nonumber\\
&\le&\|h\|_W\|(1+x^2-\Delta)(\chi\psi)\|_{L^2(\RR)}.
\end{eqnarray}
By (\ref{Delta psi}),
we obtain
\begin{eqnarray}
\left|\int_I(m^2-\Delta_D)^{-1}(h\chi)(x)\varphi(x)dx\right|
&\le&
\|h\|_W\left(\|\varphi\|_{L^2(I)}+\|\psi\|_{L^2(I)}+
\|\psi'\|_{L^2(I)}\right)\nonumber\\
&\le& C\|h\|_W\|\varphi\|_{L^2(I)}
\end{eqnarray}
which proves the statement (1).
The result (2) is an elementary subject.
\end{proof}

Also we have

\begin{lem}\label{Sobolev on I}
Let $p\ge 2$ and  $0<\tau<1$.
Then we have the following estimates.

\noindent
$(1)$~$\|h\|_{L^p(I)}\le C(L)\|h\|_{H_0^{1/2}(I)}$.

\noindent
$(2)$~$\|h\|_{H_0^{1/2}(I)}\le \|h\|_{H_0^1(I)}^{\tau}
\|h\|_{H_0^{(1-2\tau)/(2-2\tau)}(I)}^{1-\tau}$.
\end{lem}

\begin{proof}
The statement (1) can be proved by using
an interpolation argument and the proof is 
well-known.
We prove (2).
Let $h=\sum_na_ne_n\in H_0^s$.
Then
\begin{eqnarray*}
\|\varphi\|_{H^{1/2}}^2&=&
\sum_{n}\omega(n)|a_n|^2\\
&=&\sum_{n}\omega(n)^{2\tau}|a_n|^{2\tau}\cdot 
\omega(n)^{1-2\tau}|a_n|^{2-2\tau}\\
&\le&\left(\sum_n\omega(n)^2|a_n|^2\right)^{\tau}
\left(\sum_{n}\omega(n)^{(1-2\tau)/(1-\tau)}|a_n|^2\right)^{1-\tau}
\end{eqnarray*}
which completes the proof.
\end{proof}

\begin{lem}\label{cut-off}
Let $L$ be a positive number such that
the support of $g$ is included in $[-L/8,L/8]$.
Let $\chi$ be a smooth non-negative function such that
$\chi(x)=1$ for $|x|\le 1/4$ and
$\chi(x)=0$ for $|x|\ge 1/3$.
Let $\chi_L(x)=\chi(x/L)$.
Let $-C(P)=\inf_xP(x)$.
Then for any $0<\ep<1$, by taking $L$ large enough, we have
\begin{equation}
U(h)\ge (1-\ep)U(h\chi_{L})-\ep^2C(P)\|g\|_{L^1}
\qquad \mbox{for all $h\in H^1(\RR)$}.
\label{hchiL}
\end{equation}
\end{lem}

\begin{proof}
Let $h\in H^1(\RR)$.
By using the integration by parts and a simple calculation,
\begin{eqnarray}
U(h)&=&U(\chi_Lh+(1-\chi_L)h)\nonumber\\
&=&U(\chi_Lh)+\frac{1}{4}\|(1-\chi_L)h\|^2_{H^1}+
\frac{1}{2}\int_{\RR}\chi_L(x)(1-\chi_L(x))h'(x)^2dx\nonumber\\
& &+
\frac{1}{4}\int_{\RR}
\left(2\chi_L(x)-1\right)\chi_L''(x)h(x)^2dx\nonumber\\
&\ge&U(\chi_Lh)-\frac{1}{4L^2}\|\chi''\|_{\infty}\|h\|^2_{L^2(\RR)}.
\end{eqnarray}
By the definition of $C(P)$,
$V(h)=\int_{\RR}P(h(x))g(x)\ge -C(P)\int_{\RR}g(x)dx=-C(P)\|g\|_{L^1}$.
Therefore, 
\begin{eqnarray}
\lefteqn{U(h)-(1-\ep)U(\chi_Lh)}\nonumber\\
& &\ge
\ep U(h)-(1-\ep)\frac{1}{4L^2}\|\chi''\|_{\infty}\|h\|_{L^2}^2\nonumber\\
& &=\ep\left(
U(h)-(1-\ep)\frac{1}{4L^2\ep}\|\chi''\|_{\infty}\|h\|_{L^2}^2
\right)\nonumber\\
& &=\ep\left(
(1-\ep)U(h)+\ep
U(h)-\frac{1-\ep}{4\ep L^2}\|\chi''\|_{\infty}\|h\|_{L^2}^2
\right)\nonumber\\
& &=\ep\left\{
(1-\ep)U(h)+
\left(\frac{m^2\ep}{4}-\frac{1-\ep}{4\ep L^2}\|\chi''\|_{\infty}
\right)\|h\|_{L^2}^2
-\ep C(P)\|g\|_{L^1}
\right\}.
\end{eqnarray}
Therefore,
setting
\begin{eqnarray}
L&=&\frac{\sqrt{(1-\ep)\|\chi''\|_{\infty}}}{\ep m}
\end{eqnarray}
we obtain the estimate~(\ref{hchiL}).
\end{proof}

\begin{proof}[Proof of Lemma~$\ref{Agmon and dWU}$]
Let $\ep,L$ be the positive number in 
Lemma~\ref{cut-off} and set $I=[-L/2,L/2]$.
Here we take $L$ large enough so that
\begin{equation}
\|h_i\chi_L-h_i\|_{H^1(\RR)}\le \delta/8,\label{L and h}
\end{equation}
where $h_1=h_0, h_2=-h_0$.
In this proof, we set $s_0>2$
and let $\tau_0=\frac{2s_0+1}{2s_0+2}$.
That is, the estimate
$\|h\|_{L^p}\le 
C(p,L)\|h\|_{H^1_0}^{\tau_0}
\|h\|_{H^{-s_0}}^{1-\tau_0}$ holds.
Let us use the function
$\uZ(w)=\min_{i=1,2}\|w-h_i\|_{W}^2$.
Note that there exists $0<\ep_0<1$ such that
\begin{equation}
U(h)\ge \ep_0 \uZ(h) \quad \mbox{for all $h\in H^1$}.
\end{equation}
For $w\in W$, we write $w_L=\chi_L\cdot w$ for simplicity.
This multiplication is well-defined by Lemma~\ref{multiplication}
and $w_L\in H^{-2}_0(I)$.
Let $R$ be a positive number and $N$
be a natural number.
Let us define a subset of $W$ by
\begin{eqnarray}
\lefteqn{W_{R,N,L,\delta}}\nonumber\\
& &=
\left\{w\in W~\Big |~
\|P_Nw_L\|_{H^{1/2}_0(I)}\le R, \|P_N^{\perp}w_L\|_{H^{-2}_0(I)}\le R,
\min_{i=1,2}\|P_Nw_L-h_i\|_{H^1(\RR)}\ge\delta
\right\}.\nonumber\\
& &
\end{eqnarray}
Here we identify $P_Nw_L\in H^1_0(I)$ as an element
of $H^1(\RR)$ by the zero extension.
Let $h\in H^1(\RR)\cap W_{R,N,L,\delta}$.
We have
\begin{eqnarray}
U(h_L)&=&
\frac{1}{4}\|h_L\|_{H^1(\RR)}^2+V(h_L)\nonumber\\
&=&
\frac{1}{4}\|P_Nh_L\|_{H^1_0(I)}^2+
V(P_{N}h_L)+
\frac{1}{4}\|P_{N}^{\perp}h_L\|_{H^1_0(I)}^2
+V(h_L)-V(P_Nh_L)
\nonumber\\
&=&
U(P_{N}h_L)+
\frac{1}{4}\|P_{N}^{\perp}h\|_{H^1_0(I)}^2+
V(h_L)-V(P_N h_L).
\end{eqnarray}
We have
\begin{eqnarray}
\lefteqn{V(h_L)-V(P_{N}h_L)}\nonumber\\
& &=
a_{2M}\|P_N^{\perp}h_L\|_{L^{2M}(I,gdx)}^{2M}+
a_{2M}\sum_{r=1}^{2M-1}\int_I{\binom{2M}{r}}(P_Nh_L)^{2M-r}(x)
\left(P_{N}^{\perp}h_L(x)\right)^{r}g(x)dx\nonumber\\
& &+\sum_{k=1}^{2M-1}a_k\sum_{r=1}^{k}\int_I
{\binom{k}{r}}
(P_Nh_L)^{k-r}(x)
\left(P_{N}^{\perp}h_L(x)\right)^{r}g(x)dx.
\end{eqnarray}
Let $1\le r\le 2M-1,~ r\le k\le 2M$.
For sufficiently small $\sigma>0$,
\begin{eqnarray}
\lefteqn{\left|\int_I(P_Nh_L)^{k-r}(x)
(P_N^{\perp}h_L)^r(x)g(x)dx\right|}\nonumber\\
& &\le
\|P_Nh_L\|_{L^{k}(I,gdx)}^{k-r}\|P_N^{\perp}h_L\|_{L^k(I,gdx)}^r\nonumber\\
& &\le C(L)\|P_Nh_L\|_{L^{k}(I,gdx)}^{k-r}
\|P_N^{\perp}h_L\|_{L^k(I,gdx)}^{r-\sigma}
\|P_N^{\perp}h_L\|_{H^1_0}^{\sigma\tau_0}
\|P_N^{\perp}h_L\|_{H^{-s_0}_0}^{\sigma(1-\tau_0)}.\nonumber\\
& &\le
C(L)\|P_Nh_L\|_{L^{k}(I,gdx)}^{k-r}
\|P_N^{\perp}h_L\|_{H^{-s_0}_0}^{\sigma(1-\tau_0)}
\left(
(1-\frac{\sigma\tau_0}{2})
\|P_N^{\perp}h_L\|_{L^k(I,gdx)}^{(r-\sigma)/(1-(\sigma\tau_0/2))}
+\frac{\sigma\tau_0}{2}
\|P_N^{\perp}h_L\|_{H^{1}_0}^2
\right)\nonumber\\
& &\le
C(L) (1+R)^{2M-1}\left(\omega(N+1)^{2-s_0}R\right)
^{\sigma(1-\tau_0)}
\left(1+\|P_N^{\perp}h_L\|_{L^k(I,gdx)}^{2M}
+\|P_N^{\perp}h_L\|_{H^1_0}^2\right)
\end{eqnarray}
Thus
we obtain
\begin{eqnarray}
\lefteqn{V(h_L)-V(P_{N}h_L)}\nonumber\\
& &\ge
\left(a_{2M}-C(L)(1+R)^{2M}b_N\right)
\|P_{N}^{\perp}h_L\|_{L^{2M}(I,gdx)}^{2M}
-C(L)
(1+R)^{2M}b_N(1+\|P_N^{\perp}h_L\|_{H^1_0}^2).\nonumber\\
& &
\end{eqnarray}
where $b_N=\omega(N+1)^{(2-s_0)(1-\tau_0)\sigma}$.
Hence
\begin{eqnarray}
\lefteqn{\frac{1}{4}\|P_{N}^{\perp}h_L\|_{H^1_0}^2+
V(h_L)-V(P_Nh_L)}\nonumber\\
& &\ge 
\left(\frac{1}{4}-C(L)
(1+R)^{2M}b_N\right)\|P_N^{\perp}h_L\|_{H^1_0}^2+
\left(a_{2M}-C(L)(1+R)^{2M}b_N\right)
\|P_{N}^{\perp}h_L\|_{L^{2M}(I,gdx)}^{2M}\nonumber\\
& &\qquad -C(L)
(1+R)^{2M}b_N.
\end{eqnarray}
Let
$\theta(\delta)=\inf\{U(h)~|~\min_{1\le i\le 2}
\|h-h_i\|_{H^1(\RR)}\ge \delta\}$.
Clearly $\theta(\delta)>0$.
We have
\begin{eqnarray}
U(h)&=&
\ep U(h)+(1-\ep)U(h)\nonumber\\
&\ge&
\ep\ep_0\uZ(h)+
(1-\ep)^2U(h_L)-\ep^2C(P)\|g\|_{L^1}\nonumber\\
&\ge&
\ep\ep_0\uZ(h)
+(1-\ep)^3U(P_Nh_L)+
(1-\ep)^2\ep\theta(\delta)-C(L)(1+R)^{2M}b_N
-\ep^2C(P)\|g\|_{L^1}\nonumber\\
& &+
(1-\ep)^2\left(\min\left(\frac{1}{4},a_{2M}\right)
-C(L)(1+R)^{2M}b_N\right)
\left(\|P_N^{\perp}h_L\|_{H^1_0}^2+
\|P_N^{\perp}\|_{L^{2M}(I,gdx)}^{2M}\right)\nonumber\\
& &\qquad \mbox{for any $h\in W_{R,N,L,\delta}\cap H^1(\RR)$}.
\label{estimate on U large N}
\end{eqnarray}
Thus, for fixed $\delta>0$,
take $\ep$ sufficiently small so that
\begin{equation}
(1-\ep)^2\ep\theta(\delta)-\ep^2C(P)\|g\|_{L^1}\ge
(1-\ep)^2\ep^2\theta(\delta).\label{ep and delta}
\end{equation}
Next, take $L$ large enough as in Lemma~\ref{cut-off}
and finally, taking $N$ sufficiently large, 
we get
\begin{equation}
U(h)\ge
\ep\ep_0\uZ(h)+(1-\ep)^3U(P_Nh_L)+
\frac{1}{2}(1-\ep)^3\ep^2\theta(\delta)
\qquad \mbox{for $h\in W_{R,N,L,\delta}\cap H^1(\RR)$}.
\label{estimate on U large N}
\end{equation}
Let us consider a set $\overline{W_{R,N,L,\delta}^c}$.
The closure is taken with respect to the topology of $\|~\|_W$.
It is equal to the union of
\begin{equation}
\left\{w~|~
\|P_Nw_L\|_{H^{1/2}_0(I)}\ge R\right\},~
\{w~|~\|P_N^{\perp}w_L\|_{H^{-2}_0(I)}\ge R\},~
\{w~|~
\min_{i=1,2}\|P_Nw_L-h_i\|_{H^1(\RR)}\le \delta\}.
\end{equation}
For a closed subset $F$ in $W$ in the topology which is defined by the norm
$\|~\|_W$, let $d_W(w,F)=\inf\{\|w-\eta\|_W~|~\eta\in F\}$.
Then $d_W(w,F)=0$ is equivalent to $w\in F$ and 
$w\mapsto d_W(w,F)$ is a Lipschitz
continuous function whose Lipschitz constant is less than or equal to $1$.
Let
$$
\psi(w)=
\frac{d_W(w,\overline{W_{R,N,L,\delta}^c})}{d_W(w,W_{R/2,N,L,2\delta})
+d_W(w,\overline{W_{R,N,L,\delta}^c})}.
$$
Define
\begin{eqnarray}
\lefteqn{u_{R,\ep,N,L,\delta}(w)}\nonumber\\
& &=
\left(
\ep\ep_0\min\left(\uZ(w), 2R^2\right)
+(1-\ep)^3 U(P_Nw_L)\right)
\psi(w)+
\min\left(\ep_0\uZ(w),2R^2\right)(1-\psi(w)).\nonumber\\
& &
\end{eqnarray}
Note that $L$ is a large positive number
which depends on $\ep,\delta$ and
$N$ is a large positive natural number which
depends on $\ep,\delta,L,R$.
From now on, we set
$$
\delta<\frac{1}{8}\min\left(\|h_1\|_{W}, \|h_1\|_{H^1}\right),~~~
R>8\max\left(\|h_1\|_W, \|h_1\|_{H^1}\right).
$$
Also we take $N$ sufficiently large
so that
\begin{equation}
\|P_N(h_i\chi_L)-h_i\|_{H^1(\RR)}\le \delta/4.\label{N and h}
\end{equation}
Since $w(\in W)\mapsto P_Nw_L\in H^1(I)$ is a continuous map
and $h_i\in H^1$,
for sufficiently small $\ep'$, we have
\begin{eqnarray}
\|P_Nw_L-h_1\|_{H^1}\le \delta/2 \quad \mbox{for $w\in B_{\ep'}(h_1)$},
\quad
\|P_Nw_L-h_2\|_{H^1}\le \delta/2\quad \mbox{for $w\in B_{\ep'}(h_2)$}.
\label{epsilon prime}
\end{eqnarray}
Hence
$
u_{R,\ep,N,L,\delta}(w)=\ep_0\uZ(w)
$
in a neighborhood of 
$h_1, h_2$ in the topology
of $W$.
Also it is easy to see
$$
\inf_w\left(d_W(w,W_{R/2,N,L,2\delta})
+d_W(w,\overline{W_{R,N,L,\delta}^c})\right)>0.
$$
Hence 
$u_{R,\ep,N,L,\delta}\in \FWU$.
Also we note that
\begin{equation}
u_{R,\ep,N,L,\delta}(w)=
u_{R,\ep,N,L,\delta}(-w)
\quad\quad \mbox{for all $w\in W$}.\label{symmetry of u 2}
\end{equation}
We prove 
\begin{equation}
\sup_{R,N,L,\delta,\ep}\underline{\rho}_{u_{R,\ep,N,L,\delta}}
^W(h_1,h_2)=d_U^{Ag}(h_1,h_2).
\end{equation}
For simplicity, we denote $u_{R,N,L,\ep,\delta}$ by $u$.
Take a path $c$ on $W$ such that
$c(0)\in B_{\ep'}(h_1)$, $c(1)\in B_{\ep'}(h_2)$
and $\{c(t)-c(0)~|~0\le t\le 1\}\in 
AC(L^2(\RR))$,
where $\ep'$ is the positive number 
in (\ref{epsilon prime}).
We give lower bound estimates for the length of $c$.
First we consider the case where
$\sup_{0\le t\le 1}\|c(t)\|_W\ge R$.
Since $\|c(0)\|_W\le \ep'+\|h_1\|_W\le \ep'+R/8$,
there exist times
$0<s_1<s_2\le 1$ such that
$\|c(s_1)\|_W=R/2$,
$R/2\le \inf_{s_1\le t\le s_2}\|c(t)\|_W\le
\sup_{s_1\le t\le s_2}\|c(t)\|_W\le R$ and
$\|c(s_2)\|_W=R$.
By the definition of $u$, we have
for $s_1\le t\le s_2$,
$$
u(c(t))\ge \ep\ep_0\uZ(c(t))\ge
\ep\ep_0\left(\|c(t)\|_W-\frac{R}{8}\right)^2
\ge
\ep\ep_0\left(\|c(t)\|_W-\frac{\|c(t)\|_W}{4}\right)^2
\ge\frac{9\ep\ep_0}{16}\|c(t)\|_W^2.
$$
Noting $C\|w\|_W\le \|w\|_{L^2}$, we get
\begin{eqnarray}
\int_0^1\sqrt{u(c(t))}\|c'(t)\|_{L^2}dt&\ge&
\frac{3C}{4}\sqrt{\ep\ep_0}\int_{s_{1}}^{s_2}
\|c(t)\|_W\|c'(t)\|_Wdt\nonumber\\
&\ge&
\frac{3C}{8}\sqrt{\ep\ep_0}
\left(\|c(s_2)\|_W^2-\|c(s_1)\|_W^2\right)
\nonumber\\
&=&\frac{3C}{32}\sqrt{\ep\ep_0}R^2.\label{other case}
\end{eqnarray}
Next, we consider the case where $\sup_{0\le t\le 1}\|c(t)\|_W\le R$.
Let
\begin{eqnarray}
t_1&=&\sup\left\{t~|~
\|P_N(c(t)_L)-h_1\|_{H^1(\RR)}\le 2\delta
\right\}\\
t_2&=&\inf\left\{t~|~\|P_N(c(t)_L)-h_2\|_{H^1(\RR)}
\le 2\delta
\right\}.
\end{eqnarray}
Then $0<t_1<t_2<1$.
There are three cases where
\begin{itemize}
\item[(a)]~$c(t)\in 
W_{R/2,N,L,2\delta}$ for all $t_1\le t\le t_2$,
\item[(b)]~there exists a minimum time $t_1<t_{\ast}<t_2$ such that
$\|P_N(c(t_{\ast}))_L\|_{H^{1/2}_0(I)}=R/2$ and\\
$\sup_{t_1\le t\le t_{\ast}}\|P_N^{\perp}(c(t)_L)\|_{H^{-2}_0(I)}\le R/2$,
\item[(c)]~there exists a time $t_1\le t_{\ast}\le t_2$ such that
$\|P_N^{\perp}(c(t_{\ast})_L)\|_{H^{-2}_0(I)}=R/2$
and \\
$\sup_{t_1\le t\le t_{\ast}}\|P_N(c(t)_L)\|_{H^{1/2}_0(I)}\le R/2$
\end{itemize}
We consider the case (a).
We have
\begin{eqnarray}
\int_0^1\sqrt{u(c(t))}\|c'(t)\|_{L^2}dt
&\ge&
\int_{t_1}^{t_2}\sqrt{u(c(t))}\|c'(t)\|_{L^2}dt\nonumber\\
&\ge&\int_{t_1}^{t_2}\sqrt{(1-\ep)^3U\left(P_{N}(c(t)_L)\right)}
\|c'(t)\|_{L^2}dt\nonumber\\
&\ge&
(1-\ep)^{3/2}\int_{t_1}^{t_2}\sqrt{U\left(P_{N}(c(t)_L)\right)}
\|\left(P_{N}(c(t)_L)\right)'\|_{L^2}dt.
\end{eqnarray}
Now we define a curve $\tilde{c}$ by
$$
\tilde{c}(t)=
\begin{cases}
\frac{t}{t_{1}}P_{N}(c(t_1)_L)+\frac{t_1-t}{t_1}
h_1 & 0\le t\le t_1 \\
P_{N}(c(t)_L) & t_1\le t\le t_2\\
\frac{1-t}{1-t_2}P_{N}(c(t_2)_L)+\frac{t-t_2}{1-t_2}
h_2 & t_2\le t\le 1.
\end{cases}
$$
Then $\tilde{c}\in AC_{h_1,h_2}(H^1(\RR))$.
Let 
$\gamma(\delta)=\sup\{U(h)~|~\min_{i=1,2}\|h-h_i\|_{H^1(\RR)}\le \delta\}$.
Then
\begin{eqnarray}
\int_0^1\sqrt{u(c(t))}\|c'(t)\|_{L^2}dt
&\ge&
(1-\ep)^{3/2}\int_0^1\sqrt{U(\tilde{c}(t))}\|\tilde{c}'(t)\|_{L^2}dt
-4\frac{\delta}{m} \sqrt{\gamma(2\delta)}\nonumber\\
&\ge& (1-\ep)^{3/2}\dUAg(h_1,h_2)-4\frac{\delta}{m} \sqrt{\gamma(2\delta)}.
\label{case a}
\end{eqnarray}
Next, we consider the case (b).
In this case, there exist times $t_1<s_{\ast}<t_{\ast}<t_2$
such that $\|P_{N}(c(s_{\ast})_L)\|_{H^{1/2}_0(I)}= R/4$  and
$\|P_{N}(c(t)_L)\|_{H^{1/2}_0(I)}\ge R/4$ for $s_{\ast}\le t\le t_{\ast}$.
Since $V$ is a function bounded from below, by taking 
$R$ sufficiently large,
$U(P_N(c(t)_L))\ge \frac{1}{4}U_0(P_N(c(t)_L))$ for
$s_{\ast}\le t\le t_{\ast}$,
where
$
U_0(h)=\frac{1}{4}\|h\|_{H^1_0(I)}^2.
$
Since $\psi(c(t))=1$ for $s_{\ast}\le t\le t_{\ast}$,
\begin{eqnarray}
\int_0^1\sqrt{u(c(t))}\|c'(t)\|_{L^2}dt
&\ge&
\int_{s_{\ast}}^{t_{\ast}}\sqrt{(1-\ep)^3U(P_N(c(t)_L)}
\|P_N(c(t)_L))'\|_{L^2}dt\nonumber\\
&\ge&
\frac{(1-\ep)^{3/2}}{2}\int_{s_{\ast}}^{t_{\ast}}
\sqrt{U_0(P_N(c(t)_L)}\|P_N(c(t)_L)'\|_{L^2}dt\nonumber\\
&\ge&\frac{(1-\ep)^{3/2}}{4}\int_{s_{\ast}}^{t_{\ast}}
\frac{d}{dt}\Bigl\{\left(P_N(c(t)_L),P_N(c(t)_L)\right)_{H^{1/2}_0(I)}
\Bigr\}dt
\nonumber\\
&=&\frac{(1-\ep)^{3/2}}{4}\left(\|P_N\left(c(t_{\ast})_L\right)\|_{H^{1/2}_0(I)}^2-
\|P_N\left(c(s_{\ast})_L\right)
\|_{H^{1/2}_0(I)}^2\right)\nonumber\\
&=&\frac{3(1-\ep)^{3/2}R^2}{64}.\label{case b}
\end{eqnarray}
We consider the case (c).
Using the estimate
$\|h\|_{H^{-2}_0(I)}\le m^{-5/2}\|h\|_{H^{1/2}_0(I)}$,\\
if $\|P_N(c(t_{\ast})_L)\|_{H^{-2}_0(I)}\ge R/3$,
then $\|P_{N}(c(t_{\ast})_L)\|_{H^{1/2}_0(I)}\ge m^{5/2}R/3$.
Hence we can argue similarly to (b).
So we assume $\|P_{N}(c(t_{\ast})_L)\|_{H^{-2}_0(I)}\le R/3$.
By the continuity of the map $w(\in W)\mapsto \chi_L w\in H^{-2}_0(I)$,
there exists $C'(L)>0$ such that
\begin{multline}
\|c(t_{\ast})\|_W\ge 
C'(L)\|c(t_{\ast})\chi_L\|_{H^{-2}_0(I)}\ge
C'(L)\left(\|P_{N}^{\perp}(c(t_{\ast})_L)\|_{H^{-2}_0(I)}-
\|P_{N}(c(t_{\ast})_L)\|_{H^{-2}_0(I)}\right)\\
\ge C'(L)R/6.
\phantom{llllllllllllllllllllllllllllllllllllllllllllllllllllll}
\end{multline}
Again, there exist times $0<s_{\ast}<t_{\ast}$
such that $\|c(s_{\ast})\|_W=C'(L)R/7$ and
$\inf_{s_{\ast}\le t\le t_{\ast}}\|c(t)\|_W\ge C'(L)R/7$.
Now we take $R$ sufficiently large so that
$\|h_i\|_W\le \frac{C'(L)R}{14}$.
Then we have
\begin{equation}
\uZ(c(t))\ge \left(\|c(t)\|_W-\frac{C'(L)R}{14}\right)^2
\ge \left(\|c(t)\|_W-\frac{1}{2}\|c(t)\|_W\right)^2
=
\frac{1}{4}\|c(t)\|_W^2
\qquad
\mbox{for $s_{\ast}\le t\le t_{\ast}$}.
\end{equation}
By the assumption $\sup_{0\le t\le 1}\|c(t)\|_W\le R$,
$u(c(t))\ge \ep\ep_0\uZ(c(t))$ holds for all $t$.
As before, we get
\begin{eqnarray}
\int_0^1\sqrt{u(c(t))}\|c'(t)\|_{L^2}dt&\ge&
\frac{\sqrt{\ep\ep_0}}{2}
\int_{s_{\ast}}^{t_{\ast}}
\|c(t)\|_W\|c'(t)\|_Wdt\nonumber\\
&\ge&
\frac{\sqrt{\ep\ep_0}}{4}
\left(\|c(t_{\ast})\|_W^2-\|c(s_{\ast})\|_W^2\right)
\nonumber\\
&=&\frac{13}{4\cdot 36\cdot 49}
\sqrt{\frac{\ep\ep_0}{2}}C'(L)^2R^2.\label{case c}
\end{eqnarray}
By the estimates (\ref{other case}),
(\ref{case a}), (\ref{case b}),
(\ref{case c}), we are going to finish the proof.
First, we take $\delta$ and $\ep$ sufficiently small
taking the estimates (\ref{ep and delta}) and (\ref{case a}) into account.
For these $\ep, \delta$, we choose $L$ in
Lemma~\ref{cut-off}, (\ref{L and h}).
Next, we take $R$ sufficiently large so that the lower bounds in
(\ref{other case}), (\ref{case b}), (\ref{case c}) are large.
After that, we choose large $N$ for which
(\ref{estimate on U large N}) and (\ref{N and h}) hold.
Finally, by taking $\ep'$ sufficiently small
in (\ref{epsilon prime}), all the above
estimates prove the desired result.
\end{proof}

\section{Example}
We present an example which satisfies assumptions
(A1), (A2), (A3).
Let $P(x)=p(x^2)$ where
$p(x)=(x-1)^{2M_0}$ and $M_0$ is a natural number.
Let us take two positive numbers $a,R$.
Let $g_R$ be a smooth non-negative function with
${\rm supp}~g_R\subset [-R,R]$
and $\|g_R\|_{\infty}\le 1$.
We consider a potential function
$$
V_{a,R}(h)=a\int_{\RR}P(h(x))g_R(x)dx.
$$
Recall that our potential function for the
corresponding classical motion is
$$
U_{a,R}(h)=\frac{1}{4}\int_{\RR}h'(x)^2dx+
\frac{m^2}{4}\int_{\RR}h(x)^2dx+V_{a,R}(h).
$$
We have

\begin{pro}\label{example 1}
For large $a$,
there exist two minimizers $\{h_0, -h_0\}$ of $U_{a,R}$.
Here $h_0$ is a strictly positive $C^2$ function.
Moreover the Hessians of 
$U_{a,R}$ are strictly positive at
$\pm h_0$.
\end{pro}

By this proposition, for large $a$, the polynomial function
\begin{equation}
P_a(x)=aP(x)-\frac{\min U_{a,R}}{\int_{\RR}g_{R}(x)dx}
\end{equation}
satisfies assumptions $\mathrm{(A1), (A2), (A3)}$.

\begin{proof}[Proof of Proposition~$\ref{example 1}$]
~The proof of this proposition is essentially similar to
the proof of Theorem~7.2 in \cite{aida4}.
We give the proof for the sake of completeness.
By Lemma~\ref{sobolev} and a standard argument,
we see that $U_{a,R}$ has a minimizer
$h_0$.
Since $U_{a,R}(|h_0|)\le U_{a,R}(h_0)$,
we may assume that $h_0$ is non-negative.
We show $h_0\not\equiv 0$.
To this end, let $\psi_R$ be a piecewise linear function with
$\psi_R(x)=1$ for  $-R\le x\le R$ and
$\psi_R(x)=0$ for $|x|\le R+1$.
Then for large $a$,
$$
U_{a,R}(\psi_R)<U_{a,R}(0)
$$
which implies $h_0\not\equiv 0$.
For simplicity, 
we denote $aP(x)$ by $P(x)$ and $g_R(x)$ by $g(x)$.
Since $h_0$ satisfies the Euler-Lagrange equation, 
\begin{equation}
(m^2-\Delta)h_0(x)+2P'(h_0(x))g(x)=0,~~~ x\in \RR,
\end{equation}
we see that $h_0\in C^2(\RR)$.
Also $h_0(x)>0$ for all $x$ by the maximum principle.
Thus, the set of minimizers consists of two functions
$h_0,-h_0$ at least.
We need to prove that there are no minimizers other than
$\{h_0,-h_0\}$.
Let
$$
q(x)=\frac{2P'(x)}{x}=4p'(x^2).
$$
Then $h_0$ is the ground state of the Schr\"odinger operator
$
-H_{h_0}=-\Delta+m^2+q(h_0(x)))g(x)
$
with the simple lowest eigenvalue $0$.
Also since the essential spectrum of $-H_{h_0}$ is included in
$[m^2,\infty)$, 
\begin{equation}
\inf\left(\sigma(-H_{h_0})\setminus \{0\}\right)>0
\label{one particle hamiltonian gap}
\end{equation}
holds.
We write $U_0(h)=\frac{1}{4}\int_{\RR}h'(x)^2dx
+\frac{m^2}{4}\int_{\RR}h(x)^2dx$ and
$\tilde{V}(h)=\int_{\RR}P(h(x))g(x)dx$.
Using the derivative $\nabla$ in $L^2(\RR)$, we obtain
\begin{eqnarray}
\lefteqn{U(h)-U(h_0)}\nonumber\\
& &=
\frac{1}{2}\nabla^2U(h_0)(h-h_0,h-h_0)\nonumber\\
& &~~+U(h)-U(h_0)-\nabla U(h_0)(h-h_0)-
\frac{1}{2}\nabla^2U(h_0)(h-h_0,h-h_0)\nonumber\\
& &=\frac{1}{2}\nabla^2U_0(h_0)(h-h_0,h-h_0)
+\frac{1}{2}\nabla^2\tilde{V}(h_0)(h-h_0,h-h_0)\nonumber\\
& &~~+\tilde{V}(h)-\tilde{V}(h_0)-\nabla \tilde{V}(h_0)(h-h_0)-
\frac{1}{2}\nabla^2 \tilde{V}(h_0)(h-h_0,h-h_0).
\end{eqnarray}
Thus
\begin{eqnarray*}
\lefteqn{U(h)-U(h_0)}\nonumber\\
& &=\frac{1}{2}\nabla^2U_0(h_0)(h-h_0,h-h_0)
+\frac{1}{4}\int_{\RR}q(h_0(x))(h(x)-h_0(x))^2g(x)dx\nonumber\\
& &~~+\tilde{V}(h)-\tilde{V}(h_0)-
\int_{\RR}P'(h_0(x))(h(x)-h_0(x))g(x)dx\nonumber\\
& &\quad-\frac{1}{4}\int_{\RR}q(h_0(x))(h(x)-h_0(x))^2g(x)dx.\nonumber\\
& &=\frac{1}{4}\left(-H_{h_0}(h-h_0),(h-h_0)\right)_{L^2}\nonumber\\
& &~~+
\int_{\RR}
\left(p(h(x)^2)-p(h_0(x)^2)-
p'(h_0(x)^2)(h(x)^2-h_0(x)^2)\right)g(x)dx\nonumber\\
& &=\left(-H_{h_0}(h-h_0),(h-h_0)\right)_{L^2}\nonumber\\
& &~~+
\int_{\RR}\left\{\int_0^1
\left(\int_0^{\theta}p''(h_0(x)^2+\tau (h(x)^2-h_0(x)^2))
d\tau\right)d\theta\right\}(h(x)^2-h_0(x)^2)^2g(x)dx.
\end{eqnarray*}
Combining the formula above and  
(\ref{one particle hamiltonian gap}),
we see that the minimizers of $U$ are $\{\pm h_0\}$ only.
Finally, we prove that the bottom of the spectrum of
$m^2-\Delta+2P''(h_0(x))g(x)$ in $L^2(\RR)$ is strictly positive.
Noting 
$$
2P''(h_0(x))g(x)=q(h_0(x))g(x)+
8h_0(x)^2p''(h_0(x)^2)g(x),
$$
(\ref{one particle hamiltonian gap})
and the fact that
$h_0$ is the ground state of $-H_{h_0}$,
we obtain
$$
\inf\sigma\left(m^2-\Delta+2P''(h_0(x))g(x)\right)>0
$$
which completes the proof.
\end{proof}

\section{Appendix}

\subsection{Proof of Lemma~\ref{Kv}, Lemma~\ref{quadratic},
Lemma~\ref{quadratic 2}}
We prove Lemma~\ref{Kv}, Lemma~\ref{quadratic}
and Lemma~\ref{quadratic 2}.
Some parts of the proofs are similar to that of Lemma~2.8
in \cite{aida2}.

\begin{proof}[Proof of Lemma~$\ref{Kv}$]
Let $a>0$ and $\alpha>1$.
Then we have the following estimate:
\begin{equation}
\int_0^{\infty}\frac{e^{-t-\frac{a^2}{t}}}{t^{\alpha}}dt\le
\frac{C}{\alpha-1}
e^{-a/2}\left(a^{2(1-\alpha)}+a^{1-\alpha}\right).
\end{equation}
Therefore by the functional calculus, we have
an estimate on the integral kernel,
\begin{equation}
0\le \tilde{A}^{-1}(x,y)\le C\frac{e^{-C|x-y|}}{\sqrt{|x-y|}}.
\end{equation}
By this estimate, if $v$ is a non-negative function with compact support,
then
\begin{eqnarray}
\left(\tilde{A}^{-1}M_v\tilde{A}^{-1}\right)(x,y)&\le&
C\|v\|_{\infty}\int_{{\rm supp}\, v}\frac{e^{-C|x-z|}}{\sqrt{|x-z|}}
\frac{e^{-C|z-y|}}{\sqrt{|z-y|}}dz\nonumber\\
&\le&Ce^{-C(|x|+|y|)}\left(1+\log\left(\frac{1}{|x-y|}\vee 1\right)
\right).
\end{eqnarray}
This implies the Hilbert-Schmidt property of
$\tilde{A}^{-1}M_v\tilde{A}^{-1}$.
Other statements are clear.
We prove (2).
Let $\tilde{S}_n$ be the bounded linear operator on 
$L^2(\RR)$ such that
$(\tilde{S}_n\varphi)(x)=\int_{\RR}p_n(x-y)\varphi(y)dy$,
where 
$p_n(x)=\left(\frac{n}{4\pi}\right)^{1/2}\exp\left(-\frac{nx^2}{4}\right)$.
Also recall that we denote
by $w_n(x)={}_{{\cal S}(\RR)}\langle p_{n}(x-\cdot),
w\rangle_{{\cal S}(\RR)'}$
for $w\in {\cal S}(\RR)'$.
Clearly, $\varphi_n(x)=(\tilde{S}_n\varphi)(x)$ for any
$\varphi\in L^2(\RR)$.
Let $S_n=\Phi\circ\tilde{S}_n\circ \Phi^{-1}$.
We also have $S_nh(x)=\tilde{S}_nh(x)$ and
$\lim_{n\to\infty}S_n=I$ strongly.
We approximate $K_v$ by $S_nK_vS_n$.
By the definition, we obtain
\begin{eqnarray}
S_nK_vS_nh&=&
\Phi\left(\tilde{S_n}\tilde{A}^{-1}M_v\tilde{A}^{-1}\tilde{S}_n\right)
\Phi^{-1}h.
\end{eqnarray}
Using the commutativity of
$\tilde{S}_n$ and $\tilde{A}$
and the fact that $\tilde{S}_n M_v$ is a Hilbert-Schmidt operator,
we can conclude that
$S_nK_vS_n$ is a trace class operator.
Also we have
\begin{eqnarray}
\left(S_nK_vS_nh,h\right)_H&=&
\left(\tilde{A}\tilde{S_n}\tilde{A}^{-2}M_v\tilde{S}_nh,
\tilde{A}h\right)_{L^2}\nonumber\\
&=&\left(\tilde{S}_nM_v\tilde{S}_nh, h\right)_{L^2}\nonumber\\
&=&
\int_{\RR}\tilde{S}_nh(x)^2v(x)dx\nonumber\\
&=&\int_{\RR}h_n(x)^2v(x)dx.
\end{eqnarray}
By the continuity $h(\in H)\mapsto \int_{\RR}h_n(x)^2v(x)dx$ in the topology
of $W$ and the trace class property of $S_nK_vS_n$, we obtain that
\begin{equation}
:\langle S_nK_vS_nw,w\rangle:=
\int_{\RR}w_n(x)^2v(x)dx-\tr (S_nK_vS_n).
\end{equation}
Since $E_{\mu}[:\langle S_nK_vS_nw,w\rangle:]=0$,
we have $\tr (S_nK_vS_n)=c_n^2\int_{\RR}v(x)dx$, where
$c_n^2=E_{\mu}[w_n(x)^2]$.
Letting $n\to\infty$, we complete the proof of
the statement (2).
\end{proof}

\begin{proof}[Proof of Lemma~$\ref{quadratic}$]
(1) follows from the definition of $K_v$.
We prove (2) (i).
Let $P^{(v)}_t=e^{-t(m^2+4v-\Delta)}$.
By the functional calculus,
we have
\begin{eqnarray}
\tilde{A}_v^2-\tilde{A}^2&=&
\frac{2}{\sqrt{\pi}}\int_0^{\infty}
\frac{P_t-P^{(v)}_t}{t^{3/2}}dt\nonumber\\
&=&\frac{2}{\sqrt{\pi}}\int_0^{1}
\frac{P_t-P^{(v)}_t}{t^{3/2}}dt+
\frac{2}{\sqrt{\pi}}\int_1^{\infty}
\frac{P_t-P^{(v)}_t}{t^{3/2}}dt=:I_1(x,y)+I_2(x,y).
\end{eqnarray}
Here $I_i(x,y)$ are the kernel functions.
By the Feynman-Kac formula,
\begin{eqnarray}
\frac{|P_t(x,y)-P^{(v)}_t(x,y)|}{t^{3/2}}
&\le&
\frac{4}{t}\exp\left(-m^2t+4\|v\|_{\infty}t-\frac{(x-y)^2}{4t}\right)J(t,x,y),
\end{eqnarray}
where 
\begin{equation}
J(t,x,y)=E[\max_{0\le s\le t}|v(x+\sqrt{2}B(s))|~|~\sqrt{2}B(t)=y]
\end{equation}
and $B(t)$ is the 1-dimensional standard Brownian motion starting at
$0$.
We give an estimate for $J(t,x,y)$.
Suppose that the support of $v$ is included in
$[-L,L]$.
Let $r=\min\{|x|,|y|\}$.
We have
\begin{multline}
J(t,x,y)\le \|v\|_{\infty}
P\left(\max_{0\le s\le t}|B(s)|\ge \frac{r}{2\sqrt{2}}~|~
B(t)=0\right)
\le C\|v\|_{\infty}\exp\left(-C\frac{r^2}{t}\right)\\
\mbox{for $x,y\ge 2L$~ or~ $x,y\le -2L$}.
\phantom{lllllllllllllllllllllllllllllllllllllllll}
\label{estimate for J}
\end{multline}
Noting 
\begin{equation}
\int_0^1\frac{\exp(-\frac{a^2}{t})}{t}dt
\le C\left(1+\log(\max(\frac{1}{a},1))\right)
\exp\left(-\frac{a^2}{2}\right), \label{log}
\end{equation}
we obtain
\begin{equation}
I_1(x,y)\le
C(L)\left(|\log(x^2+y^2)|+1\right)\exp\left(-C(x^2+y^2)\right) \qquad
\mbox{for $x,y\ge 2L$~ or~ $x,y\le -2L$}.
\end{equation}
We have similar estimates for other cases.
Consequently,
$\int_{\RR^2}I_1(x,y)^2dxdy<\infty$.
Next, we show $I_2\in L^2$.
By (\ref{estimate for J}),
we get an estimate for the Hilbert-Schmidt norm
of $P^{(v)}_t-P_t$:
\begin{equation}
\|P^{(v)}_t-P_t\|_{L_{(2)}(L^2(\RR))}
\le
C\sqrt{t^{3/2}+4L^2t^2\|v\|_{\infty}^2}e^{4\|v\|_{\infty}t-m^2t}.
\end{equation}
By using the method in page 3349, 3350 in \cite{aida2},
we obtain the following estimate:
There exists a positive number
$C(n)$ which depends only on the natural number $n$
such that
\begin{equation}
\|P_{nt}^{(v)}-P_{nt}\|_{L_{(2)}(L^2(\RR))}
\le C(n)e^{-cnt}\|P^{(v)}_t-P_t\|_{L_{(2)}(L^2(\RR))},
\end{equation}
where
$c=\min\left(\inf\sigma(m^2-\Delta+4v), m^2\right)>0$.
Thus we get $I_2\in L^2$.
We prove (ii).
Since $A_v^2-A^2$ is unitarily equivalent to
$\tilde{A}_v^2-\tilde{A}^2$,
it suffices to prove
$\inf\sigma(T_v)>-1$.
Let $h\in {\rm D}(A)$.
Then $A^{-1}h\in {\rm D}(A)={\rm D}(A_v)$ and
$$
\left(T_vh,h\right)_H=\left(A^{-1}A_v^2A^{-1}h,h\right)_H-
\|h\|_H^2\ge \left(\inf\sigma(A^{-1}A_v^2A^{-1})-1\right)\|h\|_H^2.
$$
Since there exists $C>0$ such that
$(A_v^2h,h)_H\ge C(A^2h,h)_H$ for any $h\in {\rm D}(A^2)$,
we get $\inf\sigma(A^{-1}A_v^2A^{-1})>0$ which implies (ii).
We prove (iii).
Let $S_n$ be the mollifier operator in the proof of Lemma~\ref{Kv}.
Note that $S_n$ and $A$ commute.
Let $T_{v,n}=S_nT_vS_n$.
Let us define 
\begin{equation}
\Omega_{v,n}=\det{}_{(2)}(I+T_{v,n})^{1/4}
\exp\left[
-\frac{1}{4}:\langle T_{v,n}w,w \rangle_H:
\right].\label{Omegavn}
\end{equation}
By a simple calculation, we obtain
\begin{eqnarray}
\left(-L_A
+\frac{1}{4}:\langle(T_{v,n}A^2+A^2T_{v,n}+T_{v,n}A^2T_{v,n})w, w\rangle:
\right)\Omega_{v,n}
&=&-\frac{1}{4}\tr\left(T_{v,n}A^2T_{v,n}\right)\Omega_{v,n}.\nonumber\\
& &\label{LAOmegavn}
\end{eqnarray}
Note that
${\rm D}(A_v^2)={\rm D}(A^2)$,
${\rm D}(A_v^4)={\rm D}(A^4)$
and 
\begin{equation}
A_v^2({\rm D}(A^{2}))\subset H,~~
A_v^2({\rm D}(A^4))\subset {\rm D}(A^2).
\end{equation}
By the definition of $T_v$,
we have
\begin{equation}
A^2+AT_vA=(A^4+4AK_vA)^{1/2}(=A_v^2).
\end{equation}
Hence $(AT_vA)({\rm D}(A^4))\subset {\rm D}(A^2)$
and $(AT_vA)({\rm D}(A^2))\subset H$.
Therefore for any $h\in {\rm D}(A^4)$,
\begin{equation}
A^3T_vAh+AT_vA^3h+(AT_vA)(AT_vA)h
=4AK_vAh\in H.
\end{equation}
Hence for any $h\in {\rm D}(A^3)$,
\begin{equation}
A^2T_vh+T_vA^2h+T_vA^2T_vh=4K_vh\in {\rm D}(A).
\end{equation}
This implies
\begin{eqnarray}
\frac{1}{4}
\left(T_{v,n}A^2+A^2T_{v,n}+T_{v,n}A^2T_{v,n}\right)
&=&
S_nK_vS_n
+\frac{1}{4}S_nT_vA^2(S_{2n}-I)T_vS_n.
\end{eqnarray}
By combining this with (\ref{LAOmegavn}) ,
\begin{eqnarray}
\lefteqn{\left(-L_A+:\langle K_v w,w\rangle:\right)\Omega_{v,n}}
\nonumber\\
& &=
\left(:\langle (K_v-S_nK_vS_n)w,w \rangle:+
:\langle R_{v,n}w,w\rangle:\right)\Omega_{v,n}
-\frac{1}{4}\|S_n(A^2-A_v^2)A^{-1}S_n\|_{L_{(2)}(H)}^2
\Omega_{v,n},\nonumber\\
& &
\end{eqnarray}
where 
$
R_{v,n}=
\frac{1}{4}
S_nT_v(I-S_{2n})A^2T_vS_n.
$
Letting $n\to\infty$, we see that
$\Omega_v$ is a positive eigenfunction of $-L_A+:\langle K_vw,w\rangle:$
with the eigenvalue 
$-\frac{1}{4}\|(A^2-A_v^2)A^{-1}\|_{L_{(2)}(H)}^2$ which implies
(2).
(3) can be proved by a linear transformation formula of Gaussian measures.
\end{proof}

To prove Lemma~\ref{quadratic 2},
we need the following lemma.

\begin{lem}\label{Hilbert-Schmidt inequality}
Let $A$ be a strictly positive self-adjoint operator
and $K$ be a Hilbert-Schmidt self-adjoint operator on $H$.
Assume that $A^4+K$ is also a strictly positive operator.
Then 
${\rm D}((A^4+K)^{1/2})={\rm D}(A^2)$.
Moreover
$(A^4+K)^{1/2}-A^2$ is a Hilbert-Schmidt operator and
\begin{equation*}
\left\|(A^4+K)^{1/2}-A^2\right\|_{L_{(2)}(H)}\le
\left(\left\{\inf\sigma(\sqrt{A^4+K})\right\}^{-1/2}+
\left\{\inf\sigma(A^2)\right\}^{-1/2}\right)
\|K\|_{L_{(2)}(H)}
\end{equation*}
\end{lem}

\begin{proof}
We can prove this result
using
the L\"owner's theorem:
For any strictly positive self-adjoint operator
$T$, it holds that
\begin{equation}
T^{1/2}\varphi=\frac{1}{\pi}\int_0^{\infty}
u^{-1/2}T(T+u)^{-1}\varphi du.
\end{equation}
Let $T_1$ and $T_2$ be strictly positive self-adjoint operators
such that $T_1-T_2$ is a bounded linear operator.
Applying the above representation, we get
\begin{eqnarray*}
\lefteqn{\sqrt{T_1}-\sqrt{T_2}}\nonumber\\
& &=
\frac{1}{\pi}\int_0^{\infty}\frac{1}{\sqrt{u}}
(T_1-T_2)(T_1+u)^{-1}du+
\frac{1}{\pi}\int_0^{\infty}
\frac{1}{\sqrt{u}}T_1(T_1+u)^{-1}(T_2-T_1)(T_2+u)^{-1}du.
\end{eqnarray*}
This shows ${\rm D}(\sqrt{T_1})={\rm D}(\sqrt{T_2})$.
Applying this and the identity above to the case where
$T_1=A^4+K$ and $A^4$, we get the desired result.
\end{proof}

\begin{proof}[Proof of Lemma~$\ref{quadratic 2}$]
$(1)$~Note that $A_{v,J}^2-A^2=A_{v,J}^2-A_v^2+A_v^2-A^2$.
By Lemma~\ref{Hilbert-Schmidt inequality},
$A_{v,J}^2-A_v^2$ is a Hilbert-Schmidt operator.
Since $A_v^2-A^2$ is also a Hilbert-Schmidt operator,
$A_{v,J}^2-A^2$ is a Hilbert-Schmidt operator.
The proof of (2) and (3) are similar to 
that of Lemma~\ref{quadratic}.
So we omit the proof.
\end{proof}

\subsection{Properties of Agmon distance}

In Section~2, we defined the length of 
a curve and the Agmon distance on $H^1(\RR)$.
However the distance can be extended to
a distance on $H^{1/2}(\RR)$.
In this subsection, we define
the length and the energy of a curve on $H^{1/2}(\RR)$
which is an extension of the previous one.
Through out this subsection, we assume that
$U$ satisfies (A1) and (A2).
Of course we include the case where ${\cal Z}=\emptyset$.
We consider the case where
the space is $\RR$.
However, all statements together with those in
next subsection hold true for 
the finite interval cases by similar arguments.

\begin{defin}\label{agmon distance2}
$(1)$~Let $h,k\in H^{1/2}$.
Let $U$ be a non-negative potential function as in
$(\ref{U})$.
Let ${\cal P}_{T,h,k,U}$ be all continuous paths
$c=c(t)$~$(0\le t\le T)$ on
$H^{1/2}$ such that $c(0)=h, c(T)=k$ and
\begin{itemize}
\item[{\rm (i)}] $c\in AC_{T,h,k}(L^2(\RR))$,
\item[{\rm (ii)}] $c(t)\in H^1(\RR)$ for $\|c'(t)\|dt$
-a.e.~$t\in [0,T]$
and
\begin{equation}
\int_0^T\sqrt{U(c(t))}\|c'(t)\|_{L^2}dt<\infty.\label{agmon length P}
\end{equation}
\end{itemize}
We define the length $\ell_U(c)$ of $c\in {\cal P}_{T,h,k,U}$ by
the integral value of $(\ref{agmon length P})$.
Also we define the energy of $c$ by
\begin{equation}
e_U(c)=\int_0^TU(c(t))\|c'(t)\|_{L^2}^2dt.\label{agmon energy P}
\end{equation}

\noindent
$(2)$~
Let $0<T<\infty$.
We define the Agmon distance between $h, k\in H^{1/2}(\RR)$
by
\begin{eqnarray}
\widetilde{\dUAg}(h,k)&=&
\inf\left\{\ell_U(c)~|~
c\in {\cal P}_{T,h,k,U}
\right\}.
\end{eqnarray}
\end{defin}

We may omit writing $T$, $U$, $h$, $k$ 
in the notations ${\cal P}_{T,h,k,U}, \ell_U, e_U$
if there are no confusion.
By using natural reparametrization of path,
we see that the definition of $\widetilde{\dUAg}$ does not depend on
$T$.
If $h,k\in H^1(\RR)$,
$AC_{T,h,k}(H^1(\RR))\subset{\cal P}_{T,h,k,U}$ holds.
When $h,k\notin H^1$, it is not obvious
but elementary to see ${\cal P}_{T,h,k,U}$ is not empty.
Let $H^1_{T,h,k}(\RR)$ be the set of functions
$u=u(t,x)$~$(t\in (0,T), x\in \RR)$ in $H^1$-Sobolev space
on $(0,T)\times \RR$ and 
$u(0,\cdot)=h, u(T,\cdot)=k$ in the sense of trace.
Then $H^1_{T,h,k}(\RR)\subset {\cal P}_{T,h,k}$.
To check 
$u\in H^1_{T,h,k}(\RR)$ satisfies 
(\ref{agmon length P}),
we consider a functional
\begin{eqnarray}
I_{T,P}(u)
&=&
\frac{1}{4}\iint_{(0,T)\times \RR}\left(
\left|\frac{\partial u}{\partial t}(t,x)\right|^2+
\left|\frac{\partial u}{\partial x}(t,x)\right|^2\right)dtdx
\nonumber\\
& &\qquad\qquad+
\iint_{(0,T)\times \RR}
\left(
\frac{m^2}{4}u(t,x)^2+
P(u(t,x))g(x)
\right)dtdx.\label{ITPu}
\end{eqnarray}
By Sobolev's theorem, 
$I_{T,P}(u)<\infty$ for any $u\in H^1((0,T)\times \RR)$.
Since
\begin{equation}
\int_{0}^T\sqrt{U(u(t))}\|\partial_tu(t)\|_{L^2}dt
\le \int_{0}^TU(u(t))dt+
\frac{1}{4}\int_{0}^T\|\partial_tu(t)\|_{L^2}^2dt
=I_{T,P}(u),
\end{equation}
the boundedness (\ref{agmon length P}) holds.
Let $h, k\in H^1$
and take $c\in AC_{T,h,k}(H^1(\RR))$.
It is evident that
the definition of the length
$\ell_U(c)$ above coincides with the previous one
in Definition~\ref{agmon distance}.
Actually the distance above on $H^1$ coincides with the previous one
in Definition~\ref{agmon distance}.

\begin{lem}\label{two agmon}
For any $h,k\in H^1(\RR)$,
\begin{eqnarray}
\widetilde{\dUAg}(h,k)=\dUAg(h,k).
\end{eqnarray}
\end{lem}

\begin{proof}
Let $T=1$.
We need only to prove $\dUAg(h,k)\le \widetilde{\dUAg}(h,k)$.
To this end, take $c\in {\cal P}_{1,h,k}$.
Let $c_{\ep}=c_{\ep}(t,x)=P_{\ep}c(t)$, where
$P_{\ep}=e^{\ep\Delta_x}$.
Then $c_{\ep}\in AC(H^1(\RR))$ and
$c_{\ep}(0)=P_{\ep}h$ and $c_{\ep}(1)=P_{\ep}k$.
We have 
\begin{eqnarray}
\lim_{\ep\to 0}U(c_{\ep}(t))&=&U(c(t))
\quad\mbox{for $t$ such that $c(t)\in H^1$,}\label{claim 1}\\
\lim_{\ep\to 0}\|c_{\ep}'(t)\|_{L^2}&=&\|c'(t)\|_{L^2}
\quad a.e. t\in [0,1].\label{claim 2}
\end{eqnarray}
The convergence (\ref{claim 1}) is a consequence of 
strong continuity of $P_{\ep}$
in $H^1(\RR)$ and $L^p(\RR)$.
The convergence (\ref{claim 2}) follows from the fact that
$c_{\ep}'(t)=P_{\ep}(c'(t))$ holds at
the differentiable point $t$ of $c=c(t)$.
By the contraction property of $P_{\ep}$ on
$H^1(\RR)$, $H^{1/2}(\RR)$ and
$\sup_{t}\|c(t)\|_H<\infty$, 
we have
$U(c_{\ep}(t))\le U(c(t))+C$.
Thus
for any $0<\ep<1$
\begin{equation}
\sqrt{U(c_{\ep}(t))}\|c_{\ep}'(t)\|_{L^2}\le
\left(\sqrt{U(c(t))}+C\right)\|c'(t)\|_{L^2}
\label{claim 3}
\end{equation}
Note that the function on the right-hand side of (\ref{claim 3})
is integrable on $[0,1]$.
Let $\delta>0$.
Noting $P_{\ep}h$ and $P_{\ep}k$ converge to
$h$ and $k$ in $H^1$ respectively and using the path $c_{\ep}$ and the
line segments connecting $h$ and $P_{\ep}h$,
$k$ and $P_{\ep}k$,
we can construct a path $\tilde{c}_{\ep}\in AC_{h,k}(H^1(\RR))$ 
such that $\ell_U(\tilde{c}_{\ep})\le \ell_U(c)+\delta$
which completes the proof.
\end{proof}

\begin{rem}\label{remark on absolute continuity}
In Definition~$\ref{agmon distance2}$, 
we assume $c$ satisfies
$\int_0^T\|c'(t)\|_{L^2}dt<\infty$.
However, for curves which pass through the zero points
of $U$,
it may be natural to consider 
the case where the $L^2$ length of the curves themselves
are infinite
near zero points.
In view of this observation,
we introduce a larger set of paths 
${\cal P}^{loc}_{T,h,k,U}$ which includes
${\cal P}_{T,h,k,U}$.
We say that a continuous path $c$ on $H^{1/2}$
starting at $h$ and ending at $k$
belongs to ${\cal P}^{loc}_{T,h,k,U}$ if and only if
the following two conditions hold:
\begin{itemize}
\item[{\rm (i)}]~
there exist a finitely many times $0=t_0<\cdots<t_n=T$
such that for any closed interval $I\subset (t_i,t_{i+1})$
~$(0\le i\le n-1)$,
the restricted path $c|_I$ is an absolutely continuous path.
\item[{\rm (ii)}]
The same condition as in Definition~$\ref{agmon distance2}$~{\rm (ii)}
holds.
\end{itemize}
Clearly, if $\int_{I}\|c'(t)\|_{L^2}dt=\infty$
for some interval $I$ including $t_i$,
then the finiteness of $\ell(c)$ implies that
there exists a sequence of times
$s_n\to t_i$ such that
$U(c(s_n))\to 0$ and this implies
$c(t_i)\in {\cal Z}$.
By this observation,
the value of Agmon distance $\widetilde{\dUAg}$ does not change even if
including all paths in ${\cal P}^{loc}_{T,h,k,U}$
in the definition
of the distance by a simple argument.
However, probably, minimal geodesics in $H^{1/2}$
belong to ${\cal P}_{T,h,k,U}$.
\end{rem}

The same as finite dimensional cases,
it is useful to consider the reparametrization of path
by the length.

\begin{lem}\label{reparametrization}
\noindent
Let $c\in {\cal P}^{loc}_{1,h,k,U}$.

\noindent
$(1)$~It holds that
$\ell(c)\le \sqrt{e(c)}$.

\noindent
$(2)$~
Let $c\in {\cal P}^{loc}_{1,h,k,U}$
and assume that
$J_{{\cal Z}}=\{t\in [0,T]~|~c(t)\in {\cal Z}\}$
is a finite set.
Let
$$
\tau(t)=\frac{1}{\ell(c)}\int_0^t\sqrt{U(c(t))}\|c'(t)\|_{L^2}dt
\qquad\qquad (0\le t\le 1).
$$
Then there exists $c_{\ast}\in {\cal P}^{loc}_{1,h,k,U}$
such that $c_{\ast}(\tau(t))=c(t)$~$(0\le t\le 1)$.
Moreover, $\ell(c)=\ell(c_{\ast})$ and
\begin{equation}
\sqrt{U(c_{\ast}(t))}\|c_{\ast}'(t)\|_{L^2}=\ell(c_{\ast})
=\sqrt{e(c_{\ast})}\le \sqrt{e(c)}\quad
a.e.~t\in [0,1].
\end{equation}
\end{lem}

\begin{proof}
The estimate $(1)$ follows from the Schwarz inequality.
We prove (2).
Let $\sigma(t)=\inf\{s~|~\tau(s)>t\}$ and
set $c_{\ast}(t):=c(\sigma(t))$.
Then $c_{\ast}$ is a continuous curve on $H$ and
$c_{\ast}(\tau(t))=c(t)$.
These follow from that $\tau(t_1)=\tau(t_2)$~$(t_1<t_2)$
is equivalent to $c'(s)=0$ for a.e.~$t_1\le t\le t_2$.
Moreover the image measure of the Lebesgue
measure by $\sigma$ is given by $\sigma_{\sharp}dt=\tau'(t)dt$.
So $\sigma_{\sharp}dt$ is absolutely continuous to
$\|c'(t)\|_{L^2}dt$.
Let $J_{{\cal Z}}=\{t_1<\cdots<t_n\}$.
Then $\{t~|~U\left(c_{\ast}(t)\right)=0\}=\{\tau(t_1),\ldots,\tau(t_n)\}$.
Let $\tau(t_i)<s<t<\tau(t_{i+1})$.
By using a change of variable formula,
we obtain 
\begin{equation}
c_{\ast}(t)-c_{\ast}(s)=
-\ell(c)\int_s^t\frac{v(\sigma(u))}{\sqrt{U(c_{\ast}(u))}}du,
\end{equation}
where $v(t)=\frac{c'(t)}{\|c'(t)\|_{L^2}}1_{c'(t)\ne 0}$.
This implies $c_{\ast}\in {\cal P}^{loc}_{1,h,k,U}$
and 
$$
\sqrt{U(c_{\ast}(t))}\|c_{\ast}'(t)\|_{L^2}=\ell(c_{\ast})=
\sqrt{e(c_{\ast})}=\ell(c)\qquad a.e.~t\in [0,1].
$$
\end{proof}

In this subsection, we prove the following properties
of Agmon distance.

\begin{thm}\label{Theorem 1 for Agmon}
$(1)$~The function $\dUAg$ is a distance function on $H$.
Moreover the topology defined by
$\dUAg$ on $H$ is the same as the one defined by
the Sobolev norm of $H^{1/2}$.

\noindent
$(2)$~
Let us consider the case where 
$U(h)=U_0(h)=\frac{1}{4}\|Ah\|_H^2$.
In this case, we have
$d_{U_0}^{Ag}(0,h)=\frac{1}{4}\|h\|_H^2$.
\end{thm}

\begin{thm}\label{Theorem 2 for Agmon}
Assume ${\cal Z}$ consists of two points $\{h,k\}$.
There exists a curve $c_{\star}\in {\cal P}^{loc}_{1,h,k}$
such that
$\ell(c_{\star})=\dUAg(h,k)$.
This $c_{\star}$ has the following properties.

\noindent
$(1)$~
$c_{\star}(t)\notin {\cal Z}$ for $0<t<1$.

\noindent
$(2)$~
$c_{\star}=c_{\star}(t,x)$ is a $C^{\infty}$ function of
$(t,x)\in (0,1)\times \RR$
and $c_{\star}\in H^1(\ep, 1-\ep)\times \RR)$
for all $0<\ep<1$.

\noindent
$(3)$~
$\int_0^{\ep}\|c_{\star}'(t)\|_{L^2}^2dt
=\int_{1-\ep}^1\|c_{\star}'(t)\|_{L^2}^2dt
=+\infty$ for any $\ep>0$.
\end{thm}

We prepare a lemma for the proof of
Theorem~\ref{Theorem 1 for Agmon}.

\begin{lem}\label{agmon lemma}
Let $S_t=e^{-t\sqrt{m^2-\Delta}}$ be the Cauchy semigroup on
$L^2(\RR)$.
Let $T>0$.
For $h,k\in H^{1/2}$, define a function
$f=f(t,x)$~$(0<t<T, x\in \RR)$ by
\begin{equation}
f(t)=S_{T-t}(I-S_{2T})^{-1}(k-S_Th)+S_t(I-S_{2T})^{-1}(h-S_Tk).
\label{fhkT}
\end{equation}
Then the following hold.

\noindent
$(1)$~We have $\lim_{t\to 0}\|f(t)-h\|_{H^{1/2}}=
\lim_{t\to 1}\|f(t)-k\|_{H^{1/2}}=0$.
Let $I_{T,0}$ be the functional
in the case where $P=0$.
Then it holds that
\begin{eqnarray}
I_{T,0}(f)&=&\frac{1}{4}
\Bigl\{
2\left(\left((I-S_{2T})^{-1}-(I+S_T)^{-1}\right)(h-k), h-k\right)_H+
\left((I-S_T)(I+S_T)^{-1}h,h\right)_H\nonumber\\
& &+
\left((I-S_T)(I+S_T)^{-1}k,k\right)_H\Bigr\}.
\label{IT0f}
\end{eqnarray}
In particular, 
$f\in H^1_{T,h,k}(\RR)$.

\noindent
$(2)$~It holds that
$\displaystyle{
\sup_{0<t<T}\|f(t)\|_{H^{1/2}}\le
3\left(\|h\|_{H^{1/2}}+\|k\|_{H^{1/2}}\right)}
$
and $f(t)\in H^n(\RR)$ for all $0<t<T$
and $n\in {\mathbb N}$.

\noindent
$(3)$~Let $0<\ep<1$ and fix $h\in H$.
Then there exists $0<\delta(\ep)\le 1$ such that
for any $k\in H$ with $\|h-k\|_H\le \delta(\ep)$,
$
\dUAg(h,k)\le\ep
$
holds.

\noindent
$(4)$~If $\lim_{n\to\infty}U(h_n)=0$, then 
$\lim_{n\to\infty}\min
\left\{\|h_n-h\|_{H^1}~|~h\in {\cal Z}\right\}=0$.

\noindent
$(5)$~Let $h\in H$. If $\lim_{n\to\infty}\dUAg(h_n,h)=0$,
then
$\lim_{n\to\infty}\|h_n-h\|_{L^2}=0$ holds.

\noindent
$(6)$~Let $U_0(h)=\frac{1}{4}\|Ah\|_H^2$.
Then for any $h, k\in H$,
\begin{equation}
d_{U_0}^{Ag}(h,k)\ge \frac{1}{4}\left\{\max(\|h\|_H^2,\|k\|_H^2)-
\min(\|h\|_H^2,\|k\|_H^2)\right\}.\label{agmon h k}
\end{equation}

\noindent
$(7)$~Assume ${\cal Z}=\{h,k\}$ is a two point set.
Let $c\in {\cal P}_{1,h,k,U}$.
Let
\begin{equation}
\delta(\ep)=
\inf
\left\{U(\varphi)~\Big |~
\min\{\|\varphi-h\|_{L^2}, \|\varphi-k\|_{L^2}\}\ge \ep/2
\right\}.
\end{equation}
Let $\ep<\frac{\|h-k\|_{L^2}}{4}$.
Then for any $t$ such that
\begin{equation}
\max\left\{t,1-t\right\}\le
\frac{\delta(\ep)\ep^2}{4e_U(c)},
\end{equation}
it holds that
\begin{equation}
\max\left\{\|h-c(t)\|_{L^2}, \|k-c(1-t)\|_{L^2}\right\}
\le \ep.
\end{equation}
\end{lem}

\begin{rem}
Since the function $f$ in $(\ref{fhkT})$ depends on $T$,
we denote it by $f^T$.
We note that $f^T$ in $(\ref{fhkT})$
is the unique minimizer of the functional
$I_{T,0}$ on $H^1_{T,h,k}(\RR)$
since $f^T$ satisfies 
\begin{multline*}
m^2f^T(t,x)-\left(\frac{\partial^2}{\partial t^2}+
\frac{\partial^2}{\partial x^2}\right)f^T(t,x)=0\qquad
(t,x)\in (0,T)\times \RR,\\
f^T(0,x)=h(x),~~~f^T(T,x)=k(x)\qquad\qquad  x\in \RR.
\phantom{lllllllllllllllllllllllllllllllllllllllll}
\end{multline*}
Also it is easy to show that
if $h\ne 0, h\ne k$ and $h-k$ is small enough, then
\begin{itemize}
\item[{\rm (i)}]~$\partial_T\left(I_{T,0}(f^T)\right)>0$ for large $T$,
\item[{\rm (ii)}]~$\lim_{T\to 0}I_{T,0}(f^T)=+\infty$.
\end{itemize}
Hence there exists $T_{\ast}>0$ such that
$I_{T_{\ast},0}(f^{T_{\ast}})=\min_T I_{T,0}(f^T)$.
We could obtain the geodesic between $h$ and $k$
under the Agmon distance $d_{U_0}^{Ag}$ 
by the reparametrization of $f^{T_{\ast}}$.
Of course, this kind of calculation is related with
the another representation of the Agmon distance which is given by
Carmona and Simon~{\rm \cite{cs}}.
\end{rem}

\begin{proof}
(1)~
Because $S_t$ is a $C_0$-contraction semigroup on $H^{1/2}$,
we have
$\lim_{t\to 0}\|f(t)-h\|_{H^{1/2}}=0$
and $\lim_{t\to T}\|f(t)-k\|_{H^{1/2}}=0$.
Note that
\begin{eqnarray*}
\partial_tf(t)&=&
\sqrt{m^2-\Delta}\left(S_{T-t}(I-S_{2T})^{-1}(k-S_{T}h)
-S_t(I-S_{2T})^{-1}(h-S_Tk)\right)\\
\sqrt{m^2-\Delta}f(t,\cdot)
&=&
\sqrt{m^2-\Delta}\left(S_{T-t}(I-S_{2T})^{-1}(k-S_{T}h)
+S_t(I-S_{2T})^{-1}(h-S_Tk)\right)
\end{eqnarray*}
Hence
\begin{eqnarray}
4I_{T,0}(f)&=&
2\int_0^T\left(
\tilde{A}^4S_{2(T-t)}(I-S_{2T})^{-2}
(k-S_Th), (k-S_Th)\right)_{L^2}dt\nonumber\\
& &
+2\int_0^T\left(
\tilde{A}^4S_{2t}(I-S_{2T})^{-2}
(h-S_Tk), (h-S_Tk)\right)_{L^2}dt\nonumber\\
&=&
\left((I-S_{2T})^{-1}(k-S_Th), (k-S_Th)\right)_H\nonumber\\
& &+\left((I-S_{2T})^{-1}(h-S_Tk), h-S_Tk\right)_H\nonumber\\
&=&
2\left((I-S_{2T})^{-1}(h-k), h-k\right)_H+
\left((I-S_T)(I+S_T)^{-1}h,h\right)_H\nonumber\\
& &+
\left((I-S_T)(I+S_T)^{-1}k,k\right)_H
-2\left((I+S_T)^{-1}(h-k),h-k\right)_H
\end{eqnarray}
These imply $f\in H^1_{T,h,k}(\RR)$.

\noindent
$(2)$~
We rewrite (\ref{fhkT}).
\begin{eqnarray}
f(t)&=&
(S_{T-t}-S_t)(I-S_{2T})^{-1}(k-h)+
S_{T-t}(I+S_T)^{-1}h+S_t(I+S_T)^{-1}k.
\end{eqnarray}
Since $\|(S_{T-t}-S_t)(I-S_{2T})^{-1}\|_{L(L^2,L^2)}\le 2$,
we obtain the desired result.
It is obvious that $f(t)\in H^n(\RR)$ for all $n\in {\mathbb N}$
because the image of $L^2$ by $S_t$~$(t>0)$ belongs to $H^n$.

\noindent
(3)~It suffices to consider $k$ such that
$\|k\|_H\le \|h\|_H+1$.
We estimate the distance using the upper bound
by $I_{T,P}(f)$ and the function $f$ in
(\ref{fhkT}) choosing $T$ appropriately small.
First, we consider the nonlinear term containing $P$.
Since $g$ is a continuous function with compact support, 
using the estimate in Lemma~\ref{agmon lemma}~(2),
we have
\begin{eqnarray}
\left|\iint_{(0,T)\times \RR}P(f(t,x))g(x)dtdx\right|
&\le&
\int_0^TC_g
\left(1+\sum_{k=2}^{2M}\|f(t)\|_{L^k(\RR)}^k\right)dt\nonumber\\
&\le&\int_0^T
C_g\left(1+\sum_{k=2}^{2M}\|f(t)\|_{H^{1/2}(\RR)}^k\right)dt\nonumber\\
&\le& C_g\left(1+\|h\|_H+\|k\|_H\right)^{2M}T.
\label{non-linear term}
\end{eqnarray}
Hence, by setting $T\le T(\ep,h):=\left(2^{2M}C_g\right)^{-1}
\left(1+\|h\|_H)\right)^{-2M}\ep/3$,
we get 
\begin{equation}
\left|\iint_{(0,T)\times \RR}P(f(t,x))g(x)dtdx\right|\le \frac{\ep}{3}.
\end{equation}
Next, we estimate $I_{T,0}(f)$.
Because $h\in H$, there exists $T\in (0,1)$ such that
$T\le T(\ep,h)$ and
$|((I-S_T)(I+S_T)^{-1})h,h)_H|\le\ep$.
By the identity (\ref{IT0f}),
we have
\begin{eqnarray}
I_{T,0}(f)&\le&
\frac{1}{2(1-e^{-2Tm})}\|h-k\|_H^2+\frac{1}{2}\ep+
\frac{1}{4}\|h-k\|_H^2+\frac{1}{2}\|h-k\|_H\|h\|_H.
\end{eqnarray}
Therefore, taking $\|h-k\|_H$ sufficiently
small, we obtain $I_{T,0}(f)\le 2\ep/3$.
All the estimates above imply the desired result.

\noindent
$(4)$~We have $\sup_n\|h_n\|_{H^1}<\infty$.
Hence there exists a subsequence
$\{h_{n(k)}\}$ which converges weakly to
some $h_{\ast}\in H^1$.
By Lemma~\ref{gagliard-nirenberg estimate},
$\lim_{n\to\infty}V(h_{n(k)})=V(h_{\ast})$.
On the other hand, $U_0(h_{\ast})\le
\liminf_{n\to\infty}U_0(h_{n(k)})$.
Since $U$ is non-negative, we have
$h_{\ast}\in {\cal Z}$,
$\lim_{n\to\infty}\|h_n\|_{H^1}=\|h_{\ast}\|_{H^1}$
and $\lim_{n\to\infty}\|h_n\|_{L^2}=\|h_{\ast}\|_{L^2}$.
Again by Lemma~\ref{gagliard-nirenberg estimate},
if necessary, by taking a subsequence,
$h_{n(k)}(x)\to h_{\ast}(x)$~a.e.$x$.
These imply $\lim_{n\to\infty}\|h_n-h_{\ast}\|_{H^1}=0$.

\noindent
$(5)$~Let $\delta>0$.
There exist $0<\delta_1<\delta_2<\delta$ such that
$D\cap {\cal Z}=\emptyset$,
where $D=\{k~|~\delta_1\le \|k-h\|_{L^2}\le \delta_2\}$.
By the result in (4),
$\delta_3:=
\inf\{U(\varphi)~|~\varphi\in D\}>0.$
Let us choose $k\in H$ such that
$\|h-k\|_{L^2}\ge \delta$
and $c\in {\cal P}_{1,k,h,U}$.
Then there exist times $0<\tau_1<\tau_2<1$
such that
$c(t)\in D$ for $\tau_1\le t\le\tau_2$
and $\|c(\tau_1)-h\|_{L^2}=\delta_1,
\|c(\tau_2)-k\|_{L^2}=\delta_2$.
Hence
$\ell(c)\ge \int_{\tau_1}^{\tau_2}\sqrt{U(c(t))}\|c'(t)\|_{L^2}dt
\ge \sqrt{\delta_3}(\delta_2-\delta_1)$.
This completes the proof.

\noindent
$(6)$~
Let $c\in {\cal P}_{1,h,k,U_0}$.
Then 
\begin{eqnarray}
\int_0^1\sqrt{U_0(c(t))}\|c'(t)\|_{L^2}dt
&=&\frac{1}{2}\int_0^1\|c(t)\|_{H^1}\|c'(t)\|_{L^2}dt\nonumber\\
&\ge&\frac{1}{2}\int_0^1\left(Ac(t),A^{-1}c'(t)\right)_Hdt\nonumber\\
&=&\frac{1}{2}\int_0^1\left(c(t),c'(t)\right)_Hdt\nonumber\\
&=&\frac{1}{4}\left(\|k\|_H^2-\|h\|_H^2\right)\label{k and h}
\end{eqnarray}
which implies (\ref{agmon h k}).
The expression of the second and third equation on the right-hand side
may be rough but
the final estimate is true by an approximation argument.

\noindent
(7)~We need only to prove
$\|h-c(t)\|_{L^2}\le \ep$.
If $\|c(t)-h\|_{L^2}>\ep$,
then there exists $0<s_{\ast}<t_{\ast}<t$ such that
$\|c(s_{\ast})-h\|_{L^2}=\ep/2$,
$\inf_{s_{\ast}\le s\le t_{\ast}}\|c(s)-h\|_{L^2}\ge \ep/2$,
$\inf_{s_{\ast}\le t\le t_{\ast}}\|c(s)-k\|_{L^2}\ge \ep$
and $\|c(t_{\ast})-h\|_{L^2}=\ep$.
We have
\begin{eqnarray}
\int_{s_{\ast}}^{t_{\ast}}U(c(s))\|c'(s)\|_{L^2}^2ds
&\ge&
\int_{s_{\ast}}^{t_{\ast}}\delta(\ep)\|c'(s)\|_{L^2}^2ds\nonumber\\
&>&\delta(\ep)\left(\int_{s_{\ast}}^{t_{\ast}}\|c'(s)\|_{L^2}ds\right)^2
\frac{1}{t_{\ast}-s_{\ast}}\nonumber\\
&=&
\frac{\ep^2\delta(\ep)}{4(t_{\ast}-s_{\ast})}.
\end{eqnarray}
Hence, $t>\frac{\ep^2\delta(\ep)}{4e_U(c)}$.
This implies the desired estimate.
\end{proof}

\begin{proof}[Proof of Theorem~$\ref{Theorem 1 for Agmon}$]
First we prove $\dUAg(h,k)=0$
is equivalent to $h=k$.
Assume $h=k$.
Let $f=f(t,x)$ be 
the function in (\ref{fhkT}).
Then
\begin{eqnarray}
\dUAg(h,h)&\le&
I_{T,0}(f)+\iint_{(0,T)\times \RR}
P(f(t,x))g(x)dtdx\nonumber\\
&\le&
\frac{1}{2}
\left(
(I-S_T)(I+S_T)^{-1}h,h
\right)_H+
\iint_{(0,T)\times \RR}
P(f(t,x))g(x)dtdx.
\end{eqnarray}
By the same calculation as before,
we have
\begin{eqnarray}
\left|\iint_{(0,T)\times \RR}P(f(t,x))g(x)dtdx\right|
&\le& C_gT\left(1+\|h\|_H+\|k\|_H\right)^{2M}.
\end{eqnarray}
Letting $T\to 0$, we get $\dUAg(h,h)=0$.
Next, we show $\dUAg(h,k)>0$ if $h\ne k$.
Let ${\mathcal Z}'$ be the union of ${\mathcal Z}$ and
$h,k$.
For any $\ep>0$, we have
\begin{equation}
\delta(\ep):=\inf\left\{U(\varphi)~|~
\min_{\phi\in {\mathcal Z}'}\|\varphi-\phi\|_{L^2}>\ep\right\}>0.
\end{equation}
Let
$r_0=\min\left\{\|\varphi-\phi\|_{L^2}~|~
\varphi,\phi\in {\mathcal Z}', \varphi\ne \phi\right\}$
and set
$
\ep<\frac{1}{3}r_0.
$
Let $c\in {\cal P}_{T,h,k}$.
Since $c=c(t)$ is a continuous curve on $L^2$,
there exist times $0<\tau_1<\tau_2<T$ and distinct
$\psi_1, \psi_2\in {\mathcal Z}'$ such that
$$
\min_{\varphi\in {\mathcal Z}'}\|c(t)-\varphi\|_{L^2}\ge \ep
\quad (\tau_1\le t\le \tau_2),\quad
\|c(\tau_1)-\psi_1\|_{L^2}=
\|c(\tau_2)-\psi_2\|_{L^2}=\ep.
$$
We have
\begin{eqnarray*}
\int_{\tau_1}^{\tau_2}
\sqrt{U(c(t))}\|c'(t)\|_{L^2}dt
&\ge&
\sqrt{\delta(\ep)}\int_{\tau_1}^{\tau_2}
\|c'(t)\|_{L^2}dt
\ge
\sqrt{\delta(\ep)}\|c(\tau_1)-c(\tau_2)\|_{L^2}
\ge\frac{\sqrt{\delta(\ep)}r_0}{3}
\end{eqnarray*}
which implies $\dUAg(h,k)>0$.
The relation $\dUAg(h,k)=\dUAg(k,h)$ is trivial.
We prove the triangle inequality
$\dUAg(h,l)\le \dUAg(h,k)+\dUAg(k,l)$
for any $h,k,l\in H$.
Let $c_1\in {\cal P}_{1,h,k}$ and $c_2\in {\cal P}_{1,k,l}$.
Define an path
$\eta=\eta(t)$ by
$\eta(t)=c_1(t)$~$(0\le t\le 1)$,
$\eta(t)=c_2(t-1)$~$(1\le t\le 2)$.
Then $\eta\in {\cal P}_{2,h,l}$.
Hence we have
\begin{eqnarray}
\dUAg(h,l)&\le&
\int_0^2\sqrt{U(\eta(t))}\|\eta'(t)\|_{L^2}dt\nonumber\\
&=&
\int_0^1\sqrt{U(c_1(t))}\|c_1'(t)\|_{L^2}dt+
\int_0^1\sqrt{U(c_2(t))}\|c_2'(t)\|_{L^2}dt
\end{eqnarray}
which implies the triangle inequality.
Let us fix $h\in H$.
By Lemma~\ref{agmon lemma}~(3),
it suffices to prove that if $\lim_{n\to\infty}\dUAg(h,h_n)=0$, then
$\lim_{n\to\infty}\|h_n-h\|_H=0$.
First consider the case $h\notin {\cal Z}$.
Then there exist $c_n\in {\cal P}_{1, h_n, h,U}$
and $\delta>0$ such that
$U(c_n(t))\ge \delta$ for all $t\in [0,1]$
and $\lim_{n\to\infty}\ell(c_n)=0$.
This follows from Lemma~\ref{agmon lemma}~(4) and
$d_U^{Ag}(h,k)>0$ for all $k\in {\cal Z}$.
Let $\inf_{h\in H^1} V(h)=:-R$.
We have $U(h)+R\ge U_0(h)$ for all $h\in H^1$.
Hence if $U(h)\ge \delta>0$, then
$U(h)\ge \frac{\delta}{R+\delta}U_0(h)$.
Using this, we have
\begin{eqnarray}
\ell(c_n)&=&\int_0^1\sqrt{U(c_n(t))}\|c_n'(t)\|_{L^2}dt\nonumber\\
&\ge&
\sqrt{\frac{\delta}{R+\delta}}\int_0^1\sqrt{U_0(c_n(t))}\|c_n'(t)\|_{L^2}dt
\nonumber\\
&\ge&
\frac{1}{4}\sqrt{\frac{\delta}{R+\delta}}\left|\|c_n(1)\|_H^2-\|c_n(0)\|_H^2\right|
\nonumber\\
&=&
\frac{1}{4}\sqrt{\frac{\delta}{R+\delta}}
\left|\|h_n\|_H^2-\|h\|_H^2\right|.
\end{eqnarray}
This and Lemma~\ref{agmon lemma}~(5) imply
$\lim_{n\to\infty}\|h_n-h\|_H=0$.
Let us consider the case where
$h\in {\cal Z}$.
Take $h_n$ such that
$\lim_{n\to\infty}\dUAg(h_n,h)=0$.
Then $\lim_{n\to\infty}\|h_n-h\|_{L^2}=0$.
Let us choose a sufficiently small positive number $\ep$.
Then by the nondegeneracy of the second derivative of $U$ at $h$,
there exists a positive number $\delta(\ep)$ such that 
$U(\varphi)\ge \delta(\ep)\|\varphi-h\|_{H^1}^2$
for any $\varphi\in H^1(\RR)\cap \{\varphi~|~\|\varphi-h\|_{L^2}<\ep\}$.
We find a curve $c_n\in {\cal P}_{1,h_n,h,U}$
such that $\sup_{n,t}\|c_n(t)-h\|_{L^2}<\ep$ and
$\ell(c_n)\to 0$ as $n\to\infty$.
Thus,
\begin{eqnarray}
\ell(c_n)&=&\int_0^1
\sqrt{U(c_n(t))}\|c_n'(t)\|_{L^2}dt\nonumber\\
&\ge&
\frac{\sqrt{\delta(\ep)}}{2}
\int_0^1\|c_n(t)-h\|_{H^1}\|c_n'(t)\|_{L^2}dt\nonumber\\
&\ge&
\frac{\sqrt{\delta(\ep)}}{2}
\|h_n-h\|_H^2\to 0 \quad \mbox{as $n\to \infty$}
\end{eqnarray}
which completes the proof.

\noindent
(2)~Let $S_t$ be the Cauchy semigroup as in Lemma~\ref{agmon lemma}.
Let $c(t)=S_{T-t}h$.
Then we have 
$
U_0(c(t))=\frac{1}{4}\|c^{\prime}(t)\|_{L^2}^2.
$
Therefore,
\begin{eqnarray*}
d_{U_0}^{Ag}(S_Th,h)&\le&
\int_0^T\frac{1}{2}\|c^{\prime}(t)\|_{L^2}^2dt\\
&=&\frac{1}{2}\int_0^T\left(\tilde{A}^4e^{2(T-t)\tilde{A}^2}h,h\right)
_{L^2}dt
=\frac{1}{4}\left(
\|h\|_H^2-\left\|S_{T}h\right\|_H^2
\right).
\end{eqnarray*}
Since 
$\lim_{T\to \infty}\|S_Th\|_H=0$,
$\lim_{T\to \infty}d_{U_0}^{Ag}(S_Th,h)=d_{U_0}^{Ag}(0,h)$.
Consequently, we get
$d_{U_0}^{Ag}(0,h)\le \frac{1}{4}\|h\|_H^2$.
Combining this estimate and (\ref{agmon h k}),
we obtain the desired result.
\end{proof}

Next we prove Theorem~\ref{Theorem 2 for Agmon}.

\begin{lem}\label{energy of path 2}
Under the same assumption as in Theorem~$\ref{Theorem 2 for Agmon}$,
there exists  $c_{\star}\in {\cal P}^{loc}_{1,h,k}$
such that $c_{\star}(t)\notin {\cal Z}$ for all
$0<t<1$ and
\begin{equation}
U(c_{\star}(t))\|c_{\star}'(t)\|_{L^2}^2=\dUAg(h,k)^2=
\ell(c_{\star})^2=e(c_{\star}).
\qquad
a.e.~t.\label{energy}
\end{equation}
\end{lem}

This lemma shows the existence of the minimizer
$c_{\star}$ which attains $\dUAg(h,k)$ and
the result of (1) in Theorem~\ref{Theorem 2 for Agmon}.
We prove the other properties in Theorem~\ref{Theorem 2 for Agmon}
in the next subsection because
they are related with instanton.

\begin{proof}
Since $h,k\in H^1(\RR)$,
by Lemma~\ref{two agmon},
there exist $\{c_n\}_{n=1}^{\infty} \subset AC_{1,h,k}(H^1(\RR))$ 
such that
$\ell(c_n)^2\le\dUAg(h,k)^2+\frac{1}{n}$.
We may assume that
$c_n(t)\notin {\cal Z}$ for $0<t<1$.
By reparametrizing of the paths, we see that
there exist $\{c_n\}_{n=1}^{\infty}\subset {\cal P}^{loc}_{1,h,k}$
such that $c_n(t)\notin {\cal Z}$, $c_n$
is a continuous path in $H^1(\RR)$ 
and
\begin{equation}
\sqrt{U(c_n(t))}\|c_n'(t)\|_{L^2}=\ell(c_n)=\sqrt{e(c_n)}
\le\sqrt{d_U^{Ag}(h,k)^2+\frac{1}{n}}
\qquad a.e.~t.\label{constant speed}
\end{equation}
Also we may assume that
\begin{itemize}
\item[(i)]~$\sup_{n,t}\|c_n(t)\|_{H}<\infty$
\item[(ii)]~Let $\tau_n$ be the maximum
time such that $\|c(\tau_n)-h\|_{H^1}\le\ep$
and $\tilde{\tau}_n$ be the minimum time 
such that $\|c_n(t)-k\|_{H^1}\le \ep$.
Then there exists a constant $C_i>0$ such that
\begin{equation}
\max\left\{\tau_n,1-\tilde{\tau}_n\right\}
\le C_1\ep^2 \qquad \mbox{for all $n\ge C_2\ep^{-2}$}.
\end{equation}
\end{itemize}
The boundedness in (i) follows from 
the result that
$\lim_{\|\varphi\|_H\to\infty}\dUAg(h,\varphi)=+\infty$.
This result can be shown by a similar argument in
(\ref{k and h}).
We prove $\tau_n=O(\ep^2)$ for
$n\ge C_2\ep^{-2}$.
The proof for $\tilde{\tau}_n$ is similar to it.
Let us define
a curve $\tilde{c}_n=\tilde{c}_n(t)$ by
$\tilde{c}_n(t)=h+\frac{3t}{\tau_n}\left(c_n(\tau_n)-h\right)$
for $0\le t\le \frac{\tau_n}{3}$
and
$\tilde{c}_n(t)=c_n(\chi_n(t))$
for $\frac{\tau_n}{3}\le t\le 1$,
where
$$
\chi_n(t)=\left(
\frac{3-3\tau_n}{3-\tau_n}t+
\frac{2\tau_n}{3-\tau_n}\right).
$$
Since $e(\tilde{c}_n)\ge e(c_n)-\frac{1}{n}$,
we get
\begin{equation}
\frac{1}{3}e(c_n)\tau_n\le
\frac{1}{n}+\frac{C\ep^4}{\tau_n}.
\end{equation}
Hence we have
$$
\frac{1}{3}e(c_n)\tau_n\le \frac{2}{n}
\qquad \mbox{or}
\qquad
\frac{1}{3}e(c_n)\tau_n\le
\frac{2C\ep^4}{\tau_n}
$$
which implies (ii).
By (ii), 
for any $\delta>0$, there exists a natural number
$N(\delta)$ and large positive number $R(\delta)$
such that for any $n\ge N(\delta)$
it holds that
$$
\|c_n'(t)\|_{L^2}\le R(\delta)\qquad \mbox{for $\delta<t<1-\delta$}.
$$
Hence there exists a bounded measurable path 
$c :[0,1]\to H$ such that
\begin{eqnarray}
\mbox{{\rm (i)}}& &\mbox{If necessary, by taking a subsequence,
$c_n(t)\to c(t)$ weakly in $H$ for any $t$}
\phantom{llllllllllllll}\label{property of c 1}\\
\mbox{{\rm (ii)}}& &
\mbox{$c$ is a locally Lipschitz path on $L^2(\RR)$ such that}\nonumber\\
& &\mbox{$\|c'(t)\|_{L^2}\le R(\delta)$ for almost every
$t\in [\delta,1-\delta]$.}\label{property of c 2}\\
\mbox{{\rm (iii)}}& &
\mbox{It holds that $c(t)\in H^1(\RR)$ for a.e. $t\in S$, where
$S=\{t\in [0,1]~|~c'(t)\ne 0\}$.}\label{property of c 3}
\end{eqnarray}
We prove the above properties.
The item (i) follows from the locally uniform Lipschitz continuity of
$c_n$ in $L^2$ and the uniform boundedness of $c_n$ in $H$.
We prove (ii) and (iii).
Let $0<a<b<1$.
Let $\varphi=\varphi(t)$~$(0\le t\le 1)$ be a $C^1$ path
on $L^2(\RR)$ with $\varphi(t)=0$ for
$t\in (a,b)^c$.
Then
$\lim_{n\to\infty}\int_a^b\left(c_n'(t),\varphi(t)\right)_{L^2}dt=
\int_a^b\left(c'(t),\varphi(t)\right)_{L^2}dt$.
This implies $c_n'$ converges to $c'$ weakly in
$L^2((a,b)\to L^2(\RR),dt)$ for any $a,b$.
Hence by reverse Fatou's lemma, we get
\begin{equation}
\int_a^b\|c'(t)\|_{L^2}^2dt
\le\liminf_{n\to\infty}\int_a^b
\|c_n'(t)\|_{L^2}^2dt\nonumber\\
\le\int_a^b\limsup_{n\to\infty}\|c_n'(t)\|_{L^2}^2dt
\end{equation}
which implies (ii).
Also we have
$\limsup_{n\to\infty}\|c_n'(t)\|_{L^2}\ge \|c'(t)\|_{L^2}$
for almost every $t\in [0,1]$.
Therefore for almost every $t\in S$, $\limsup_{n\to\infty}\|c_n'(t)\|_{L^2}>0$
holds. This and (\ref{constant speed}) implies (iii).
In view of (i) and Lemma~\ref{agmon lemma}~(7),
we obtain 
\begin{equation}
\lim_{t\to 0}\|c(t)-h\|_{L^2}=\lim_{t\to 1}\|c(t)-k\|_{L^2}=0.
\label{cthk}
\end{equation}
We estimate
$
U(c(t))\|c'(t)\|_{L^2}.
$
For any $0\le a\le b\le 1$,
\begin{eqnarray}
\int_a^bU(c(t))\|c'(t)\|_{L^2}^2dt&\le&
\int_a^b
\left(\liminf_{n\to\infty}U(c_n(t))\right)
\left(\limsup_{n\to\infty}\|c_n'(t)\|_{L^2}^2dt\right)dt\nonumber\\
&\le&\int_a^b\limsup_{n\to\infty}\left(U(c_n(t))\|c_n'(t)\|^2\right)dt=
\dUAg(h,k)^2(b-a).
\end{eqnarray}
Thus we obtain $U(c(t))\|c'(t)\|_{L^2}^2\le \dUAg(h,k)^2$ 
for $a.e.\, t\in [0,1]$.
We prove the following:
\begin{eqnarray}
\mbox{{\rm (iv)}}& & \mbox{$c=c(t)$~$(0<t<1)$ is a continuous path on $H$.}\\
\mbox{{\rm (v)}}& & \mbox{
$\lim_{t\to 0}\|c(t)-h\|_H=0$ and $\lim_{t\to 1}\|c(t)-k\|_H=0$.}\\
\mbox{{\rm (vi)}}& &
\mbox{$c(t)\notin {\cal Z}$ for all $t\in (0,1)$}
\end{eqnarray}
Let $0<\delta\le s<t\le 1-\delta<1$.
By an argument similar to
(\ref{k and h}),
we have $\left|\|c(t)\|_H^2-\|c(s)\|_H^2\right|\le C(\delta)|t-s|$.
This and the continuity of $c$ in $L^2$ implies
(iv).
We prove (v).
It suffices to consider the case where $t$ converges to $0$.
If $\int_0^{\ep}\|c'(t)\|_{L^2}dt<\infty$,
again by an similar argument to (\ref{k and h}),
we obtain the convergence $\lim_{t\to 0}\|c(t)\|_H=\|h\|_H$
which implies the assertion.
If it is not the case, there exists a decreasing sequence
$t_n\downarrow 0$ such that $U(c(t_n))\to 0$.
By (\ref{cthk}), this implies $\lim_{n\to\infty}\|c(t_n)-h\|_{H^1}=0$.
Noting $\dUAg(c(t),c(s))\le \dUAg(h,k)|t-s|$ for any
$0<s<t<1$, we obtain 
\begin{equation}
\dUAg(h,c(t))\le \dUAg(h,c(t_n))+\dUAg(c(t_n),c(t))
\le
\dUAg(h,c(t_n))+\dUAg(h,k)|t-t_n|.
\end{equation}
This shows $\lim_{t\to 0}\dUAg(h,c(t))=0$
and we complete the proof of (v).
Consequently, we have
$c\in {\cal P}^{loc}_{1,h,k}$
and by the definition of
$\dUAg$,
$U(c(t))\|c'(t)\|_{L^2}^2=\dUAg(h,k)^2~a.e.~t$.
We prove (vi).
If there exists a time $0<t<1$ such that
$c(t)\in {\cal Z}$, we can construct a
path belonging to ${\cal P}^{loc}_{1,h,k}$
whose length is smaller than that of $c$.
This is a contradiction
and we see that the above $c$ is a desired path $c_{\star}$.
\end{proof}

\subsection{Instanton}\label{instanton}

In this subsection, we assume $U$ satisfies
the assumptions (A1), (A2) and
${\cal Z}$ consists of two points $\{h,k\}$.
We do not assume (A3).
So far, we consider paths on function spaces 
defined on the time interval $[0,T]$.
However, it is convenient to consider paths
defined on $[-T,T]$ to discuss instanton.
In this subsection,
we denote by $H^1_{T,h,k}(\RR)$ the set of functions
$u=u(t,x)$~$((t,x)\in (-T,T)\times \RR)$ 
which belongs to $H^1((-T,T)\times \RR)$ with
$u(-T,\cdot)=h(\cdot)$ and $u(T,\cdot)=k(\cdot)$.
Accordingly,
we define $I_{T,P}(u)$ for $u\in H^1_{T,h,k}(\RR)$ similarly.
We note that 
the Agmon distance has another equivalent form
which is due to Carmona and Simon~$\cite{cs}$
in the case of finite dimensional Schr\"odinger operators.
The functional
$I_{T,P}(u)$ is the action integral of 
the classical dynamics given by
\begin{equation}
\frac{\partial^2 u}{\partial t^2}(t,x)=2(\nabla U)(u(t,x))
\label{imaginary time equation}
\end{equation}
which is obtained by changing the time $t$ to the imaginary
time $\sqrt{-1}t$ in the Klein-Gordon equation~$(\ref{nlkg})$.
For $u=u(t,x)$, we define
\begin{eqnarray}
I_{\infty,P}(u)&=&
\frac{1}{4}
\int_{-\infty}^{\infty}\|\partial_tu(t)\|_{L^2(\RR)}^2dt
+\int_{-\infty}^{\infty}U(u(t))dt.
\end{eqnarray}
The Agmon distance $\dUAg(h, k)$
is related with the minimizer of
the action integral $I_{\infty,P}$
with the condition
$u(-\infty,\cdot)=h, u(+\infty,\cdot)=k$.
The minimizer is called an instanton.
Simon~\cite{simon4} used a path integral approach 
in tunneling estimate in which
the relation between the Agmon distance and the instanton
is used.
The equation (\ref{imaginary time equation})
reads
\begin{equation}
\frac{\partial^2 u}{\partial t^2}(t,x)+\frac{\partial^2 u}{\partial x^2}(t,x)
=m^2u(t,x)
+2P'(u(t,x))g(x)
\label{instanton equation}
\end{equation}
Let $T>0$.
We write
\begin{equation}
{\mathcal I}(T)=
\inf\left\{I_{T,P}(u)~|~u\in H^1_{T,h,k}(\RR)\right\}.\label{minimizing}
\end{equation}
Note that the critical point $u$ of the functional $I_{T,P}$ on 
$H^1_{T,h,k}(\RR)$
satisfies
the equation (\ref{instanton equation}) on
$(-T,T)\times \RR$.
We prove the existence of an instanton and the action integral is
equal to the Agmon distance between $h$ and $k$.

\begin{thm}\label{existence of instanton}
There exists a solution $u_{\star}=u_{\star}(t,x)$~$(t,x)\in \RR^2$
to the equation $(\ref{instanton equation})$
which satisfies the following properties.

\noindent
$(1)$~It holds that 
$u_{\star}|_{(-T,T)\times \RR}\in H^1\left((-T,T)\times \RR\right)
\cap C^{\infty}((-T,T)\times \RR)$
for any $T>0$
and
$I_{T,P}(u_{\star}|_{(-T,T)\times \RR})=\inf\{I_{T,P}(u)~|~u\in 
H^1_{T,u_{\star}(-T),u_{\star}(T)}(\RR)\}$.

\noindent
$(2)$~We have $I_{\infty,P}(u_{\star})=\dUAg(h,k)$ and
$u_{\star}$ is a minimizer of the functional
$I_{\infty,P}$ in the set of functions $u$ satisfying the
following conditions:
\begin{itemize}
\item[{\rm (i)}]~$u|_{(-T,T)\times \RR}\in H^1((-T,T),\RR)$ for all $T>0$,
\item[{\rm (ii)}]~$\lim_{t\to-\infty}\|u(t)-h\|_H=0$ and
$\lim_{t\to\infty}\|u(t)-k\|_H=0$.
\end{itemize}

\noindent
$(3)$~$\lim_{T\to\infty}{\mathcal I}(T)=d_U^{Ag}(h, k)$.
\end{thm}

We need a lemma.

\begin{lem}\label{functional cal I}
Let us consider the functional
${\mathcal I}(T)$.

\noindent
$(1)$~There exists a minimizer $u^T\in H^1_{T,h,k}(\RR)$
such that ${\mathcal I}(T)=I_{T,P}(u^T)$

\noindent
$(2)$~${\mathcal I}$ is a strictly decreasing function of $T$.
\end{lem}

\begin{proof}[Proof of Lemma~$\ref{functional cal I}$]

(1)~ This can be proved by a standard method and we omit the proof.

(2)~Let $T'>T>0$ and suppose
${\mathcal I}(T')={\mathcal I}(T)$.
Let $u^T$ be a minimizer of the minimizing problem (\ref{minimizing}).
Let $\tilde{u}^{T'}$ be a function in $H^1_{T',h, k}(\RR)$
such that
$
\tilde{u}^{T'}(t)=
h$~$(-T'\le t\le -T)$,
$
\tilde{u}^{T'}(t)=
u^T(t)$~$(-T<t<T)$,
$
\tilde{u}^{T'}(t)=
k
$~$(T\le t\le T')$.
Then $\tilde{u}^{T'}$ is a minimizer for
(\ref{minimizing}) replacing $T$ by $T'$.
Let us define $\tilde{u}^{T'}_0(t,x)=h(x)$
~$(t,x)\in (-T',T')\times \RR$.
Then both functions $\tilde{u}^{T'}, \tilde{u}_0^{T'}$ are
solutions to (\ref{instanton equation}) on
$(-T',T')\times \RR$
with $u(-T',x)=h(x)$.
The difference $v^{T'}=\tilde{u}^{T'}-\tilde{u}^{T'}_0$
is an eigenfunction of
a Schr\"odinger operator
satisfying the boundary condition
$v^{T'}(-T')=0$ and $v^{T'}(T')=k-h$.
Also $v^{T'}(t,x)\equiv 0$ for all $(t,x)\in (-T',T)\times \RR$.
By the unique continuation theorem for the solution, we obtain
$v^{T'}\equiv 0$ which is a contradiction.
This shows that ${\mathcal I}$ is a strictly decreasing function.

\end{proof}

\begin{proof}[Proof of Theorem~$\ref{existence of instanton}$ 
and Theorem~$\ref{Theorem 2 for Agmon}$~$(2), (3)$]
Let $c=c(t,x)$ be the geodesic path $c_{\star}$ in
Lemma~\ref{energy of path 2}.
We construct $u_{\star}$ by reparametrizing the time parameter
of $c$.
Let 
$$
\rho(t)=\int_{1/2}^t\frac{\|c'(s)\|_{L^2}}{2\sqrt{U(c(s))}}
=\frac{1}{2\dUAg(h,k)}\int_{1/2}^t\|c'(s)\|_{L^2}^2ds
\quad\quad 0<t<1.
$$
Then $\rho$ is a strictly increasing absolutely continuous
function.
Define $\sigma(t)=\rho^{-1}(t)$~$(\rho(0+)<t<\rho(1-))$.
We prove $\rho(0+)=-\infty$ and $\rho(1-)=+\infty$.
To this end, we
set $u(t)=c(\sigma(t))$
which will turn out to be the desired $u_{\star}$.
We have $\|u'(t)\|_{L^2}=2\sqrt{U(u(t))}$.
Therefore for any $\rho(+0)<t_1<t_2<\rho(1-)$
\begin{eqnarray}
\int_{\sigma(t_1)}^{\sigma(t_2)}
\sqrt{U(c(s))}\|c'(s)\|_{L^2}ds&=&
\int_{t_1}^{t_2}\sqrt{U(u(t))}\|u'(t)\|_{L^2}dt\nonumber\\
&=&\int_{t_1}^{t_2}\left(\frac{1}{4}\|u'(t)\|_{L^2}^2+
U(u(t))\right)dt.\label{c and u}
\end{eqnarray}
Suppose $\rho(+0)>-\infty$.
Then we can set $t_1=\rho(+0)$ and
$u(t_1)=h$.
Then by an argument similar to the proof of 
Lemma~\ref{functional cal I}~(2), for any $\ep>0$,
we can find $\tilde{u}=\tilde{u}(t,x)$~$(t_1-\ep<t<t_2)$
such that $\tilde{u}(t_1-\ep)=h$,
$\tilde{u}(t_2)=u(t_2)$ and
$$
\int_{t_1-\ep}^{t_2}\sqrt{U(\tilde{u}(t))}\|\tilde{u}'(t)\|_{L^2}dt
\le
\int_{t_1-\ep}^{t_2}
\left(\frac{1}{4}\|\tilde{u}'(t)\|_{L^2}^2+
U(\tilde{u}(t))\right)dt
<\int_{t_1}^{t_2}\sqrt{U(u(t))}\|u'(t)\|_{L^2}dt.
$$
Since $\tilde{u}(t_2)=u(t_2)=c(\sigma(t_2))$,
by connecting the two paths at $u(t_2)$, we obtain
the shorter path between $h$ and $k$ than
$c$. This is a contradiction.
Therefore we get $\rho(+0)=-\infty$ and
$\rho(1-)=+\infty$ similarly.
This proves Theorem~\ref{Theorem 2 for Agmon}~(3)
and $I_{\infty,P}(u)=\dUAg(h,k)$.
Clearly, this $u$ satisfies the conditions (i), (ii) in (2)
in Theorem~\ref{existence of instanton}.
Also if $u$ satisfies (i) and (ii), then
$I_{\infty,P}(u)\ge I_{T,P}(u)\ge \dUAg(u(-T),u(T))\to
\dUAg(h,k)$.
Hence $u$ is the minimizer in the sense of
Theorem~\ref{existence of instanton} (2).
Note that $u|_{(-T,T)\times \RR}$ is the minimizer
of $I_{T,P}$ on
$H^1_{T,u(-T),u(T)}(\RR)$ because of the identity
(\ref{c and u}) and the fact that the length of $c$ is 
shortest.
This proves Theorem~\ref{existence of instanton} (1).
Now we prove Theorem~\ref{existence of instanton}~(3).
Let $\ep$ be a small positive number.
Take a large $T$ such that $\|h-u(-T)\|_H$ and
$\|k-u(T)\|_H$ are sufficiently small.
Then
there exist $f_1=f_1(t,x)$~$(0<t<\ep_1)$ and 
$f_2=f_2(t,x)$~$(0<t<\ep_2)$ which are defined by
the Cauchy semigroup as in Lemma~\ref{agmon lemma} such that
\begin{itemize}
\item[(i)]~$f_1(0)=h,\, f_1(\ep_1)=u(-T)$,~~~$f_2(0)=u(T),\, f_2(\ep_2)=k$,
\item[(ii)]~$I_{\ep_1,P}(f_1)<\ep$, $I_{\ep_2,P}(f_2)<\ep$.
\end{itemize}
Hence there exists $v\in H^1_{T+(\ep_1+\ep_2)/2,h,k}(\RR)$
such that
\begin{equation}
I_{T+(\ep_1+\ep_2)/2,h,k}(v)\le \dUAg(h,k)+2\ep
\end{equation}
which implies the assertion.
It remains to prove the statement (2) in Theorem~\ref{Theorem 2 for Agmon}.
Note that
\begin{eqnarray}
\int_{-\infty}^t\|u'(s)\|_{L^2}^2ds&=&
\int_{-\infty}^{t}2\sqrt{U(u(s))}\|u'(s)\|_{L^2}ds\nonumber\\
&=&\int_{-\infty}^t2\sqrt{U(c(\sigma(s)))}\|c'(\sigma(s))\|_{L^2}
\sigma'(s)ds\nonumber\\
&=&\int_{0}^{\sigma(t)}2\sqrt{U(c(s))}\|c'(s)\|_{L^2}ds=2d_U^{Ag}(h,k)\sigma(t).
\end{eqnarray}
Therefore
$2d_U^{Ag}(h,k)\sigma'(t)=\|u'(t)\|_{L^2}^2$.
Since $\|u'(t)\|_{L^2}=2\sqrt{U(u(t))}>0$ for all
$t$, $\sigma'(t)>0$
for all $t$.
The function $t(\in \RR)\mapsto \|u'(t)\|^2_{L^2}$
is a $C^{\infty}$ function and so $\sigma$ and $\rho$ are.
Since $u(\rho(t),x)=c(t,x)$,
we complete the proof.
\end{proof}

Let us consider a simple example in the case where
the space is the finite interval $I=[-l/2,l/2]$.
Let $a$ and $x_0$ be positive numbers.
We consider the case where
$$
U(h)=\frac{1}{4}\int_Ih'(x)^2dx+a\int_I\left(h(x)^2-x_0^2\right)^2dx.
$$
For example, setting $b^2=x_0^2+\frac{m^2}{8a}$
and
$$
P(x)=a(x^2-b^2)^2-
a\left(b^4-\left(b^2-\frac{m^2}{8a}\right)^2\right)^2,
$$
we obtain the potential function above.
Note ${\cal Z}=\{h_0, -h_0\}$,
where $h_0(x)\equiv x_0$ is a constant function.
$\pm x_0$ are the zero points also of
the potential function 
$$
Q(x)=a(x^2-x_0^2)^2 \quad x\in \RR.
$$
Let
\begin{eqnarray*}
d_{1dim}^{Ag}(-x_0,x_0)&=&
\inf\Bigl\{
\int_{-T}^T\sqrt{Q(x(t))}|x'(t)|dt~\Big |~
\quad x(-T)=-x_0,~~x(T)=x_0
\Bigr\}.
\end{eqnarray*}
This is the Agmon distance which corresponds to
$1$-dimensional Schr\"odinger operator $-\frac{d^2}{dx^2}+Q(x)$
defined in $L^2(\RR,dx)$
and 
\begin{equation}
d_{1dim}^{Ag}(-x_0,x_0)=
\int_{-x_0}^{x_0}\sqrt{Q(x)}dx=\frac{4\sqrt{a}x_0^3}{3}.
\end{equation}

We can prove the following.

\begin{pro}\label{example of instanton}
Assume 
$\displaystyle{2ax_0^2l^2\le \pi^2}$.
Let
$\displaystyle{
u_0(t)=x_0\tanh\left(2\sqrt{a}x_0t\right).}
$
Then $u_0(t)$ is the solution to 
\begin{eqnarray}
u''(t)&=&2Q'(u(t))\qquad \mbox{for all}~t\in \RR,\label{example instanton equation}\\
\lim_{t\to-\infty}u(t)=-x_0,&~&\lim_{t\to\infty}u(t)=x_0
\label{boundary condition}
\end{eqnarray}
and
\begin{eqnarray}
I_{\infty,P}(u_0)&=&
\left(\frac{1}{4}\int_{-\infty}^{\infty}
u_0'(t)^2dt+\int_{-\infty}^{\infty}Q(u_0(t))dt\right)l,
\label{instanton energy1}\\
&=&d^{Ag}_{1dim}(-x_0,x_0)l\label{instanton energy2}\\
&=&d^{Ag}_U(-h_0,h_0).\label{instanton energy3}
\end{eqnarray}
\end{pro}

The proposition above claims that
$u_0$ is the instanton for both operators:
$1$-dimensional Schr\"odinger operator
$\displaystyle{-\frac{d^2}{dx^2}+\la Q(\cdot/\sqrt{\la})}$
and $-L_A+V_{\la}$.

\begin{proof}[Proof of Proposition~$\ref{example of instanton}$]
~We consider a projection operator $P$ on $L^2(I)$ onto
the subset of constant functions.
Let $h\in H^1(I)$.
For simplicity, we write $c=Ph$ and $\tilde{h}=h-Ph$.
Then
\begin{eqnarray}
U(h)&=&\frac{1}{4}\int_I\tilde{h}'(x)^2dx+
2a(c^2-x_0^2)\int_I\tilde{h}(x)^2dx\nonumber\\
& &+a\int_I\left(2c\tilde{h}(x)+\tilde{h}(x)^2\right)^2dx
+a\int_I\left(c^2-x_0^2\right)^2dx.
\end{eqnarray}
By the Poincare inequality
\begin{equation}
\left(\frac{2\pi}{l}\right)^2\int_I\tilde{h}(x)^2dx\le \int_I\tilde{h}'(x)^2dx
\end{equation}
and the assumption on $a, x_0$,
we obtain
\begin{equation}
U(h)\ge a\int_I(c^2-x_0^2)^2dx=U(Ph).
\end{equation}
Therefore for any $u\in H^1_{T,h_1,h_2}$
\begin{eqnarray}
\int_{-T}^T\sqrt{U(u(t,\cdot))}\|u(t,\cdot)\|_{L^2}dt
&\ge&\int_{-T}^T\sqrt{Q(Pu(t))}\|Pu(t)\|_{L^2}dt.
\end{eqnarray}
Since the set $\{Pu~|~u\in H^1_{T,-h_0,h_0}\}$
coincides with the set of all
paths $v=v(t)$ in $H^1\left((-T,T)\to\RR\right)$ with the constraint 
$v(-T)=-x_0, v(T)=x_0$ ,
we get
$d_U^{Ag}(-h_0,h_0)=l d^{Ag}_{1dim}(x_0,-x_0)$.
It is an elementary calculation to check that $u_0$ satisfies 
(\ref{example instanton equation}) and (\ref{boundary condition}).
Finally, we prove the identity (\ref{instanton energy1}),
(\ref{instanton energy2}), (\ref{instanton energy3}).
Since $u_0'(t)=2\sqrt{Q(u_0(t))}$ holds,
we have
\begin{eqnarray}
\frac{1}{4}\int_{-\infty}^{\infty}u_0'(t)^2dt
+\int_{-\infty}^{\infty}Q(u_0(t))dt
&=&\int_{-\infty}^{\infty}\sqrt{Q(u_0(t))}u_0'(t)dt\nonumber\\
&=&\int_{-x_0}^{x_0}\sqrt{Q(x)}dx\nonumber\\
&=&\frac{4\sqrt{a}}{3}x_0^3
=d^{Ag}_{1dim}(-x_0,x_0)\nonumber
\end{eqnarray}
which completes the proof.
\end{proof}

\bigskip

\noindent
{\bf Acknowledgments}
~
The author would like to thank Professors Christian G\'erard,
Tetsuya Hattori, Keiichi Ito, Izumi Ojima for 
their comments on the first version of this paper.
Also the author is grateful to referees for their
valuable comments which improved the paper.

\end{document}